%% file: curve_approximation_JSC.tex
\documentclass[preprint,3p,authoryear]{elsarticle}
\usepackage{amsmath,amssymb,amsfonts,amsthm}
\usepackage[utf8]{inputenc}
\usepackage{mathtools}
\usepackage{lipsum}  
\usepackage{algorithm}
\usepackage{todonotes}
\usepackage{algpseudocode}
\usepackage{pdfpages}
\usepackage{graphicx}
\usepackage{xcolor}
\usepackage{tikz}
\usetikzlibrary{external}
\usetikzlibrary{calc}

\title{Certified simultaneous isotopic approximation of curves via subdivision\tnoteref{t1}}
\tnotetext[t1]{This work was supported by NSF grant DMS-1913119 and Simons Foundation collaboration grant \#964285 awarded to Burr.}

\newtheorem{definition}{Definition}
\newtheorem{example}[definition]{Example}
\newtheorem{lemma}[definition]{Lemma}
\newtheorem{proposition}[definition]{Proposition}
\newtheorem{corollary}[definition]{Corollary}
\newtheorem{remark}[definition]{Remark}
\newtheorem{theorem}[definition]{Theorem}

\newcommand{\ZZ}{{\mathbb{Z}}}
\newcommand{\RR}{{\mathbb{R}}}
\newcommand{\CC}{{\mathbb{C}}}

\newcommand{\cA}{{\mathcal{A}}}

\newcommand{\cN}{{\mathcal{N}}}
\newcommand{\cS}{{\mathcal{S}}}
\newcommand{\cV}{{\mathcal{V}}}

\newcommand{\del}{\nabla}

\newcommand{\true}{\textnormal{\textsc{True}}}
\newcommand{\false}{\textnormal{\textsc{False}}}

\newcommand{\Diam}{\operatorname{Diam}}
\newcommand{\dist}{\operatorname{dist}}

\DeclareMathOperator{\area}{Area}

\begin{document}
	
	\author[1]{Michael Burr}
	\ead{burr2@clemson.edu}
	\author[1]{Michael Byrd}
	\ead{mbyrd6@clemson.edu}
	\affiliation[1]{organization={Clemson University},
		addressline={220 Parkway Drive},
		postcode={29634},
		city={Clemson, SC},
		country={USA}}
	
	\begin{abstract}
		We present a certified algorithm based on subdivision for computing an isotopic approximation to any number of curves in the plane.  Our algorithm is based on the certified curve approximation algorithm of Plantinga and Vegter.  The main challenge in this algorithm is to correctly and efficiently identify and isolate all intersections between the curves.  To overcome this challenge, we introduce a new and simple test that guarantees the global correctness of our output.  A main step in our algorithm for approximating any number of curves is to correctly approximate a pair of curves.  In addition to developing the details of this special case, we provide complexity analyses for both the number of steps and the bit-complexity of this algorithm using both worst-case bounds as well as those based on continuous amortization. 
	\end{abstract}
	
	
	\maketitle
	
	\section{Introduction}
	In \citet{PV:2004,PV:2007}, the authors introduced a pair of algorithms to construct piecewise-linear approximations to smooth and bounded real curves and surfaces in two and three dimensions, respectively.  They prove that the outputs of their algorithms are both topologically correct and can be refined to be as close to the underlying curve as desired.  Their algorithms are particularly interesting as they are symbolic-numeric algorithms based on subdivision whose predicates are simple, efficient, and easy to implement.  On singular input, however, the Plantinga and Vegter algorithms do not terminate as both of their predicates fail on regions containing singular points.  The current paper presents an algorithm in the spirit of the original Plantinga and Vegter algorithm for correctly approximating a union of smooth curves in the plane with transverse crossings.
	
	We begin by posing the main problem addressed in this paper.  Suppose that $f_1,\dots,f_n\in\ZZ[x,y]$ define $n$ smooth curves $\cV(f_1),\dots,\cV(f_n)$ in the real plane.  In addition, we assume that, pairwise, these curves intersect transversely and that no three of these curves intersect simultaneously.  We construct $n$ piecewise-linear approximations $\cA(f_1),\dots,\cA(f_n)$ such that each $\cA(f_i)$ approximates $\cV(f_i)$ and $(\cA(f_1),\dots,\cA(f_n))$ is topologically equivalent to $(\cV(f_1),\dots,\cV(f_n)).$  In our setting, topologically correct means that there is an ambient isotopy that deforms space while simultaneously taking each $\cV(f_i)$ to $\cA(f_i)$.  In particular, the intersection points of the curves are carried to the crossings of the approximations.  As in the original algorithm of Plantinga and Vegter, the approximations and curves can be made to be as close as desired in Hausdorff distance.
	
	The main challenge in extending the Plantinga and Vegter algorithm is that while the original algorithm can compute topologically correct approximations to any single curve, there is no guarantee that the corresponding ambient isotopies are compatible when applied to multiple curves, see Figure \ref{fig:missing_intersection}.  For example, it is possible for $\cA(f_i)\cap\cA(f_j)$ to have extraneous or missing intersections when compared to $\cV(f_i)\cap\cV(f_j)$.  Even if all of the intersections between individual pairs of curves are computed correctly, the order in which the intersections $\cA(f_j)\cap\cA(f_i)$ and $\cA(f_l)\cap\cA(f_i)$ occur along $\cA(f_i)$ might differ from the order in which the intersections $\cV(f_j)\cap\cV(f_i)$ and $\cV(f_l)\cap\cV(f_i)$ occur along $\cV(f_i)$.
	
	\begin{figure}[htb]
		\centering
		\begin{tabular}{cc}
			\input{figures_source/missing_intersection_before.tex}
			&
			\input{figures_source/missing_intersection_after.tex}\\
			(a)&(b)
		\end{tabular}
		\caption{Two approximations (thick lines) of a pair of curves (thin lines).  The approximations and curves are paired by color and line style.  A naive approach (a) misses a pair of intersections due to excursions while our approach (b) correctly approximates the curves and their intersections.}
		\label{fig:missing_intersection}
	\end{figure}
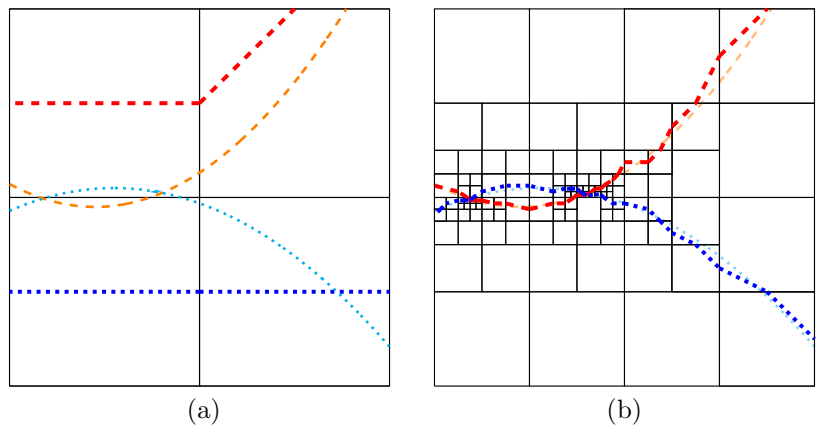
	
	\subsection{Previous work}
	
	An extension of the Plantinga and Vegter algorithm was introduced in \citet{BCGY:2012} to handle unbounded and singular input.  This approach can solve the problem currently under consideration by computing an approximation to the variety $\cV(fg)$, but this approach relies on separation bounds between singular points.  These separation bounds are typically so pessimistic that it is questionable whether this algorithm is practical.  On the other hand, the algorithm introduced in \citet{LSVY:2014} also studies the problem considered here, but their algorithm uses more restrictive tests than what we propose, and they may require a significant number of subdivisions to characterize the local behavior of curves within a region.  For instance, their algorithm has more topological requirements on boxes that contain intersections of curves than what is required in our approach.  The algorithm presented in \citet{CWZ:2023} solves the related problem of isolating the points of $\cV(f_1)\cap\cV(f_2)$ without approximating the curves themselves.  Their approach is based on studying the local behavior of the curves of $\cV(f_1)$ and $\cV(f_2)$ and also uses more restrictive tests than what we propose.  Our correctness statements are slightly weaker than the corresponding correctness statements in \citet{LSVY:2014} and \citet{CWZ:2023}, as we do not identify which boxes contain crossings of $\cV(f_1)$ and $\cV(f_2)$.  Our correctness statements, however, are still quite strong and more in-line with the statements appearing in the original work of Plantinga and Vegter.
	
	A preliminary version of the work in this paper appeared in \citet{BurrByrd:2023}.  In that paper, we provided the details and correctness of an algorithm to correctly approximate a pair of smooth curves in the plane, see Section \ref{sec:subdivisionstep}.  In the current paper, we add several complexity analyses of that algorithm, including both counting the number of regions produced by the algorithm as well as the bit-complexity of the algorithm.  Our complexity analyses use both nonamortized bounds as well as amortized bounds based on continuous amortization.  In addition, we extend our algorithm to correctly approximate any number of curves in the plane.
	
	\subsection{Outline}
	In Section \ref{sec:background}, we recall the details of the Plantinga and Vegter algorithm for curve approximation.  In Section \ref{sec:subdivisionstep}, we provide the details of the subdivision step for our algorithm to approximate a pair of curves in the plane.  This special case of approximating a pair of curves is a key step of our main algorithm.  In Section \ref{sec:complexity}, we analyze the complexity of Algorithm \ref{sec:subdivisionstep} for approximating a pair of curves in the plane.  In Section \ref{sec:arbitrarycurves}, we apply the algorithm from Section \ref{sec:subdivisionstep} to develop our main algorithm, which correctly approximates any number of curves.  Finally, we present a gallery of examples in Section \ref{sec:examples}.
	
	\section{Background}\label{sec:background}
	
	The Plantinga and Vegter algorithm for curve approximation is an adaptive subdivision-based algorithm based on the marching cube algorithm \citep{LC:1987}.  The input to the algorithm is a bivariate polynomial $f\in\ZZ[x,y]$ and an input region $R=[a,b]\times[c,d]$ in $\RR^2$ such that $a,b,c,d\in\ZZ$.  The output is a piecewise linear approximation $\cA(f)$, which is guaranteed to be ambient isotopic to the real variety $\cV(f)$.  The algorithm consists of three main steps: 
	(1) the subdivision step, (2) the balancing step, and (3) the approximation step.
	
	In the subdivision step, square regions are considered by the algorithm and either accepted and further processed or rejected and split into four equal-sized subboxes for further consideration.  The acceptance and rejection is based on two predicates called $C_0$ and $C_1$.  These predicates take, as input, a square region $B\subseteq R$ and produce either \true\ or \false.  A square is accepted when either of the predicates returns \true\ and is rejected when both predicates return \false.  The predicates have the following properties: When $C_0(B)=\true$, we conclude that $\cV(f)$ does not intersect $B$.  On the other hand, when $C_1(B)=\true$, there do not exist any pair of points $((x_1,y_1),(x_2,y_2))\in B\times B$ such that $\nabla f(x_1,y_1)$ and $\nabla f(x_2,y_2)$ are perpendicular.
	
	The resulting subdivision consists of a union of boxes, see Figure \ref{fig:missing_intersection}, and we use the word {\em side} to denote one of the four sides of a box (as a quadrilateral) in the subdivision.  We define an {\em edge} of the subdivision to be a side of some box that is not composed of a union of sides of smaller neighboring boxes.
	
	In the balancing step, additional subdivisions are performed until the side length of neighboring boxes differ in length by at most a factor of two.  Therefore, if $B$ is a box of the final subdivision, then the edges of $B$ are either sides of $B$ or half-sides of $B$.  The edges of $B$ are half-sides when $B$'s neighbor in that direction is smaller than $B$.  Finally, in the approximation step $f$ is evaluated on each vertex of every box $B$ which satisfies $C_1(B)$ but not $C_0(B)$.  For each edge in the subdivision where $f$ changes signs, the algorithm adds a vertex on this edge.  Finally, in each box, the vertices are connected in such a way so that the segments do not intersect, and, if there are four vertices on the sides of a box, the two vertices on the same side of the box are not connected.
	
	\subsection{Correctness and complexity}
	
	In the original presentation \citep{PV:2004,PV:2007}, the correctness of the curve approximation algorithm was only proved for smooth and bounded curves.  In \citet{BCGY:2012}, by slightly weakening the correctness statement, the algorithm was extended to unbounded curves.  In addition, in \citet{BCGY:2012}, the algorithm and correctness statement were extended to nonsingular curves, but the practicality of this approach remains in question.  In \citet{LC:2011}, the authors extend the algorithm to non-square regions within the subdivision.
	
	In each of these algorithms, an important feature of their correctness statements is that they guarantee global correctness, not local correctness.  For instance, the correctness statement does not guarantee that the approximation and the variety are isotopic when restricted to a box of the subdivision, only that they are isotopic in the input region $R$ (or a set containing the input region).  The main difference between the variety and its approximation are {\em excursions}. An excursion occurs when the variety briefly enters a neighboring box, but this behavior does not appear in the approximation, see Figure~\ref{fig:missing_intersection} and Definition \ref{def:excursion}.  Instead, the ambient isotopy stretches space so that the approximation is moved into the neighboring box.  
	
	This global correctness without local correctness is a key feature of this family of algorithms.  It often leads to many fewer boxes since these algorithms do not need to resolve the behavior of small excursions.  This key feature makes the problem of approximating a pair of curves given by $f_1,f_2\in\ZZ[x,y]$ more challenging, however, since an excursion may involve an intersection between the varieties $\cV(f_1)$ and $\cV(f_2)$, but the ambient isotopies separate the curves and remove the intersection from the approximations $\cA(f_1)$ and $\cA(f_2)$, see Figure \ref{fig:missing_intersection}(a).
	
	The complexity of the Plantinga and Vegter algorithm was first studied in \citet{BGT:2020} using continuous amortization, see, for example, \citet{BK:2012} and \citet{B:2016}.  The authors found both adaptive and worst-case complexity bounds for the number of regions formed by subdivision as well as the bit-complexity of the algorithm.  In addition, they found examples which exhibited the worst-case exponential complexity.  In \citet{CET:2022}, a smoothed-analysis-based approach was used to show that the average complexity of the algorithm is polynomial.  In \citet{TT:2020}, a condition-number-based approach also showed that the average complexity of the algorithm is polynomial for other classes of random polynomials including some sparse families.
	
	\subsection{Predicate details}
	The two predicates in the Plantinga and Vegter algorithm are typically implemented using interval arithmetic, see, for example, \citep{MKC:2009} for more details.  Interval arithmetic extends the standard arithmetic operations to intervals.  For instance, \begin{align*}
		[a,b]+[c,d]&\vcentcolon=[a+c,b+d],\\
		[a,b]-[c,d]&\vcentcolon=[a-d,b-c], \text{ and}\\
		[a,b][c,d]&\vcentcolon=[\min\{ac,ad,bc,bd\},\max\{ac,ad,bc,bd\}].
	\end{align*}
	These interval operations can be extended to the evaluation of functions, and we use the symbol $\square$ to denote any such extension.  In particular, for a polynomial $f\in\ZZ[x,y]$ and a region $B$, $\square f(B)$ is an interval containing the image $f(B)\vcentcolon=\{f(x,y):(x,y)\in B\}$.  The interval $\square f(B)$ is often larger than $f(B)$, but may be significantly easier to compute.
	
	The $C_0$ test is implemented as $C_0(B)=\true$ if and only if $0\not\in\square f(B)$.  Since $\square f(B)$ is an over-approximation to $f(B)$, the condition $0\not\in\square f(B)$ implies that $0\not\in f(B)$, so the variety $\cV(f)$ cannot intersect $B$.  The $C_1$ test is slightly more complicated, as $C_1(B)=\true$ if and only if $0\not\in\square \langle \nabla f,\nabla f\rangle(B\times B)$, where $\langle\cdot{,}\cdot\rangle$ denotes the standard inner product.  In this formulation, each of the factors of $B\times B$ is the argument to one $\nabla f$.  If $0\not\in\square \langle \nabla f,\nabla f\rangle(B\times B)$, then there cannot be a pair of points $\left((x_1,y_1),(x_2,y_2)\right)\in B\times B$ such that $\langle \nabla f(x_1,y_1),\nabla f(x_2,y_2)\rangle=0$, that is, the gradient vectors cannot be perpendicular.
	
	We call these algorithms symbolic-numeric algorithms for two reasons: The predicates perform exact computations using the coefficients of $f$, that is, not merely treating $f$ as a function.  The computations themselves are performed using arbitrary-precision floating point computations on dyadic points.  In other words, the evaluations are exact, but leverage the speed of floating point calculations.
	
	\subsection{Topological details}\label{sec:topdetails}
	We collect some key facts from \citet{PV:2004,PV:2007} and provide a description of the ambient isotopy in the Plantinga and Vegter algorithm.  These facts are used throughout our correctness proofs.  
	
	\begin{definition}\label{def:excursion}
		Let $B$ be a box of a subdivision and $f\in\ZZ[x,y]$.  An {\em excursion} of $\cV(f)$ is a component of $\cV(f)\cap B$ whose two endpoints are on the same edge of the subdivision, see Figure \ref{fig:missing_intersection}(a).
	\end{definition}
	We note that excursions do not appear in the piecewise-linear approximation $\cA(f)$ as they are deformed into neighboring boxes.
	
	In \citet{PV:2004,PV:2007}, the authors use several topological lemmas to show that the predicates $C_0$ and $C_1$ exert control over the behavior of $\cV(f)$ within a box, see Figure \ref{fig:lemma_schematics}(a).
	
	\begin{lemma}[\citet{PV:2004,PV:2007}]\label{lem:PVcorrectness}
		Suppose that $B$ is a box of a subdivision, and suppose that there are two segments $s_1$ and $s_2$ in $B$ such that 
		\begin{enumerate}
			\item the lines formed from extending $s_1$ and $s_2$ are perpendicular,
			\item the value of $f$ on both endpoints of $s_1$ is the same, and
			\item the value of $f$ on both endpoints of $s_2$ is the same.
		\end{enumerate}
		Then, $C_1(B)=\false$. 
	\end{lemma}
	
	\begin{figure}[h]
		\centering
		\input{figures_source/lemma_schematics}
		\caption{Illustrations of (a) Lemma \ref{lem:PVcorrectness} and (b) Lemma \ref{lem:multicrossings}.  In (a), the existence of two perpendicular segments meeting the curve twice implies that $C_1$ fails in this box.  In (b), the existence of two parallel segments each meeting a curve twice implies that $C_1^\times$ fails in this box.}
		\label{fig:lemma_schematics}
	\end{figure}
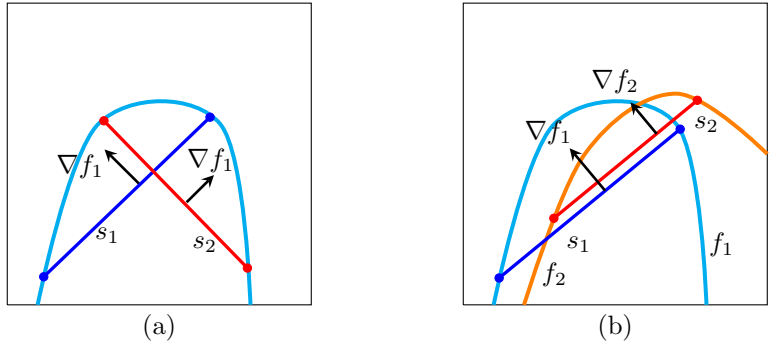
	
	This lemma follows from applying the mean value theorem on segments $s_1$ and $s_2$ to show that each segment contains a point where the gradient at that point is perpendicular to the segment (see also Lemma \ref{lem:multicrossings}).  Lemma \ref{lem:PVcorrectness} leads to several corollaries, three of which we list here:
	
	\begin{corollary}[\citet{PV:2004,PV:2007}]\label{cor:semicircle}
		Suppose that $B$ is a box of a subdivision such that $C_1(B)=\true$.  In addition, let $\gamma_f$ be a component of $\cV(f)\cap B$ which is an excursion on edge $e$ of $B$.  Then, $\gamma_f$ is contained entirely within the semicircle in $B$ whose diameter is $e$.
	\end{corollary}
	
	\begin{corollary}[\citet{PV:2004,PV:2007}]\label{cor:sides}
		Suppose that $B$ is a box of a subdivision.
		If $\cV(f)$ intersects two adjacent sides of $B$ and, on each of these sides, intersects twice, then $C_1(B)=\false.$
	\end{corollary}
	
	\begin{corollary}[\citet{PV:2004,PV:2007}]\label{cor:components}
		Suppose that $B$ is a box of a subdivision such that $C_1(B)=\true$.  There is at most one component of $\cV(f)\cap B$ that extends between opposite sides of $B$.
	\end{corollary}
	
	Finally, we briefly describe how the ambient isotopy deforms the variety $\cV(f)$ to $\cA(f)$ as a two-step procedure.  The first step of the ambient isotopy is to remove all excursions by deforming space so that the excursions are moved into neighboring boxes.  Briefly, we let $\widetilde{\cV}(f)$ be the result of applying the first step of the ambient isotopy to $\cV(f)$.  For every box $B$ of the subdivision, $\widetilde{\cV}(f)\cap B$ and $\cA(f)\cap B$ are ambient isotopic {\em within} $B$.  In particular, this means that they have the same number of components within $B$ and the components intersect the same edges of $B$.  The second step of the ambient isotopy simultaneously deforms $\widetilde{\cV}(f)\cap B$ to $\cA(f)\cap B$ within each box $B$ by sliding the points on the boundary of $B$ to their appropriate places and straightening the curves within each box $B$.  By careful consideration of these steps, we observe that the ambient isotopies derived from the Plantinga and Vegter algorithm move points at most one box away as the only points that move between boxes are those near excursions, but, by Corollary \ref{cor:semicircle}, these points are never further than one box away.
	
	\begin{definition}
		Let $S$ be a union of boxes from a subdivision and $\gamma_f:[0,1]\rightarrow\cV(f)\cap S$ a curve in the variety of $f$.  The \emph{extension of $\gamma_f$ without excursions} is denoted by $\overline{\gamma}_f$, and it is the component of $\cV(f)\cap S$ containing $\gamma_f$.  The \emph{extension of $\gamma_f$ with excursions} is the curve $\widetilde{\gamma}_f$, which is formed by following $\gamma_f$ forward and backwards until either 
		(1) the curve becomes a closed loop or (2) the curve reaches the first and last intersections of the curve with the boundary of $S$ before passing through a box not in $S$.  See Figure \ref{fig:extensions} for details.
	\end{definition}
	
	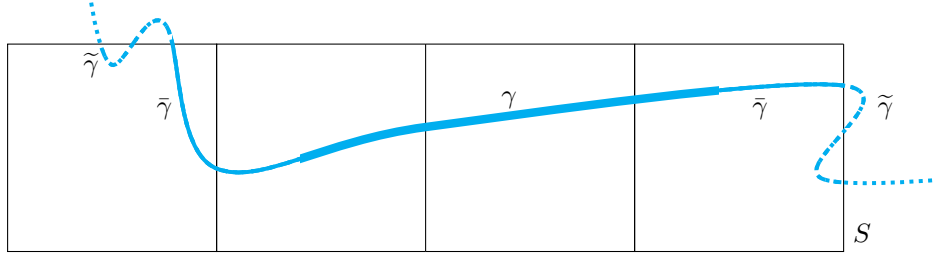
\begin{figure}
		\centering
		\begin{tikzpicture}[scale=2.75]
			\hspace*{4em}\begin{scope}[shift = {(3,0)}]
				\draw[black] (0,0) -- (1,0) -- (1,1) -- (0, 1);
				\draw[black] (-2,0) -- (-1,0) -- (-1,1) -- (-2, 1) -- cycle;
				\draw[black]  (0,0) -- (0,1);
				\draw[black] (1, 1) -- (1,0);
				\draw[color = black]   (1, 0) -- (2, 0) -- (2,1);
				\draw[color = black] (2, 1) -- (1, 1)   (0, 1) -- (-1, 1) -- (-1, 0) -- (0,0);
				\node[above right] at (2,0) {$S$};
				
				\begin{scope}
					\clip(-.6, 0) rectangle (1.4, 1.2);
					\draw[color=cyan, line width = 3.2] plot[smooth, tension = .6] coordinates{(-1.6, 1.2)(-1.5, .9) (-1.25, 1.1) (-1, .4) (0, 0.6) (2, 0.8) (1.9, 0.35) (3, .4)}; 
				\end{scope}
				
				\begin{scope}
					\clip(-1.21, 0) rectangle (1, 1.2);
					\draw[color=cyan, line width = 1.5] plot[smooth, tension = .6] coordinates{(-1.6, 1.2)(-1.5, .9) (-1.25, 1.1) (-1, .4) (0, 0.6) (2, 0.8) (1.9, 0.35) (3, .4)}; 
				\end{scope}
				
				\begin{scope}
					\clip(0, .65) rectangle (2, 1);
					\draw[color=cyan, line width = 1.5] plot[smooth, tension = .6] coordinates{(-1.6, 1.2)(-1.5, .9) (-1.25, 1.1) (-1, .4) (0, 0.6) (2, 0.8) (1.9, 0.35) (3, .4)}; 
				\end{scope}
				\begin{scope}
					\clip(-1.55, 0) rectangle (2, 1.2);
					\draw[color=cyan, line width = 1.5, dashed] plot[smooth, tension = .6] coordinates{(-1.6, 1.2)(-1.5, .9) (-1.25, 1.1) (-1, .4) (0, 0.6) (2, 0.8) (1.9, 0.35) (3, .4)};
				\end{scope}		
				\begin{scope}
					\clip(2, .6) rectangle (3, 1);
					\draw[color=cyan, line width = 1.5, dashed] plot[smooth, tension = .6] coordinates{(-1.6, 1.2)(-1.5, .9) (-1.25, 1.1) (-1, .4) (0, 0.6) (2, 0.8) (1.9, 0.35) (3, .4)};
				\end{scope}
				\begin{scope}
					\clip(-2, 0) rectangle (2.5, 1.2);
					\draw[color=cyan, line width = 1.5, dotted] plot[smooth, tension = .6] coordinates{(-1.6, 1.2)(-1.5, .9) (-1.25, 1.1) (-1, .4) (0, 0.6) (2, 0.8) (1.9, 0.35) (3, .4)};
				\end{scope}
				\node[label=:$\gamma$] at (.4,0.6) {};
				\node[label=:$\bar{\gamma}$] at (1.6,0.55) {};
				\node[label=:$\bar{\gamma}$] at (-1.25,0.55) {};
				\node[label=:$\widetilde{\gamma}$] at (2.2,0.55) {};
				\node[label=:$\widetilde{\gamma}$] at (-1.6,0.75) {};
			\end{scope}
		\end{tikzpicture}
		\caption{An illustration of the two types of extensions of a curve $\gamma$, denoted by the thickened curve, with a region $S$.  The first extension is $\overline{\gamma}$, which is denoted by the thin curve.  $\overline{\gamma}$ extends $\gamma$ to the first time that the curve containing $\gamma$ leaves the region of interest.  The second extension is $\widetilde{\gamma}$, which is denoted by the dashes.  $\widetilde{\gamma}$ includes all excursions until the curve containing $\gamma$ leaves the region of interest and does not immediately return.  The part of the curve denoted by dots is beyond the extensions $\overline{\gamma}$ and $\widetilde{\gamma}$.}
		\label{fig:extensions}
	\end{figure}
	
	The extension $\widetilde{\gamma}_f$ can be constructed by iteratively adding excursions and curve components at the ends of the path until the curve leaves $S$ and passes through other boxes of the subdivision.  The ambient isotopy described above directly leads to the following correctness statements:
	
	\begin{lemma}[{\citet{PV:2004,PV:2007}}]\label{lem:excursions}
		Let $B$ be a box from the output of the Plantinga and Vegter algorithm and suppose that $\cV(f)$ intersects an edge of $B$ twice.  Then there is a path in $\cV(f)$ connecting these two points which is a sequence of excursions.
	\end{lemma}
	
	\begin{lemma}[{\citet{PV:2004,PV:2007}}]\label{lem:approximation}
		Let $S$ be a union of boxes from the output of the Plantinga and Vegter algorithm and $\gamma_f:[0,1]\rightarrow\cV(f)\cap S$ a curve in the variety of $f$.  Suppose that $\gamma_f$ deforms to $\alpha_f\subseteq\cA(f)$ within $S$.  Let $\widetilde{\gamma}_f$ be the extension of this path with excursions.  Let $\overline{\alpha}_f$ be the component of $\cA(f)\cap S$ containing $\alpha_f$.  Either $\widetilde{\gamma}_f$ and $\overline{\alpha}_f$ are topological circles within $S$ or the endpoints of $\widetilde{\gamma}_f$ and $\overline{\alpha}_f$ are on the same edges of the subdivision.
	\end{lemma}
	
	\section{Approximating a pair of curves}\label{sec:subdivisionstep}
	In our algorithm in Section \ref{sec:arbitrarycurves} for approximating any number of curves, a main step in the procedure is to compute an approximation of a pair of curves.  In this section, we provide an algorithm for this special case.  Of the three main steps of the Plantinga and Vegter approximation algorithm, the most significant changes occur in the subdivision step, and we focus our attention on this step.  Much of the material in this section was initially presented in \citet{BurrByrd:2023}.
	
	Given $f_1,f_2\in\ZZ[x,y]$, a first attempt to approximate $(\cV(f_1),\cV(f_2))$ may be to simultaneously run the standard Plantinga and Vegter algorithm on $f_1$ and $f_2$ using a common refinement of the region $R$.  This approach leads to three different types of potential errors in the approximations:
	
	\begin{enumerate}
		\item Missing intersections: intersections of $\cV(f_1)$ and $\cV(f_2)$ which do not correspond to intersections of $\cA(f_1)$ and $\cA(f_2)$, see Figure \ref{fig:missing_intersection}.
		\item Extra intersections: intersections of $\cA(f_1)$ and $\cA(f_2)$ which do not correspond to intersections of $\cV(f_1)$ and $\cV(f_2)$, see Figure \ref{fig:tangential_intersection_nonintersect}(a).
		\item Shared edges: some of the edges of the approximations $\cA(f_1)$ and $\cA(f_2)$ may be shared between the two approximations, see Figure \ref{fig:tangential_intersection_nonintersect}(b).
	\end{enumerate}
	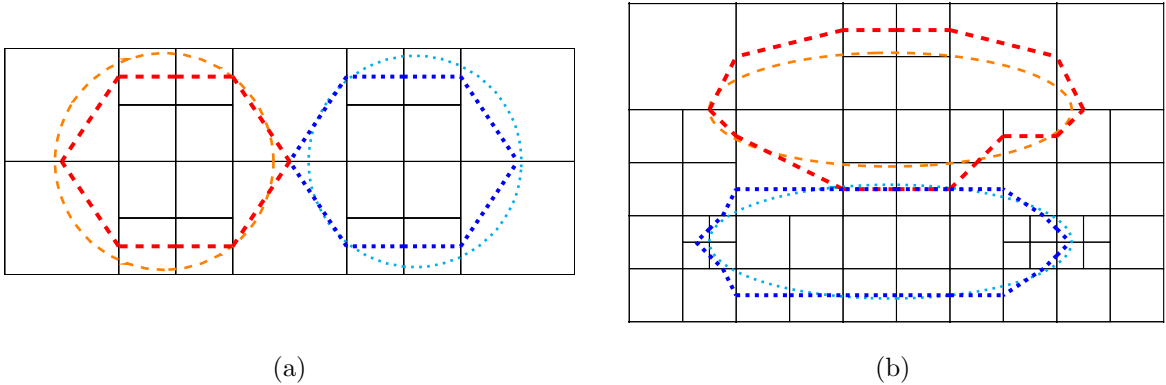
\begin{figure}[h]
		\centering
		\input{figures_source/two_ellipse_tangential_nonintersect_before.tex}
		\caption{Two approximations (thick lines) of a pair of curves (thin lines).  The approximations and curves are paired by color and line style.  The naive approach to approximating the pairs of curves include (a) an extra intersection or (b) several shared edges of the approximation.}
		\label{fig:tangential_intersection_nonintersect}
	\end{figure}
	
	Our main tool to avoid all three of these errors is a new predicate, which we call $C_1^\times$.  On a square $B$, if $C_1^\times(B)=\true$, then there do not exist any pair of points $((x_1,y_1),(x_2,y_2))\in B\times B$ such that $\nabla f_1(x_1,y_1)$ and $\nabla f_2(x_2,y_2)$ are parallel.  In the plane, we may implement this test using the cross product, that is, $C_1^\times(B)=\true$ if and only if $0\not\in\square (\nabla f_1\times\nabla f_2)(B\times B)$, where each of the factors of $B\times B$ is an argument to one of the gradients.  As an initial illustration of the utility of this new $C_1^\times$ predicate, we provide the following motivating result, see Figure \ref{fig:lemma_schematics}(b).
	\begin{lemma}\label{lem:multicrossings}
		Let $B$ be a rectangle where $s_1$ and $s_2$ are parallel line segments in $B$.  Suppose that $f_1$ attains the same value on both endpoints of $s_1$ and $f_2$ attains the same value on both endpoints of $s_2$.  Then $C_1^\times(B)=\false$.
	\end{lemma}
	\begin{proof}
		Suppose that $(x_1,y_1),(x_2,y_2)\in s_1$ such that $f_1(x_1,y_1)=f_1(x_2,y_2)$.  By applying Rolle's theorem to 
		$f_1(tx_1+(1-t)x_2,ty_1+(1-t)y_2)$, there is some $t_1\in (0,1)$ such that 
		$$\frac{d}{dt}f_1(tx_1+(1-t)x_2,ty_1+(1-t)y_2)\Big|_{t=t_1}=0.$$
		This equation, however, can be rewritten as the following dot product:
		$$\nabla f_1(t_1x_1+(1-t_1)x_2,t_1y_1+(1-t_1)y_2)\cdot (x_2-x_1,y_2-y_1)=0.$$
		In other words, there is some point in $B$ where $\nabla f_1$ is perpendicular to $s_1$.  Repeating this argument for $f_2$ gives that there is some point in $B$ where $\nabla f_2$ is perpendicular to $s_2$.  Since $s_1$ and $s_2$ are parallel, the gradients of $f_1$ and $f_2$ are parallel for some pair of points in the box $B$, and $C_1^\times(B)=\false$.
	\end{proof}
	
	This lemma leads to the following special case when $s_1=s_2$.
	
	\begin{corollary}\label{cor:multiintersection}
		Let $B$ be a rectangle and assume that $\cV(f_1)$ and $\cV(f_2)$ intersect more than once in $B$.  Then $C_1^\times(B)=\false$. 
	\end{corollary}
	
	\begin{remark}
		The use of the $C_1^\times$ test explains some of the conditions that we require on our curves.  In particular, we note that if an intersection of $\cV(f_1)$ and $\cV(f_2)$ was not a transverse crossing, then the gradients would agree at the intersection and $C_1^\times$ would be \false.
	\end{remark}
	
	We now present the algorithm for the subdivision step of the pairwise curve approximation algorithm.  For simplicity, we focus on the case where the input region $R$ is a square and leave the details for the general rectangular case to Section \ref{sec:simul_approx}.  Since there are two polynomials $f_1,f_2\in\ZZ[x,y]$, we write $C_0^{f_1}$ and $C_1^{f_1}$ for the standard tests from the Plantinga and Vegter algorithm for the function $f_1$.  Similarly, we define $C_0^{f_2}$ and $C_1^{f_2}$ for $f_2$.  In addition, we need the notion of a neighborhood of a box:
	
	\begin{definition}
		Let $R=[a,b]\times[c,d]$ be a square region and $\cS$ a partition of $R$ into squares.  For any square $B\in\cS$, the {\em neighborhood} of $B$ in $\cS$ is denoted by $\cN(B)$ and consists of $B$ along with all of the other squares in $\cS$ that have a positive-length intersection with $B$, that is, squares that only meet $B$ at its corners are not in $\cN(B)$.  More generally, we define $\cN_1(B)\vcentcolon=\cN(B)$ and $\cN_i(B)$ to be the union of all the neighborhoods of boxes in $\cN_{i-1}(B)$. 
	\end{definition}
	In other words, $\cN_i(B)$ consists of all the boxes which are at most $i$ boxes away from $B$.  We write $C_1^\times(\cN_i(B))=\true$ to denote that $C_1^\times$ holds in the smallest rectangle containing $\cN_i(B)$.
	
	\algrenewcommand\algorithmicrequire{\textbf{Input}:}
	\algrenewcommand\algorithmicensure{\textbf{Output}:}
	
	\begin{algorithm}[htb]
		\caption{Subdivision step}
		\label{alg:subdivision}
		\begin{algorithmic}[1]
			\Require polynomials $f_1,f_2\in\ZZ[x,y]$, with $\cV(f_1)$ and $\cV(f_2)$ smooth curves such that $\cV(f_1)$ and $\cV(f_2)$ intersect transversely, and a square region $R$ with integral corners such that $\cV(f_1)$ and $\cV(f_2)$ do not intersect on the boundary of $R$.
			\Ensure a partition of $R$ for further processing
			\State {initialize queue $Q$ to contain $R$}
			\While {$Q$ is not empty}
			\State {pop square $B$ from $Q$}
			\State {{\bf accept} $B$ if any of the following holds:}
			\Statex {\hspace*{{\algorithmicindent+\algorithmicindent}} $\bullet\;C_0^{f_1}(B)=C_0^{f_2}(B)=\true$}
			\Statex {\hspace*{{\algorithmicindent+\algorithmicindent}} $\bullet\;C_0^{f_1}(B)=C_1^{f_2}(B)=\true$}
			\Statex {\hspace*{{\algorithmicindent+\algorithmicindent}} $\bullet\;C_1^{f_1}(B)=C_0^{f_2}(B)=\true$}
			\Statex {\hspace*{{\algorithmicindent+\algorithmicindent}} $\bullet\;C_1^{f_1}(B)=C_1^{f_2}(B)=C_1^\times(\cN_2(B)\cap13B)=\true$}
			\State {{\bf reject} $B$ if none of the previous hold:}
			\State {\hspace*{{\algorithmicindent}}subdivide $B$ into four equal-sized boxes $B_1,\dots,B_4$}
			\State {\hspace*{{\algorithmicindent}}push $B_1,\dots,B_4$ into $Q$}
			\EndWhile
			\State\Return {accepted boxes}
		\end{algorithmic}
	\end{algorithm}
	
	We incorporate our new test into Algorithm \ref{alg:subdivision}, which, in turn, replaces the subdivision step of the standard Plantinga and Vegter algorithm.  In particular, the new subdivision step is similar to naively running the Plantinga and Vegter algorithm simultaneously on $f_1$ and $f_2$, except that when $C_1^{f_1}(B)$ and $C_1^{f_2}(B)$ both hold, we add the condition that $C_1^\times(\cN_2(B)\cap 13B)$ holds, where $13B$ denotes the box which is $13$-times larger than $B$, but with the same center.
	
	\begin{remark}
		After the balancing step of the Plantinga and Vegter algorithm, the neighbors of $B$ are at most twice the size of $B$ and their neighbors are at most four times the size of $B$.  This means that after balancing, $\cN_2(B)\subseteq 13B$.  To allow for the case where boxes are not balanced during the subdivision step, we use $\cN_2(B)\cap 13B$ in Step 4 of Algorithm \ref{alg:subdivision} as this region contains $\cN_2(B)$ of the balanced subdivision.
	\end{remark}
	
	Suppose that after this new subdivision step, the balancing and approximation steps of the standard Plantinga and Vegter algorithm are performed, resulting in approximations $\cA(f_1)$ and $\cA(f_2)$.  In the remainder of this section, we present the theoretical properties of the crossings of these approximations as well as the correctness statement that these approximations have the same topology as the underlying pair of curves.
	
	\subsection{Transversal crossing of approximations}
	For the approximations $\cA(f_1)$ and $\cA(f_2)$, we call a crossing {\em transversal} if the approximations cross within the interior of a box.  For a box whose neighbors are not further subdivided, this only happens when one approximation has an edge from the north side of a box to the south side, while the other approximation extends from the east side of the box to the west side.  We show that every transversal crossing of $\cA(f_1)$ and $\cA(f_2)$ corresponds to a unique crossing of the varieties $\cV(f_1)$ and $\cV(f_2)$.
	
	\begin{proposition}\label{prop:plus_sign}
		Let $f_1,f_2\in\ZZ[x,y]$, and suppose that $B$ is a box such that $C_1^{f_1}(B)=C_1^{f_2}(B)=C_1^\times(\cN(B))=\true.$  If $\cA(f_1)$ and $\cA(f_2)$ intersect transversely in $B$, then $\cV(f_1)$ and $\cV(f_2)$ intersect exactly once and transversely in $\cN(B)$.
	\end{proposition}
	\begin{proof}
		By Corollary \ref{cor:multiintersection}, the number of intersections of $\cV(f_1)$ and $\cV(f_2)$ is at most one in $\cN(B)$.  Moreover, any such intersection must be transversal by assumption.
		
		We let $\alpha_{f_1}$ and $\alpha_{f_2}$ be the two components of $\cA(f_1)\cap B$ and $\cA(f_2)\cap B$ which intersect in $B$.  Then, we let $\gamma_{f_1}$ be any component of $\cV(f_1)\cap B$ which deforms to a subset of $\alpha_{f_1}$ under the isotopy from the Plantinga and Vegter algorithm.  In addition, we let $\widetilde{\gamma}_{f_1}$ be the extension of $\gamma_{f_1}$ with excursions in $B$.  We define $\gamma_{f_2}$ and $\widetilde{\gamma}_{f_2}$ similarly.
		
		Let $B_D$ be the subset of $\cN(B)$ which consists of the union of $B$ with external semicircles on each edge of $B$.  Moreover, we let $\gamma_{f_1,D}$ be the component of $\cV(f_1)\cap B_D$ which contains $\gamma_{f_1}$.  By Corollary~\ref{cor:semicircle}, it follows that $\gamma_{f_1,D}\supseteq \widetilde{\gamma}_{f_1}$.  We define $\gamma_{f_2,D}$ similarly.  See Figure \ref{fig:plus_sign_schematic} for additional details.
		\begin{figure}
			\centering
			\input{figures_source/plus_sign_figure.tex}
			\caption{The neighborhood of the box $B$ with two crossing components.  The curves intersect in the neighborhood of $B$ while the approximations (not shown) intersect in $B$.}
			\label{fig:plus_sign_schematic}
		\end{figure}
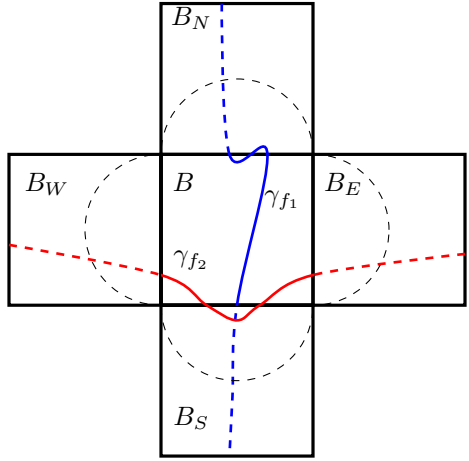
		
		We now restrict our attention to $f_1$ since the case for $f_2$ is similar.  By Lemma \ref{lem:approximation}, the endpoints of $\alpha_{f_1}$ and $\widetilde{\gamma}_{f_1}$ lie on the same edges of the subdivision, which we denote by $e_1$ and $e_2$.  Since the Plantinga and Vegter algorithm places a vertex on each of these edges, the number of intersections, counted with multiplicity, between $\cV(f_1)$ and each $e_i$ is odd.  Moreover, by Lemma \ref{lem:excursions}, all intersections between $\cV(f_1)$ and each $e_i$ are included in $\widetilde{\gamma}_{f_1}$.  Hence, $\gamma_{f_1,D}\setminus\widetilde{\gamma}_{f_1}$ consists of two curves contained in the semicircles attached to $e_1$ and $e_2$ and extending from the diameter to the circular boundary of the semicircles.
		
		Hence, $\gamma_{f_1,D}$ is a curve in $\cV(f_1)\cap B_D$ with endpoints on the two semicircles attached to $e_1$ and $e_2$.  By applying the same argument to $\gamma_{f_2,D}$, we see that $\gamma_{f_2,D}$ is a curve in $\cV(f_2)\cap B_D$ with endpoints on two other semicircles of $B_D$.  Since $\alpha_{f_1}$ separates the endpoints of $\alpha_{f_2}$ on the boundary of $B$, we conclude that the endpoints of $\gamma_{f_1,D}$ separate the endpoints of $\gamma_{f_2,D}$ on the boundary of $B_D$.  This implies the existence of a crossing in $B_D\subseteq\cN(B)$.
	\end{proof}
	
	By investigating the proof of Proposition \ref{prop:plus_sign} in detail, we find that each intersection of $\cA(f_1)$ and $\cA(f_2)$ corresponds to an intersection between $\gamma_{f_1,D}$ and $\gamma_{f_2,D}$ within $B_D$.  By construction, $\gamma_{f_1,D}$ and $\gamma_{f_2,D}$ are components of $V$ which cross $B$, perhaps with excursions.  We claim that these two crossing components uniquely determine $B$.  In particular, since $\cA(f_1)$ and $\cA(f_2)$ intersect in the interior of $B$, the ends of $\gamma_{f_1,D}$ and $\gamma_{f_2,D}$ are in different boxes of the subdivision from each other.  Hence, there is no other box of the subdivision that contains a portion of both $\gamma_{f_1,D}$ and $\gamma_{f_2,D}$ as crossing components.  Therefore, no two crossings of $\cA(f_1)$ and $\cA(f_2)$ can correspond to the same intersection of $\cV(f_1)$ and $\cV(f_2)$ because at least one of $\gamma_{f_1,D}$ and $\gamma_{f_2,D}$ changes when considering a different crossing.  Therefore, there is an injective map between transversal intersections of $\cA(f_1)$ and $\cA(f_2)$ and transversal intersections of $\cV(f_1)$ and $\cV(f_2)$.
	
	\subsection{Missing intersections}
	We begin by noting that excursions are the only reason that the approximations $\cA(f_1)$ and $\cA(f_2)$ can miss an intersection of $\cV(f_1)$ and $\cV(f_2)$.
	
	\begin{lemma}\label{lem:missingintersections}
		Suppose that $B$ is a box such that $C_1^{f_1}(B)=C_1^{f_2}(B)=C_1^\times(B)=\true$.  Suppose, in addition, that there are no excursions either entering or exiting $B$.  If the approximations $\cA(f_1)$ and $\cA(f_2)$ do not intersect in $B$, including on the boundary of $B$, then $\cV(f_1)$ and $\cV(f_2)$ do not intersect in $B$.
	\end{lemma}
	\begin{proof}
		By the properties of the first step of the Plantinga and Vegter algorithm, see Section \ref{sec:topdetails}, since there are no excursions, the approximations $\cA(f_1)\cap B$ and $\cA(f_2)\cap B$ are each ambient isotopic to $\cV(f_1)\cap B$ and $\cV(f_2)\cap B$ {\em within the box $B$}, respectively.  In particular, this implies that the number of components of $\cA(f_1)$ and $\cV(f_1)$ agree within $B$, and, similarly for $\cA(f_2)$ and $\cV(f_2)$.  We recall, however, that these isotopies are not necessarily the same.  
		
		Suppose that $\cV(f_1)$ and $\cV(f_2)$ intersect in $B$, but $\cA(f_1)$ and $\cA(f_2)$ do not intersect.  Let $\gamma_{f_1}$ be a component of $\cV(f_1)\cap B$ that intersects $\cV(f_2)\cap B$.  Let $\alpha_{f_1}$ be the corresponding subset of $\cA(f_1)$ to which $\gamma_{f_1}$ deforms under the ambient isotopy.  Let $e_1$ and $e_2$ be the edges of the subdivision that contain the endpoints of $\alpha_{f_1}$.  By Lemma \ref{lem:approximation}, $\gamma_{f_1}$ must also begin and end on these edges.  On the other hand, since $\cA(f_2)$ does not have vertices on these edges, it follows that on each edge $e_i$, the value of $f_2$ is the same at both endpoints of $e_i$.  Since there are no excursions, it follows that $\cV(f_2)$ does not intersect this edge as $\cV(f_2)$ would need to intersect this edge twice to maintain the sign properties of the endpoints and this would be an excursion by Lemma \ref{lem:excursions}.  
		
		Since $\cA(f_1)$ and $\cA(f_2)$ do not intersect, the signs of $f_2$ at both endpoints of $\alpha_{f_1}$ must be the same.  Since we showed that the sign of $f_2$ is constant on the edges containing the endpoints of $f_1$, the signs of $f_2$ on both endpoints of $\gamma_{f_1}$ must be the same.  This implies $\cV(f_2)$ must intersect $\gamma_{f_1}$ an even number of times as each intersection changes the sign of the restriction $f_2|_{\gamma_{f_1}}$.  Since $\cV(f_2)$ and $\gamma_{f_1}$ intersect at least once, they must intersect at least twice, but this is impossible by Corollary \ref{cor:multiintersection}.
	\end{proof}
	
	Lemma \ref{lem:missingintersections} implies that missing intersections must involve at least one excursion.  Our plan is to show that any missing intersection must induce a pair of intersections in the neighborhood $\cN_2(B)$, which is not possible since $C_1^\times(\cN_2(B))=\true$.
	
	\begin{proposition}\label{prop:nomissing}
		Suppose that $B$ is a box such that $C_1^{f_1}(B)=C_1^{f_2}(B)=C_1^\times(\cN_2(B))=\true$.  Suppose that $\cA(f_1)$ and $\cA(f_2)$ do not intersect in $\cN(B)$, including on the boundary of $\cN(B)$, then $\cV(f_1)$ and $\cV(f_2)$ do not intersect in $B$.  
	\end{proposition}
	\begin{proof}
		Suppose that $\cV(f_1)$ and $\cV(f_2)$ intersect in $B$. Let $\gamma_{f_1}$ be the component of $\cV(f_1)$ in $\cV(f_1)\cap B$ that includes this intersection.  Let $\widetilde{\gamma}_{f_1}$ be the extension of $\gamma_{f_1}$ into $\cN(B)$ including all excursions into $\cN_2(B)$, see Figure \ref{fig:missing_intersection_schematic}.
		
		Let $\alpha_{f_1}$ be the component of $\cA(f_1)\cap\cN(B)$ which contains the image of $\gamma_{f_1}$ under the ambient isotopy.  We note that the image of $\gamma_{f_1}$ is in $\cN(B)$ since the Plantinga and Vegter algorithm does not deform the curve further than one box away.  Let $e_1$ and $e_2$ be the edges of the subdivision that contain the endpoints of $\alpha_{f_1}$.  By Lemma~\ref{lem:approximation}, the endpoints of $\widetilde{\gamma}_{f_1}$ are also on the edges $e_1$ and $e_2$.  Let $p_1$ be the endpoint of $\widetilde{\gamma}_{f_1}$ on $e_1$ and $p_2$ be the endpoint of $\widetilde{\gamma}_{f_1}$ on $e_2$.
		
		Consider the restriction $f_2|_{\widetilde{\gamma}_{f_1}}$.  The signs of this function at the two endpoints must be opposite because each intersection between $\cV(f_2)$ and $\widetilde{\gamma}_{f_1}$ changes the sign of $f_2|_{\widetilde{\gamma}_{f_1}}$.  Having the same sign at both endpoints would require two intersections for the two sign changes, but this is impossible by Corollary \ref{cor:multiintersection} and that $C_1^\times(\cN_2(B))=\true$.
		
		Now, we prove that the sign of $f_2$ at $p_1$ agrees with its signs at the endpoints of $e_1$.  Since $\cA(f_2)$ does not intersect $\cA(f_1)$ and $\cA(f_1)$ has a vertex on $e_1$, this implies that the signs of $f_2$ on the endpoints of $e_1$ are the same.  Hence, any intersection of $\cV(f_2)$ with $e_1$ must be an excursion.  Suppose, for contradiction, that the sign of $f_2$ at $p_1$ does not match the sign of $f_2$ at the endpoints of $e_1$.  Then, there are an odd number of intersections from $\cV(f_2)\cap e_1$ on either side of $p_1$.  Therefore, there is at least one pair of points of $\cV(f_2)\cap e_1$ on either side of $p_1$ which are connected by an excursion.
		
		We show that this excursion implies that $\cV(f_1)$ and $\cV(f_2)$ intersect at least once more in $\cN_2(B)$, but such an intersection is not possible since $C_1^\times(\cN_2(B))=\true$ and Corollary \ref{cor:multiintersection}.  Since $\cV(f_2)$ has an excursion on edge $e_1$, $\cV(f_1)$ cannot have an excursion on this edge by Lemma \ref{lem:multicrossings}.  
		
		The excursion of $\cV(f_2)$ can be internal or external to $\cN(B)$, but the proofs are symmetric.  Therefore, we focus on the external case.  Since $\cV(f_1)$ does not have an excursion on $e_1$, there is a subset $\gamma_{f_1,D}$ of $\cV(f_1)$ which is a curve contained within the external semicircle and connects $p_1$ to the curved edge of the semicircle.  By Lemma \ref{lem:excursions}, the excursion of $\cV(f_2)$ separates the endpoints of $\gamma_{f_1,D}$.  Since $\cV(f_1)$ cannot have an excursion on $e_1$ and the excursion of $\cV(f_2)$ must remain within the semicircle on $e_1$ by Corollary \ref{cor:semicircle}, it must be that $\gamma_{f_1,D}$ intersects $\cV(f_2)$, implying that $\cV(f_1)$ and $\cV(f_2)$ intersect an additional time in $\cN_2(B)$, which is not possible.  For the argument in the internal case, we replace the external semicircle with the internal semicircle.
		
		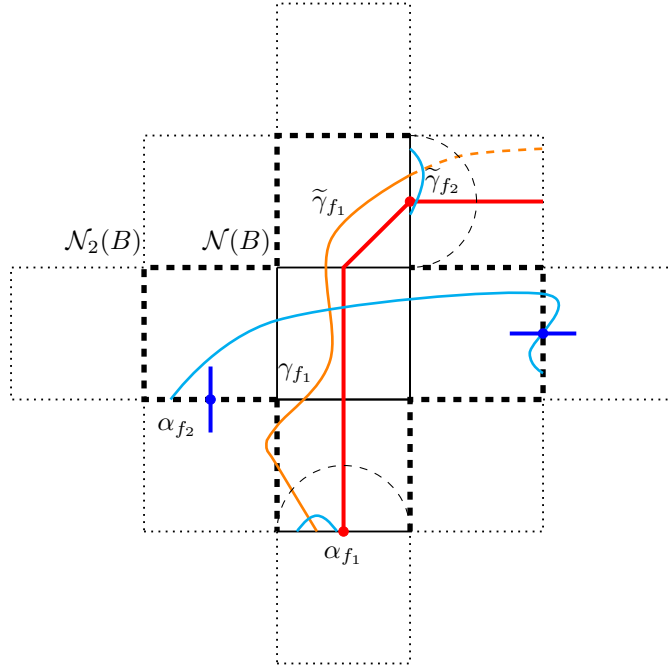
\begin{figure}
			\centering
			\input{figures_source/missing_intersection_schematic.tex}
			\caption{An illustration of the setup described in the proof of Proposition \ref{prop:nomissing}.  The two intersections between $\widetilde{\gamma}_{f_1}$ and excursions of $\widetilde{\gamma}_{f_2}$ are the impossible cases from the proof.  The thick dashed lines on the boundary of $\cN(B)$ illustrate the two components of $\partial\cN(B)$ split by the endpoints of $\widetilde{\gamma}_{f_1}$.}
			\label{fig:missing_intersection_schematic}
		\end{figure}
		
		We have shown that the sign of $f_2$ at $p_1$ agrees with the sign of $f_2$ at the endpoints of $e_1$.  Similarly, the sign of $p_2$ agrees with the sign of $f_2$ at the endpoints of $e_2$.  Moreover, since the signs of $f_2$ at $p_1$ and $p_2$ differ, the sign of $f_2$ at the endpoints of $e_1$ differ from the sign of $f_2$ at the endpoints of $e_2$.  We also observe that removing edges $e_1$ and $e_2$ from the boundary $\partial\cN(B)$ splits the boundary into two components.  Since the signs of $f_2$ are different at the endpoints of the two components, there must be an odd number of vertices of $\cA(f_2)$ on each of these two components.  However, since $\cA(f_2)$ forms a perfect matching on its vertices in $\cN(B)$, one of the edges of $\cA(f_2)$ must intersect $\alpha_{f_1}$, but this is not possible.
	\end{proof}
	
	Therefore, by Proposition \ref{prop:plus_sign}, every transversal intersection of the approximations $\cA(f_1)$ and $\cA(f_2)$ corresponds to a unique intersection of the varieties $\cV(f_1)$ and $\cV(f_2)$.  On the other hand, Proposition~\ref{prop:nomissing} implies that every intersection of the varieties $\cV(f_1)$ and $\cV(f_2)$ corresponds to an intersection of the approximations $\cA(f_1)$ and $\cA(f_2)$.  However, this intersection does not need to be transversal, see Figure \ref{fig:two_ellipse_tangential_intersect_before}.
	
	\begin{figure}[htb]
		\centering
		\input{figures_source/two_ellipse_tangential_intersect_before.tex}
		\caption{Approximations from the standard Plantinga and Vegter algorithm which share a segment while the curves intersect in two places.  The approximations and the curves are paired by color and line style.}
		\label{fig:two_ellipse_tangential_intersect_before}
	\end{figure}
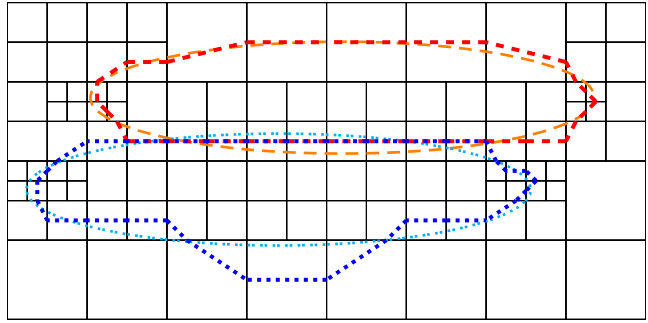
	
	\subsection{Shared edges in approximations}
	For the approximations $\cA(f_1)$ and $\cA(f_2)$, we call an intersection {\em non-transversal} if the approximations meet on the boundary of a box or the approximations coincide along shared segments.  
	
	\begin{definition}\label{def:snake}
		A contiguous sequence of boxes $S$ with shared segments or vertices of $\cA(f_1)$ and $\cA(f_2)$ is called a {\em snake} provided that the union of boxes in $S$ is topologically a disk, see Figure \ref{fig:snake_schematic}.  The boxes where the approximations separate are called the {\em heads} of the snake.  We call the non-head boxes the {\em interior} of the snake.  A {\em neighborhood} $\cN(S)$ of a snake is the union of all the neighborhoods of boxes in the snake along with the neighborhoods of the heads of the snake.  The neighborhood $\cN_i(S)$ is defined as $\cN_i(S):=\cN(\cN_{i-1}(S))$.
	\end{definition}
	\begin{figure}[hbt]
		\centering
		\input{figures_source/snake_schematic.tex}
		\caption{Example of a snake where the approximations share segments.  The shared segments are illustrated by a thicker line in purple.}
		\label{fig:snake_schematic}
	\end{figure}
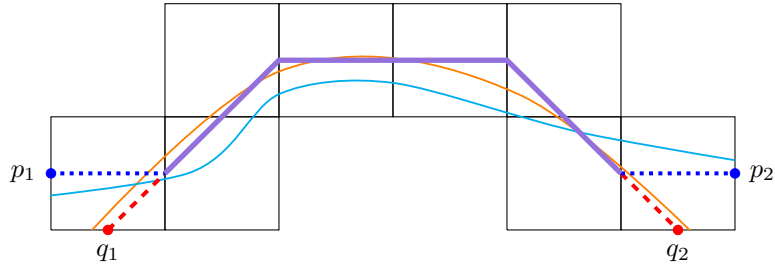
	
	The definition of a snake and the properties of boxes accepted by the approximation algorithm lead to many strong properties about the topological behavior of snakes.
	
	\begin{lemma}\label{lem:onecomponent}
		Suppose that $S$ is a snake and for every box $B\in S$, $C_1^\times(B)$ and $C_1^{f_1}(B)$ hold.  Then, every box in the interior of the snake has exactly one component of the approximation.
	\end{lemma}
	\begin{proof}
		Suppose that $B$ is an interior box of the snake with two components of the approximation.  The two components cannot intersect the same side of $B$ because this would satisfy the conditions of Lemma \ref{lem:multicrossings}, which is not possible.  Therefore, all four sides of the box must have a point of the approximation.  This leads to an alternating sign pattern on the vertices of the box, which violates $C_1^{f_1}$ as shown in \citet{PV:2004,PV:2007}.
	\end{proof}
	
	\begin{lemma}\label{lem:snakeneighbors:samesize}
		Suppose that $S$ is a snake and for every box $B\in S$, $C_1^\times(\cN(B))$ holds.  Let $B_1$ and $B_2$ be in the interior of $S$ such that $B_1$ and $B_2$ share an edge.  If $B_1$ and $B_2$ are the same size, then the shared approximation passes through that shared edge.
	\end{lemma}
	\begin{proof}
		We first observe that if two sides of $B_1$ and $B_2$ are collinear, then the approximation cannot pass through both sides, as this would satisfy the conditions of Lemma \ref{lem:multicrossings}, which is not possible.  Wlog, we consider the case where $B_1$ is above $B_2$.  Up to symmetry, there are two cases to consider.  If the curve approximation passes from the west to east sides of $B_2$, see Figure \ref{fig:snakeneighbors}(a), then there are not enough sides of $B_1$ for the curve approximation to pass through.  So, the approximation of the snake does not pass through $B_1$, so $B_1$ is not part of the snake, which is a contradiction.
		
		Suppose now that the approximation passes through the south side of $B_2$, and, up to symmetry, the west side of $B_2$, see Figure \ref{fig:snakeneighbors}(b).  This leaves two sides of $B_1$ for which the approximation can pass through without violating the observation above: the north and east sides of $B_1$.  From the curves of the approximation in these two boxes, we can infer the sign pattern of $f_1$ and $f_2$ at the corners of the boxes.  For instance, if $f_1$ has a positive value at the southwest corner of $B_2$, then it has a negative value at the northwest corner of $B_1$ and the southeast corner of $B_2$.  In addition, it has a positive value at the northeast corner of $B_1$.  This sign pattern is not allowed since, by the intermediate value theorem, there are a pair of points on the east and west sides of $B_1\cup B_2$ where $f_1$ has the same value.  By a similar argument, the same result holds for $f_2$.  These conditions satisfy the hypotheses of Lemma \ref{lem:multicrossings}, and so this configuration is impossible.  
	\end{proof}
	
	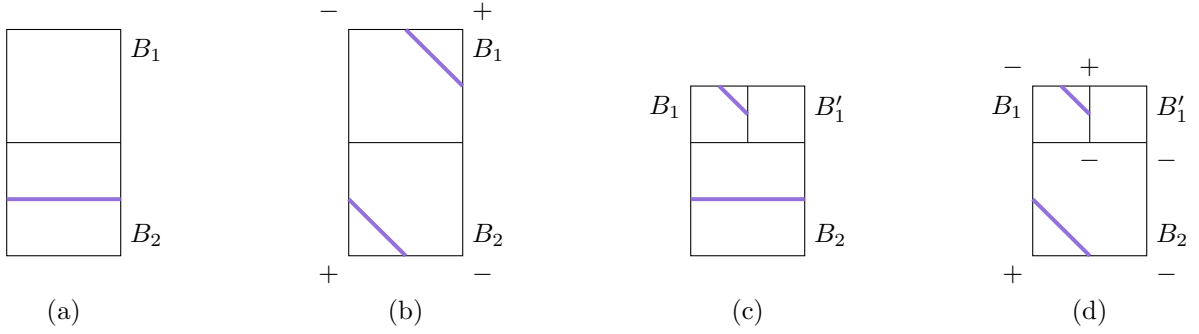
\begin{figure}
		\begin{center}
			\begin{tikzpicture}[scale=1.5]
				\definecolor{purple1}{RGB}{147, 112, 219}
				\draw (0,0) -- (1,0) -- (1,2) -- (0,2) -- cycle;
				\draw (0,1) -- (1,1);
				\draw[color = purple1,line width = 1.5pt](0,.5) -- (1,.5);
				\node[below right] at (1,2) {$B_1$};
				\node[above right] at (1,0) {$B_2$};
				\node at (.5,-.5) {(a)};
				\begin{scope}[shift = {(3,0)}]
					\draw (0,0) -- (1,0) -- (1,2) -- (0,2) -- cycle;
					\draw (0,1) -- (1,1);
					\draw[color = purple1,line width = 1.5pt](0,.5) -- (.5,0);
					\draw[color = purple1,line width = 1.5pt](1,1.5) -- (.5,2);
					\node[below right] at (1,2) {$B_1$};
					\node[above right] at (1,0) {$B_2$};
					\node[below left] at (0,0) {$+$};
					\node[below right] at (1,0) {$-$};
					\node[above left] at (0,2) {$-$};
					\node[above right] at (1,2) {$+$};
					\node at (.5,-.5) {(b)};
				\end{scope}
				\begin{scope}[shift = {(6,0)}]
					\draw (0,0) -- (1,0) -- (1,1.5) -- (0,1.5) -- cycle;
					\draw (0,1) -- (1,1);
					\draw (.5,1) -- (.5,1.5);
					\draw[color = purple1,line width = 1.5pt](0,.5) -- (1,.5);
					\draw[color = purple1,line width = 1.5pt](.25,1.5) -- (.5,1.25);
					\node[below right] at (1,1.5) {$B_1'$};
					\node[below left] at (0,1.5) {$B_1$};
					\node[above right] at (1,0) {$B_2$};
					\node at (.5,-.5) {(c)};
				\end{scope}
				\begin{scope}[shift = {(9,0)}]
					\draw (0,0) -- (1,0) -- (1,1.5) -- (0,1.5) -- cycle;
					\draw (0,1) -- (1,1);
					\draw (.5,1) -- (.5,1.5);
					\draw[color = purple1,line width = 1.5pt](0,.5) -- (.5,0);
					\draw[color = purple1,line width = 1.5pt](.25,1.5) -- (.5,1.25);
					\node[below right] at (1,1.5) {$B_1'$};
					\node[below left] at (0,1.5) {$B_1$};
					\node[above right] at (1,0) {$B_2$};
					\node at (.5,-.5) {(d)};
					\node[above left] at (0,1.5) {$-$};
					\node[above] at (.5,1.5) {$+$};
					\node[below left] at (0,0) {$+$};
					\node[below right] at (1,0) {$-$};
					\node[below right] at (1,1) {$-$};
					\node[below] at (.5,1) {$-$};
				\end{scope}
			\end{tikzpicture}
		\end{center}
		\caption{Case-by-case analysis for Lemmas \ref{lem:snakeneighbors:samesize} (a,b) and \ref{lem:snakeneighbors:differentsize} (c,d).  In every case, condition $C_1^\times$ would be violated, so every one of these examples is not possible.}
		\label{fig:snakeneighbors}
	\end{figure}
	
	\begin{lemma}\label{lem:snakeneighbors:differentsize}
		Suppose that $S$ is a snake and for every box $B\in S$, $C_1^\times(\cN(B))$ holds.  Let $B_1$ and $B_2$ be in the interior of $S$ such that $B_1$ and $B_2$ share an edge.  If $B_1$ is half the size of $B_2$, then the shared approximation passes through that shared edge or the shared edge of the common neighbor of $B_1$ and $B_2$.
	\end{lemma}
	
	\begin{proof}
		Up to symmetry, suppose that $B_1$ is above $B_2$ and that the west sides of $B_1$ and $B_2$ are aligned.  Let $B'_1$ be the complementary box, which is the same size as $B_1$ and is along the same edge of $B_2$.  We note that by the observation at the beginning of Lemma \ref{lem:snakeneighbors:samesize}, the approximations cannot pass through the same side of a box twice when the neighbor has been subdivided.
		
		Up to symmetry, there are two cases to consider.  If the curve approximation passes from the west to east sides of $B_2$, then the curve must pass from the north to east sides of $B_1$, but then, there are not enough edges in $B_1'$ to continue the approximation by the observation above.
		
		Now, assume that the approximation passes through the south side of $B_2$, and, up to symmetry, the west side of $B_2$.  This leaves two sides for $B_1$ to continue the approximation, its north and east sides.  By looking at the sign patterns that these intersections imply, we are in a nearly identical case as in the second part of Lemma \ref{lem:snakeneighbors:samesize}, and this is not possible by Lemma \ref{lem:multicrossings}.
	\end{proof} 
	
	Together, Lemmas \ref{lem:snakeneighbors:samesize} and \ref{lem:snakeneighbors:differentsize} imply that three of the sides of a head of a snake have no adjacent snake boxes.  Since a snake must be topologically a disk, if there were another neighbor to the head of the snake, then all of the boxes inside the loop would need to contain the approximation.  However, an interior box $B$ of the snake cannot be completely surrounded by other interior boxes of the snake since Lemmas \ref{lem:snakeneighbors:samesize} and \ref{lem:snakeneighbors:differentsize} would imply that there are two components in $B$, but this violates Lemma \ref{lem:onecomponent}.  Therefore, the only possibility would be for the snake to curl around the head.  This, however, leads to two crossings on the collinear sides, which is also not possible, see Figure \ref{fig:adjacenttosnakehead}.  We collect this result in the following lemma.
	
	\begin{lemma}
		Suppose that $S$ is a snake and for every box $B\in S$, $C_1^\times(\cN(B))$ holds.  Three of the sides of each head of the snake are not adjacent to any other boxes of the snake.
	\end{lemma}
	
	\begin{figure}[hbt]
		\begin{center}
			\begin{tikzpicture}[scale=1.5]
				\definecolor{purple1}{RGB}{147, 112, 219}
				\draw (0,0) -- (1,0) -- (1,1) -- (0,1) -- cycle;
				\draw (1,0) -- (1.5,0) -- (1.5,.5) -- (1,.5);
				\draw (1.5,.5) -- (1.5,1) -- (1,1);
				\draw (1,1) -- (1,2) -- (2,2) -- (2,1) -- (1.5,1);
				\draw (0,1) -- (0,2) -- (1,2);
				\node[above left] at (0,0) {$H$};
				\draw[color=blue,line width = 1.5pt, dotted] (0,.5) -- (1,.25); 
				\draw[color=purple1,line width = 1.5pt](1,.25) -- (1.25,.5) -- (1.25,1) -- (1,1.5) -- (.5,2);
				\draw[color=red,line width = 1.5pt, dashed](.5,0) -- (1,.25);
				I	\end{tikzpicture}
		\end{center}
		\caption{It is impossible for an interior box of a snake to be adjacent to the head of a snake.  The box labeled $H$ is the head of the snake and the vertical line to the right of the head of the snake contains two intersections with the approximations, which is not possible.}
		\label{fig:adjacenttosnakehead}
	\end{figure}
	
	\begin{proposition}\label{prop:snake}
		Suppose that $S$ is a snake and for every box $B\in S$, $C_1^\times(\cN(B))$ holds.  There is at most one crossing in $\cN(S)$ corresponding to the snake.
	\end{proposition}
	\begin{proof}
		Let $\alpha_{f_1}=\alpha_{f_2}$ be the two subsets of the approximation in $S$ which share vertices or segments.  Let $\gamma_{f_1}$ be the subset of $\cV(f_1)\cap\cN(S)$ that deforms to $\alpha_{f_1}$ under the ambient isotopy.  Let $\overline{\gamma}_{f_1}$ be the extension $\gamma_{f_1}$ in $\cN_2(S)$, but without excursions.  We define $\gamma_{f_2}$ and $\overline{\gamma}_{f_2}$ analogously, see Figure \ref{fig:snake_proof_schematic.txt}.
		
		We consider both $\overline{\gamma}_{f_1}$ and $\overline{\gamma}_{f_2}$ as paths and fix an orientation for each path.  Walking along the path $\overline{\gamma}_{f_1}$, the sign of $\nabla f_1\times \nabla f_2$ at intersections is completely determined by whether $\overline{\gamma}_{f_1}$ is passing from the negative side of $\overline{\gamma}_{f_2}$ to the positive side or vice versa.
		
		Suppose that $\overline{\gamma}_{f_1}$ and $\overline{\gamma}_{f_2}$ intersect multiple times in $\cN(S)$.  We show that this violates $C_1^\times(\cN(B))$ for some $B\in S$, see Figure \ref{fig:snake_proof_schematic.txt} for details.  Walking along $\overline{\gamma}_{f_1}$, we observe that the sign of $f_2|_{\overline{\gamma}_{f_1}}$ changes precisely when the curve $\overline{\gamma}_{f_1}$ crosses $\overline{\gamma}_{f_2}$.  In particular, the type of crossing, that is, from negative to positive or from positive to negative, alternate.  Therefore, since there are at least two crossings, there are two crossings where the sign of $\nabla f_1\times \nabla f_2$ differs.
		
		Finally, we apply the function $\nabla f_1\times \nabla f_2$ to the path $\overline{\gamma}_{f_1}$.  In other words, we look at $\nabla f(\overline{\gamma}_{f_1})\times \nabla g(\overline{\gamma}_{f_1})$.  We know that the value of this continuous function has different signs along the path.  Hence, by the intermediate value theorem, there is some point on this curve where this function vanishes, that is, the gradients of $f_1$ and $f_2$ are parallel, but this is not possible as it would contradict $C_1^\times(\cN(B))=\true$ for all $B\in S$.
	\end{proof}
	
	\begin{figure}[htb]
		\centering
		\input{figures_source/snake_proof_schematic.tex}
		\caption{An illustration of the details of the proof of Proposition \ref{prop:snake}.  At two consecutive intersections, the turn directions of the gradients differ.  Therefore, the gradients of $f_1$ and $f_2$ must be parallel somewhere along the path between the intersections.}
		\label{fig:snake_proof_schematic.txt}
	\end{figure}
	
	We have shown that every snake corresponds to at most one intersection of $\cV(f_1)$ and $\cV(f_2)$.  It remains to discuss how to decide if a snake corresponds to an intersection.  We begin by defining the orientation of points with respect to the snake.  Let $S$ be a snake and let $B$ be a head of the snake.  Suppose that $p$ and $q$ are two points on the external boundary $\partial B$ of $B$, that is, the three sides of $B$ that do not include the snake.  We say that 
	$q$ is \emph{clockwise} from $p$ with respect to the snake if walking around $\partial B$ in a clockwise direction starting at the snake reaches $p$ before $q$.  In this case, $p$ is \emph{counterclockwise} from $q$ with respect to the snake, see Figures \ref{fig:snake_schematic} and \ref{fig:snake_proof_schematic.txt}.
	
	\begin{lemma}\label{lem:snakecrossings}
		Let $S$ be a snake and let $\alpha_{f_1}$ and $\alpha_{f_2}$ be the two components of $\cA(f_1)\cap S$ and $\cA(f_2)\cap S$, respectively, that include the shared edges of the snake.  Let $B_1$ and $B_2$ be the two heads of $S$.  Let $p_1$ and $q_1$ be the ends of $\alpha_{f_1}$ and $\alpha_{f_2}$, respectively, in $B_1$, and define $p_2$ and $q_2$ similarly.  The snake corresponds to an intersection if and only if the orientations from $p_1$ to $q_1$ and $p_2$ to $q_2$ are the same.
	\end{lemma}
	
	\begin{proof}
		Let $\gamma_{f_1}$ be the preimage of $\alpha_{f_1}$ under the Plantinga and Vegter isotopy.  In the Plantinga and Vegter isotopy, there is some flexibility in how the isotopy is done.  In particular, when removing excursions on an edge $e$, there is at least one point of $\cV(f_1)\cap e$ which can remain fixed throughout the isotopy, for instance, one of the extremal points of $\cV(f_1)\cap e$.  By appealing to this flexibility, we may assume that $\gamma_{f_1}$ ends on the same edges of the subdivision as $p_1$ and $p_2$.  Note that $\gamma_{f_1}$ is contained in $\cN(S)$, but might not be a component of either $\cV(f_1)\cap \cN(S)$ or $\cV(f_1)\cap S$ since there is no guarantee that the endpoints of $\gamma_{f_1}$ correspond to either the first or last excursion.
		
		Let $\widetilde{\gamma}_{f_1}$ be the extension of $\gamma_{f_1}$ into $\cN(S)$ with excursions.  We define $\gamma_{f_2}$ and $\widetilde{\gamma}_{f_2}$ similarly.  Then the endpoints of $\widetilde{\gamma}_{f_1}$ and $\widetilde{\gamma}_{f_2}$ are on the same edges of the subdivision as $\alpha_{f_1}$ and $\alpha_{f_2}$.  Since the endpoints of $\alpha_{f_1}$ and $\alpha_{f_2}$ are on different edges of the subdivision, the endpoints of $\alpha_{f_1}$ and $\alpha_{f_2}$ have the same clockwise or counterclockwise relationship as the endpoints of $\widetilde{\gamma}_{f_1}$ and $\widetilde{\gamma}_{f_2}$. 
		
		Let $S_D$ be the subset of $\cN(S)$ consisting of the union of $S$ with external semicircles on the external edges of $S$.  If necessary, we slightly shrink $S_D$ in the following argument so that $S$ and $S_D$ have the same topology.  Let $\gamma_{f_1,D}$ be the extension of $\widetilde{\gamma}_{f_1}$ into $S_D$ until the curve reaches the circular boundaries of the semicircles attached to the edges containing $p_1$ and $p_2$.  By Corollary \ref{cor:semicircle}, $\gamma_{f_1,D}\supseteq\widetilde{\gamma}_{f_1}$.  We define $\gamma_{f_2,D}$ similarly.
		
		We observe that external excursions from non-head boxes of $S$ cannot include intersections as follows:  Suppose, for contradiction, that an external excursion contains an intersection.  If both $\cV(f_1)$ and $\cV(f_2)$ had an excursion on the same external edge, then this would violate Lemma \ref{lem:multicrossings}.  Therefore, if $\cV(f_1)$ has an excursion on edge $e$, then $f_2|_e$ does not change signs.  Therefore, by restricting $f_2$ to the excursion of $\cV(f_1)$, there must be an even number of intersections because each intersection corresponds to a sign change of $f_2$.  This contradicts Proposition \ref{prop:snake}, so the number of intersections is $0$.  Since there are no intersections in $S_D\setminus S$ except perhaps between the curves of the heads of the snake, we may slightly shrink $S_D$ without changing the topology of $S_D$, if necessary, so (1) that the endpoints of $\gamma_{f_1,D}$ and $\gamma_{f_2,D}$ are on the boundary of $S_D$ and (2) each region of $S_D\setminus S$ contains excursions from a single box of $S$.
		
		By the condition in Definition \ref{def:snake}, $S$ and $S_D$ are topologically disks and their boundaries are topological circles.  Moreover, $\gamma_{f_1,D}$ and $\gamma_{f_2,D}$ form curves across these disks.  The paths $\gamma_{f_1,D}$ and $\gamma_{f_2,D}$ intersect if and only if their endpoints interweave along the boundary of the circle.  Interweaving is equivalent to having the same clockwise or counterclockwise order of their endpoints around the circle.  Since their endpoints are on the curved semicircular regions attached to the edges of $S$, the endpoints of $\alpha_{f_1}$ and $\alpha_{f_2}$ are ordered in the same way as the endpoints of $\gamma_{f_1,D}$ and $\gamma_{f_2,D}$.
	\end{proof}
	
	\subsection{Topological correctness}\label{sec:simul_approx}
	
	We provide a complete algorithm for approximating a pair of curves and prove that the output of this algorithm is topologically correct.
	
	\begin{algorithm}[hbt]
		\caption{Simultaneous approximation algorithm for two curves}
		\label{alg:full}
		\begin{algorithmic}[1]
			\Require polynomials $f_1,f_2\in\ZZ[x,y]$, with $\cV(f_1)$ and $\cV(f_2)$ smooth curves such that $\cV(f_1)$ and $\cV(f_2)$ intersect transversely, and a square region $R$ with integral corners such that $\cV(f_1)$ and $\cV(f_2)$ do not intersect on the boundary of $R$.
			\Ensure approximations $\cA(f_1)$ and $\cA(f_2)$ such that $(\cA(f_1),\cA(f_2))$ approximates $(\cV(f_1),\cV(f_2))$.
			\State Subdivide $R$ using the new subdivision step, Algorithm \ref{alg:subdivision}.
			\State Further subdivide boxes $B$ intersecting $\partial R$ until either $C_0^{f_1}(B)=\true$ or $C_0^{f_2}(B)=\true$.
			\State Further subdivide boxes until the side lengths of neighboring boxes differ by at most a factor of two.
			\State Compute the Plantinga and Vegter curve approximation.
			\State Subdivide boxes which have shared approximations, but fail the condition of Definition \ref{def:snake}.
			\State Recompute the Plantinga and Vegter curve approximations, if necessary.
			\State For any snakes, apply Lemma \ref{lem:snakecrossings}.
			\renewcommand{\algorithmicif}{\hspace{\algorithmicindent}\textbf{if}}
			\If{there is no crossing}
			\renewcommand{\algorithmicif}{\textbf{if}}
			\State \hspace{\algorithmicindent}slightly separate the edges of the snake so that the common edges do not overlap.
			\renewcommand{\algorithmicelse}{\hspace{\algorithmicindent}\textbf{else}}
			\Else\hspace{\algorithmicindent}
			\renewcommand{\algorithmicelse}{\textbf{else}}
			\State \hspace{\algorithmicindent}slightly separate the ends of the snake and add an explicit crossing in the middle of the snake.
			\EndIf
			\State\Return {the approximations}
		\end{algorithmic}
	\end{algorithm}
	
	\begin{theorem}\label{thm:twocurves}
		Suppose that $f_1,f_2\in\ZZ[x,y]$, $R=[a,b]\times[c,d]\subseteq\RR^2$ a rectangle such that $a,b,c,d\in\ZZ$.  Suppose also that $\cV(f_1)$ and $\cV(f_2)$ are nonsingular within $R$, do not have a common intersection on the boundary of $R$, and intersect transversely within $R$.  Let $\cA(f_1)$ and $\cA(f_2)$ be the output of Algorithm \ref{alg:full}.  There exists a set $U$ containing $R$ and an ambient isotopy defined on $U$ that simultaneously takes $\cV(f_1)$ to $\cA(f_1)$ and $\cV(f_2)$ to $\cA(f_2)$, respectively.
	\end{theorem}
	\begin{proof}
		Through a rescaling, we reduce to the case where $R$ is a square.  Second, we use the techniques for unbounded curves from \citet{BCGY:2012} for the boundary boxes, as long as, in each boundary box, either $C_0^{f_1}(B)=\true$ or $C_0^{f_2}(B)=\true$.  Therefore, we focus on the topological correctness in the interior of $R$.
		
		By Proposition \ref{prop:nomissing}, every intersection between $\cV(f_1)$ and $\cV(f_2)$ corresponds to either a transversal crossing of $\cA(f_1)$ and $\cA(f_2)$ or a snake.  By investigating the proofs of Propositions \ref{prop:plus_sign} and \ref{prop:snake}, we see that these two features correspond to different types of crossings and cannot identify the same crossing twice.  Thus, by Lemma \ref{lem:snakecrossings}, we identify exactly when a crossing occurs.  This gives a bijection between crossings in the approximation and crossings in the varieties.  
		
		We begin the isotopy by deforming space so that each excursion is moved into a neighboring box.  This can be done as no intersection can involve two excursions as follows: Suppose, to the contrary, that there is a box $B$ such that $\cV(f_1)$ and $\cV(f_2)$ have excursions in $B$ which intersect.  Let the curves of these excursions be $\gamma_{f_1}$ and $\gamma_{f_2}$, respectively.  If these excursions are based on different edges of the subdivision, then $\gamma_{f_1}$ along with the segment between its endpoints forms a closed loop.  If the path $\gamma_{f_2}$ enters the region bounded by this closed loop, then it must also exit this region as the endpoints of $\gamma_{f_2}$ are outside of this region.  By assumption, no intersection can be on the segment between the endpoints of $\gamma_{f_1}$.  Hence, the second intersection must also be with the curve $\gamma_{f_1}$ This accounts for two intersections within $B$, which is not possible by Lemma \ref{cor:multiintersection}.  
		
		On the other hand, if these excursions are on the same edge of the subdivision, then the endpoints of $\gamma_{f_1}$ and $\gamma_{f_2}$ satisfy the conditions of Lemma \ref{lem:multicrossings}, which is also not possible.  Therefore, it is possible to apply this isotopy without introducing or removing any intersections between the curves. 
		
		Let $\widetilde{V}(f_1)$ and $\widetilde{V}(f_2)$ be the result of applying the isotopy from the paragraph above.  By the correctness of the Plantinga and Vegter algorithm, within each box $B$, $\widetilde{\cV}(f_1)\cap B$ is ambient isotopic to $\cA(f_1)\cap B$.  Similarly, $\widetilde{\cV}(f_2)\cap B$ is ambient isotopic to $\cA(f_2)\cap B$.  Let $\gamma_{f_1}$ be a component of $\widetilde{\cV}(f_1)$ and $\gamma_{f_2}$ a component of $\widetilde{\cV}(f_2)$.  If an endpoint of $\gamma_{f_1}$ and $\gamma_{f_2}$ are on the same edge of the subdivision, then the straightening step in the Plantinga and Vegter isotopy moves the endpoints to the same point, that is, it creates a snake.  
		
		On the other hand, suppose that no two components have endpoints on the same edge of the subdivision.  In addition, suppose that $\gamma_{f_1}$ and $\gamma_{f_2}$ intersect.  By Corollary \ref{cor:multiintersection}, these two curves cannot intersect more than once.  Therefore, the endpoints of $\gamma_{f_1}$ separate the endpoints of $\gamma_{f_2}$.  Moreover, since the edges containing the endpoints are different, we find, by the intermediate value theorem, that any path in $B$ between the edges containing the endpoints of $\gamma_{f_1}$ must intersect any path in $B$ between the edges containing the endpoints of $\gamma_{f_2}$.  In particular, this implies that the approximations intersect in $B$.
		
		Finally, for any components of $\widetilde{\cV}(f_1)\cap B$ and $\widetilde{\cV}(f_2) \cap B$ which do not intersect other curve components and whose endpoints lie on edges which are not shared by other curves, we observe that these components can be straightened without intersecting other components.  In particular, as a curve is straightened its endpoints do not leave the edges in which they start and the interior of the curve does not intersect the boundary of the box.  Hence, during the straightening step of the Plantinga and Vegter isotopy, each curve can be straightened individually while the rest of the curves are pushed off to the side.  This straightens all of the curve components within the box without introducing intersections.
	\end{proof}
	
	We note that small Hausdorff distance between the approximation and the variety can be achieved by making sure that the boxes containing the approximations are sufficiently small and that any snakes are also small.  The output of this algorithm on the polynomials generating Figures~\ref{fig:tangential_intersection_nonintersect}(a) and \ref{fig:tangential_intersection_nonintersect}(b) appears in Figure \ref{fig:tangential_intersection_nonintersect_after}.  These new versions correctly approximate the topology of the curves.
	
	\begin{figure}[htb!]
		\centering
		\input{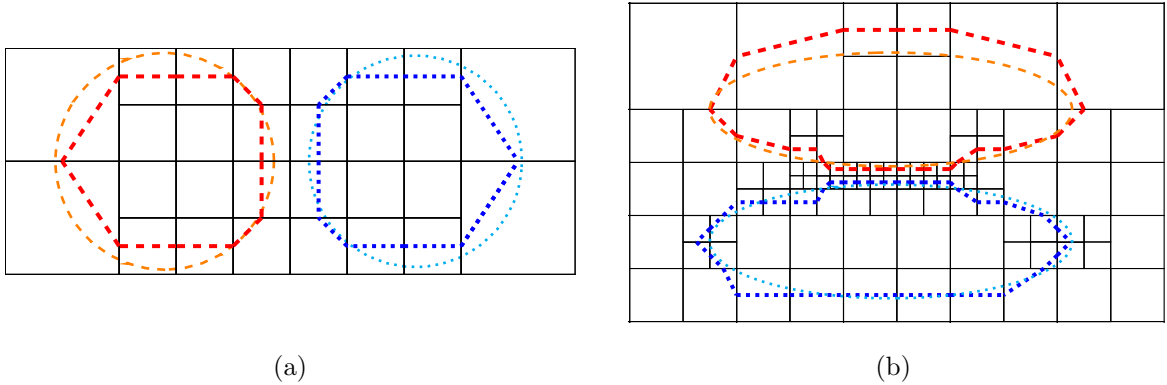}\\
		\caption{Topologically correct versions of (a) Figure \ref{fig:tangential_intersection_nonintersect}(a) and (b) Figure \ref{fig:tangential_intersection_nonintersect}(b) using Algorithm \ref{alg:full}.}
		\label{fig:tangential_intersection_nonintersect_after}
	\end{figure}
	
	\section{Complexity Analysis}\label{sec:complexity}
	
	We provide the complexity analysis of the approximation algorithm for two curves.  In this analysis, we compute the number of boxes that are used in the subdivision as well as the bit-complexity of the  algorithm.  We provide both adaptive bounds based on continuous amortization \citep{BGT:2020,B:2016,BK:2012} and non-adaptive worst-case bounds.  Our analysis relies heavily on the analysis in \citet{BGT:2020}.  For our bit-complexity analysis, we assume that Step 5 of Algorithm \ref{alg:full} is not a dominating term of the complexity.
	
	Our analysis assumes that all predicates are implemented using the {\em centered form}, see, for example \citet{RR:1984}.  In addition, we observe that all of our predicates are for tests for the form $0\not\in F(B)$ for some appropriate function $F$ and axis-aligned square or four-dimensional cube region $B$.  Using the centered form, these tests simplify to
	$$
	\left(\sum_{|\alpha|=1}^d
	\frac{|\partial f^\alpha F(m)|}{|F(m)|\alpha!}\left(\frac{w}{2}\right)^{|\alpha|}
	\right)<1
	$$
	where $m$ is the {\em midpoint} of $B$, that is, the average of the vertices, and $w$ is the {\em width} of $B$, that is, the length of a side, see \cite[Corollary 14]{BGT:2020} for additional details.
	
	In our complexity analysis, we use a slightly weaker version of Algorithm \ref{alg:subdivision}.  Instead of testing $C_1^\times(\cN_2(B)\cap 13B)$, we only test $C_1^\times(13B)$.  If $C_1^\times(13B)=\true$, then $C_1^\times(\cN_2(B)\cap 13B)=\true$.  This change can only increase the number of boxes found in the subdivision.
	
	\subsection{Adaptive bounds}\label{sec:adaptivebounds}
	We begin by recalling the basic facts of continuous amortization as developed in \citet{BGT:2020}, \citet{B:2016}, and \citet{BK:2012}.  
	
	\begin{definition}[{cf.} {\citet[Definition 4]{BGT:2020}}]
		Let $R=[a,b]\times[c,d]\subseteq\RR^2$ be a rectangle and $C$ a predicate.  A {\em local size bound} for $C$ is a function $F:R\rightarrow\RR_{\geq 0}$ with the property that 
		$$
		F(x)\leq \inf_{\substack{J\text{ a square}\\x\in J\\C(J)=\false}}\area(J).
		$$
	\end{definition}
	In other words, $F(x)$ is a lower bound on the area of any square containing $x$ which fails the predicate.  This local size bound provides a way to calculate the complexity of a subdivision-based algorithm.
	\begin{proposition}[{cf.} {\citet[Theorem 5]{BGT:2020}}]\label{prop:continuousamortization}
		Let $R=[a,b]\times[c,d]\subseteq\RR^2$ be a square, $C$ a predicate, and $F$ a local size bound for $C$.  Let $h:\RR_{\geq 0}\rightarrow \RR$ be a non-increasing function.  Let $P$ be the final partition formed by iteratively splitting $R$ into four smaller squares until $C$ holds on each square.  Then
		$$
		\sum_{B\in P}h(\area(B))\leq\max\left\{h(\area(R)),\int_R\frac{4h(F(x)/4)}{F(x)}dA\right\}.
		$$
		If $h\equiv 1$, then the integral bounds the number of regions formed by the subdivision.  If the subdivision does not terminate, then the integral is infinite.
	\end{proposition}
	
	In the current paper, there are five predicates to consider $C_0^{f_1}$, $C_0^{f_2}$, $C_1^{f_1}$, $C_1^{f_2}$, and $C_1^\times$.  Of these, there are local size bounds provided in \citet{BGT:2020} for the first four.  In particular, \citet[Corollary 16]{BGT:2020} gives a local size bound for $C_0^{f_i}$:
	$$
	G_0^{f_i}(x) = \left(\frac{ 2\ln\left(\frac{5}{4}\right)\dist_{\CC}(x, f_i) }{ 4\deg(f_i) + \sqrt{2}\ln\left(\frac{5}{4}\right) }\right)^2,
	$$
	where $\dist_{\CC}(x, f_i)$ denotes the distance from $x\in\RR^2$ to the complex variety $\cV_{\CC}(f_i)$.  Similarly, \citet[Corollary 19]{BGT:2020} gives a local size bound for $C_1^{f_i}$:
	$$
	G_1^{f_i} = \left(\frac{ 2\ln\left(\frac{65}{64}\right) \dist_{\CC}\left(\left(x,x\right), g_i\right)}{ 32\left(\deg(f_i) - 1\right) + 2\ln\left( \frac{65}{64} \right) }\right)^2,
	$$
	where $g_i(x_1, x_2, y_1, y_2) = \left\langle \nabla f_1(x_1, x_2), \nabla f_2(y_1, y_2) \right\rangle$, and $\dist_{\CC}((x,x),g_i)$ denotes the distance from the point $(x,x)$ in the diagonal of $\RR^2\times\RR^2\subseteq \CC^4$ to the complex variety $\cV_{\CC}(g_i)$. 
	
	It remains for us to compute a local size bound for $C_1^\times$.  Define $h(x_1,x_2,y_1,y_2)=\nabla f_1(x_1,x_2)\times \nabla f_2(y_1,y_2),$ where we omit the selection of the third (and only nonzero) entry in the cross product.  We observe that the $C_1^\times$ test is based on whether $\square h(B,B)$ contains $0$.  The construction in \citet{BGT:2020} shows how to develop a local size bound for any function of this form.  We apply these results to find a local size bound for $C_1^\times$.
	
	\begin{lemma}\label{lem:localsizeboundcross}
		Let $f_1,f_2 \in \RR\left[x, y\right]$ and define $h \in \RR\left[x_1, x_2, y_1, y_2\right]$ as $h(x_1, x_2, y_1, y_2) = \nabla f_1(x_1, x_2) \times \nabla f_2(y_1, y_2).$ Let $B \subseteq \RR^2,$ and suppose that there is a point $x \in B$ such that 
		$$
		13\Diam(B)\leq \frac{ 2\sqrt{2}\ln\left( \frac{65}{64} \right) \dist_{\CC}\left(\left(x,x\right), h\right)}{ 32\left(d - 1\right) + 2\ln\left( \frac{65}{64} \right) },
		$$
		where $d=\max\{\deg(f_1),\deg(f_2)\}$.  Then $C_1^\times(13B)=\true$.  Moreover, a local size bound for this test is
		$$
		G_1^\times(x) =  \left(\frac{ 2\ln\left( \frac{65}{64} \right) \dist_{\CC}\left(\left(x,x\right), h\right)}{ 416\left(d - 1\right) + 26\ln\left( \frac{65}{64}\right) }\right)^2.
		$$
	\end{lemma}
	\begin{proof}
		By inspection, the bounds arrived at in \citet[Corollaries 19 and 20]{BGT:2020} depend only on the degrees and numbers of variables of the polynomials.  We observe that the degree of $h$ is bounded by $2(d-1)$ and that $f_1$ and $f_2$ are bivariate polynomials.  These are the exact same estimates that are used to develop the bounds in \citet[Corollaries 19 and 20]{BGT:2020}.  Therefore, the bounds in \citet[Corollaries 19 and 20]{BGT:2020} apply in this case, and the bounds above follow directly.
	\end{proof}
	
	Since the predicates in Algorithm \ref{alg:subdivision} are compound predicates, we combine the local size bounds from above to accommodate these predicates.  We denote the predicates as $C_{00}$, $C_{01}$, $C_{10}$, and $C_{11}$, where the subscripts indicate whether $C_0^{f_i}$ or $C_1^{f_i}$ is used.  For example, $C_{01}(B)$ is $\true$ when both $C_0^{f_1}(B)$ and $C_1^{f_2}(B)$ are $\true$.  Hence, $C_{01}(B)$ is guaranteed to be $\true$ if $\area(B)\leq G_0^{f_1}(x)$ and $\area(B)\leq G_1^{f_2}(x)$ for some $x\in B$.  Therefore, the minimum of the two local size bounds is a local size bound for $C_{01}$.  More precisely, the local size bounds for the predicates in Algorithm \ref{alg:subdivision} follow.
	\begin{align*}
		G_{00} &= \min\left(G_0^{f_1}, G_0^{f_2}\right) &
		G_{01} &= \min\left(G_0^{f_1}, G_1^{f_2}\right) \\
		G_{10} &= \min\left(G_1^{f_1}, G_0^{f_2}\right) &
		G_{11} &= \min\left(G_1^{f_1}, G_1^{f_2}, G_1^\times\right) 
	\end{align*}
	A direct application of Proposition \ref{prop:continuousamortization} gives the following result on the complexity of the approximation algorithm for two curves.
	\begin{theorem}
		Suppose that $f_1, f_2 \in \ZZ[x, y]$.  Define $g_i(x_1,x_2,y_1,y_2)=\langle\nabla f_i(x_1,x_2),\nabla f_i(y_1,y_2)\rangle$ and $h(x_1,x_2,y_1,y_2)=\nabla f_1(x_1,x_2)\times\nabla f_2(y_1,y_2)$.  Suppose that $R=[a,b]\times[c,d]$ is a square subset of $\RR^2$ such that $a,b,c,d\in\ZZ$.  The number of regions after the subdivision step performed by Algorithm \ref{alg:subdivision} is bounded above by the maximum of $1$ and 
		$$4 \int_R \min\left\lbrace \frac{1}{G_{00}}, \frac{1}{G_{01}}, \frac{1}{G_{10}}, \frac{1}{G_{11}} \right\rbrace dA.$$ Moreover, the algorithm does not terminate if and only if the integral diverges. 
	\end{theorem}
	
	\subsection{Nonadaptive bounds}
	
	We begin by recalling the structure for complexity estimates for nonadaptive bounds from \citet{BGT:2020}.
	
	\begin{definition}[{cf.} {\citet[Definition 2]{BGT:2020}}]
		Let $R=[a,b]\times[c,d]\subseteq\RR^2$ a rectangle and $C$ a predicate.  $C$ is a {\em diameter distance test} if there is a closed set $V$ and positive constant $K$ such that for all squares $B$, we have 
		$$
		\text{If }\left(\Diam(B) < K \max_{x \in J}d(x, V) \right), \text{ then } C(B) = \textsc{True},
		$$
		where $d(x,V)=\min_{v \in V}d(x, v).$
	\end{definition}
	Suppose that there is a collection of predicates $\{C_i\}$, each with a corresponding closed set $V_i$, and $K$ a uniform constant for all tests.  Moreover, define the separation bound $\delta$ as $\displaystyle 0 < \delta \leq \min_{x\in R}\max_{i}d(x, V_i).$  The product $K\delta$ is a uniform bound on the diameter of a square which is guaranteed to be terminal.  Then \citet[Corollary 21]{BGT:2020} implies that when $R$ is square, the number of boxes produced by Algorithm \ref{alg:subdivision} is at most
	\begin{equation}\label{eq:boxcount}
		\max\left\{1,\frac{4\Diam(R)^2}{K^2\delta^2}\right\}.
	\end{equation}
	From our local size bounds, we see that a uniform bound on $K$ for our tests is
	\begin{equation}\label{eq:Kbound}
		K=\min\left\{\frac{ 2\sqrt{2}\ln\left(\frac{5}{4}\right)}{ 4\deg(f_i) + \sqrt{2}\ln\left(\frac{5}{4}\right) },
		\frac{ 2\sqrt{2}\ln\left(\frac{65}{64}\right)}{ 32\left(\deg(f_i) - 1\right) + 2\ln\left( \frac{65}{64} \right) },
		\frac{ 2\sqrt{2}\ln\left( \frac{65}{64} \right)}{ 416\left(d - 1\right) + 26\ln\left( \frac{65}{64}\right) }\right\}.
	\end{equation}
	This bound is immediate from our local size bounds from Section \ref{sec:adaptivebounds}, {cf.} \citet[Corollaries 17 and 20]{BGT:2020}.
	
	Since the tests in Algorithm \ref{alg:subdivision} are compound tests, the required closed sets are built out of multiple varieties.  For instance, for the $C_{01}$ test and a square $B$, if $\displaystyle\Diam(B)<K\max_{x\in B}d_{\CC}(x,\cV(f_1))$ and $\displaystyle\Diam(B)<K\max_{x\in B}d_{\CC}(x,\cV(g_2))$, then $C_{01}(B)$ is \true.  Therefore, we must find a separation bound $\delta$ so that
	\begin{multline}\label{eq:deltadef}
		0<\delta\leq\min_{x\in R}\max \left\lbrace \min \left\lbrace \dist_{\CC}(x, f_1), \dist_{\CC}(x, f_2) \right\rbrace, \min \left\lbrace \dist_{\CC}(x, f_1), \dist_{\CC}((x,x), g_2) \right\rbrace, \right.\\
		\left.\min \left\lbrace \dist_{\CC}((x,x), g_1), \dist_{\CC}(x, f_2) \right\rbrace, \min \left\lbrace \dist_{\CC}((x,x), g_1), \dist_{\CC}((x,x), g_2), \dist_{\CC}((x,x), h) \right\rbrace \right\rbrace.
	\end{multline}
	By exchanging the maximum and the minimums, we find that Inequality (\ref{eq:deltadef}) can be rewritten as
	\begin{multline*}
		0<\delta\leq\min_{x\in R}\min \left\lbrace \max \left\lbrace \dist_{\CC}(x, f_1), \dist_{\CC}((x, x), g_1) \right\rbrace, \max \left\lbrace \dist_{\CC}(x, f_2), \dist_{\CC}((x, x), g_2) \right\rbrace,\right.\\
		\left.\max \left\lbrace \dist_{\CC}(x, f_1), \dist_{\CC}(x, f_2), \dist_{\CC}((x,x), h)\right\rbrace  \right\rbrace.
	\end{multline*}
	
	For the remainder of this section, we heavily rely on the proofs in \citet{BGT:2020}.  In particular, we omit some details and cite results at the potential cost of some clarity.
	
	We use a direct application of \citet[Theorem 24]{BGT:2020} to compute lower bounds on both $\max \left\lbrace \dist_{\CC}(x, f_1), \dist_{\CC}((x, x), g_1) \right\rbrace$ and $\max \left\lbrace \dist_{\CC}(x, f_2), \dist_{\CC}((x, x), g_2) \right\rbrace$.  In particular, the theorem uses only the degrees and the number of variables to reach its result.  These same degree bounds and number of variables apply in this case.  Specializing the result to the bivariate case results in a bound from below of $2^{-O(d^{25}(d+\lg H))}$, where $H$ is the maximum absolute value of the coefficients of $f_1$, the coefficients of $f_2$, and the coordinates of the corners of the initial region $R$, and $d=\max\{\deg(f_1),\deg(f_2)\}$.
	
	We now focus on a lower bound for $\max \left\lbrace \dist_{\CC}(x, f_1), \dist_{\CC}(x, f_2), \dist_{\CC}((x,x), h)\right\rbrace$.  In order to bound this, we appeal to the computations in \citet{BGT:2020}.  As a first step, we make the ambient spaces of all three varieties the same by defining $\cV(f_1^\Delta)$ and $\cV(f_2^\Delta)$ to be the images of $\cV(f_1)$ and $\cV(f_2)$ in the diagonal of $\CC^4$, that is $\cV(f_i^\Delta)=\{(x_1,x_2,x_1,x_2):(x_1,x_2)\in\cV(f_i)\}.$  Since $f_1$ and $f_2$ define smooth real curves which intersect transversely in $R$, $f_1(x)$, $f_2(x)$, and $h(x,x)$ cannot all vanish simultaneously for $x\in R$.  This implies that $\cV(f_1^\Delta)$, $\cV(f_2^\Delta)$, and $\cV(h)$ cannot all intersect simultaneously at points $(x,x)\in R\times R$.  An application of \citet[Proposition 23]{BGT:2020} results in a bound for when these three varieties do not intersect.
	
	\begin{proposition}
		Let $f_1,f_2\in\ZZ[x_1,x_2]$ be of degree at most $d$ defining smooth real curves which intersect transversely.  Define $h\in\RR[x_1,x_2,y_1,y_2]$ as $h(x_1,x_2,y_1,y_2)=\nabla f_1(x_1,x_2)\times \nabla f_2(y_1,y_2)$.  Suppose that $R$ is an axis-aligned square whose corners have integral coordinates.  Let $H$ be the maximum absolute value of the coefficients of $f_1$, the coefficients of $f_2$, and the coordinates of the corners of $H$.  Let $f_i^\Delta=\{f_i,x_1-y_1,x_2-y_2\}$ be the polynomial system corresponding to the image of $\cV_{\CC}(f_i)$ in the diagonal of $\CC^4$ and $R^\Delta=\{(x,x):x\in R\}$ is the image of $R$ in the diagonal of $\CC^4$.  Let $(x,x)\in\cV_\CC(f_1^\Delta,f_2^\Delta,h)$, then the distance between $(x,x)$ and $R^\Delta$ is at least $2^{-O(d^8(d+\lg H))}$.
	\end{proposition}
	
	\begin{proof}
		The proof is nearly identical to the proof of \citet[Proposition 23]{BGT:2020}.  Namely, the only difference in the proof is in the number of inequalities and equalities needed to achieve the bound.  As in the original proof, $20$ inequalities are used to define the region, $6$ equalities each are used to define $f_1$ and $f_2$, and $2$ equalities are used to define $h$.  These equalities and inequalities have the same bit-sizes as those appearing in \citet[Proposition 23]{BGT:2020}.  This results in $6$ more equalities than in the original proof, but the resulting bound is unchanged when specialized to the bivariate case, that is, $2^{-O(d^8(d+\lg H))}$
	\end{proof}
	
	We may now follow the proofs of \citet[Proposition 22 and Theorem 24]{BGT:2020} to find a separation bound between the real parts of $\cV(f_1^\Delta,f_2^\Delta)$ and $\cV(h)$.  
	
	\begin{proposition}\label{prop:separatefh}
		Let $f_1,f_2\in\ZZ[x_1,x_2]$ be of degree at most $d$ defining smooth real curves which intersect transversely.  Define $h\in\RR[x_1,x_2,y_1,y_2]$ as $h(x_1,x_2,y_1,y_2)=\nabla f_1(x_1,x_2)\times \nabla f_2(y_1,y_2)$.  Suppose that $R$ is an axis-aligned square whose corners have integral coordinates.  Let $H$ be the maximum absolute value of the coefficients of $f_1$, the coefficients of $f_2$, and the coordinates of the corners of $H$.  Let $f_i^\Delta=\{f_i,x_1-y_1,x_2-y_2\}$ be the polynomial system corresponding to the image of $\cV_{\CC}(f_i)$ in the diagonal of $\CC^4$ and $R^\Delta=\{(x,x):x\in R\}$ is the image of $R$ in the diagonal of $\CC^4$.  The distance between $\cV_\CC(f_1^\Delta,f_2^\Delta)$ and $\cV_\CC(h)$ is $2^{-O(d^{25}(d+\lg H))}$.
	\end{proposition}
	\begin{proof}
		The proof is nearly identical to \citet[Proposition 22]{BGT:2020}, using $20$ inequalities to define the region, $6$ equalities for each $f_i$ and $2$ equalities for $h$.  This corresponds to $6$ extra equalities, but does not change the overall complexity of the result.  The scaled bit-sizes are also unchanged as was used in \citet{BGT:2020} ({cf.} the proof of Lemma \ref{lem:localsizeboundcross}).  After substituting the value of $k=O(d^8(d+\lg H))$, as in \citet[Proposition 24]{BGT:2020}, we arrive at the desired bound.
	\end{proof}
	
	By the triangle inequality we conclude that for any $x\in R$, either $\dist_\CC((x,x),h)=2^{-O(d^{25}(d+\lg H))}$ or $\dist_\CC((x,x),(f_1^\Delta,f_2^\Delta))=2^{-O(d^{25}(d+\lg H))}$.  We now bound the distance $\max \left\lbrace \dist_{\CC}(x, f_1), \dist_{\CC}(x, f_2)\right\rbrace$.  To do this, we follow \citet[Proposition 22]{BGT:2020}.  It is enough to focus on the distance $\dist_\CC(f_1,f_2)$ and not $\dist_\CC(f_1^\Delta,f_2^\Delta)$ since these two distances are related by a factor of $\sqrt{2}$.
	
	\begin{lemma}\label{lem:oldProp22}
		Let $f_1,f_2\in\ZZ[x_1,x_2]$ be of degree at most an even number $d$ and suppose that $B\subseteq \RR^2$ is a square whose corners are rational numbers.  Let $H$ be the maximum absolute value of the coefficients of $f_1$ and $f_2$ as well as the numerator of the corners of $B$.  Suppose that the denominators of the corners of $B$ are powers of $2$ bounded by $2^{O(d^8(d+\lg H))}$.  Let $k$ be a positive integer so that for any $x\in\cV(f_1,f_2),$ the distance between $x$ and $B$ is more than $\frac{\sqrt{2}}{2^{k-1}}.$  Then,
		$$
		\min_{x\in B}\max\{\dist_{\CC}(x,f_1),\dist_{\CC}(x,f_2)\}\geq \frac{1}{2^{k+O(d^8(d+\lg H))}}\left(2^{d(1+k+O(d^8(d+\lg H)))}Hd^8\right)^{-512d^8}.
		$$
	\end{lemma}
	\begin{proof}
		Let $\varepsilon=\frac{1}{2^k}$ and let $B_\varepsilon$ be a thickening of $B$ by $\varepsilon$.  More precisely, $B_\varepsilon$ is the Minkowski sum of $B$ with $([-\varepsilon,\varepsilon]\times[-i\varepsilon,i\varepsilon])^2.$  By assumption, $\cV_{\CC}(f_1,f_2)\cap B_\varepsilon$ is empty.  We now apply a homothety centered at the origin of a factor of $2^{k+O(d^8(d+\lg H))}$, which makes the coefficients of the image of $B$ 
		integers.  The resulting region has the following bounding inequalities
		\begin{align*}
			2^{k+O(d^8(d+\lg H))}H-2^{O(d^8(d+\lg H))}&\leq \Re(x_i)\leq 2^{k+O(d^8(d+\lg H))}H+2^{O(d^8(d+\lg H))}\\
			-2^{O(d^8(d+\lg H))}&\leq\Im(x_i)\leq 2^{O(d^8(d+\lg H))}.
		\end{align*}
		Moreover, after clearing fractions, the homothety changes the absolute value of the coefficients of $f_1$ and $f_2$ to $2^{dk+O(d^9(d+\lg H))}H$.  Identifying $\CC^2$ with $\RR^4$, we see that the box is defined by $8$ inequalities whose coefficients are bounded by $2^{k+O(d^8(d+\lg H))}H$.  Moreover, the polynomials $f_1$ and $f_2$ are split into their real and complex parts.  By splitting in this way, the coefficients may grow by a binomial coefficient which is trivially bounded by $2^d$.  Therefore, we have four equalities whose coefficients are bounded by $2^{d+dk+O(d^9(d+\lg H))}H$.
		
		Applying \citet[Theorem 1.2]{JPT:2013} to the inequalities for $B_\varepsilon$ with $f_1$ when compared to the inequalities for $B_\varepsilon$ with $f_2$, we get that the separation bound within the scaled copy of $B_\varepsilon$ is 
		$$
		\left(2^{d(1+k+O(d^8(d+\lg H)))}Hd^8\right)^{-512d^8}.
		$$
		Reducing this bound by the homothety results in the desired bound.
	\end{proof}
	
	Finally, to compute $\max \left\lbrace \dist_{\CC}(x, f_1), \dist_{\CC}(x, f_2), \dist_{\CC}((x,x), h)\right\rbrace$, we observe that by Proposition \ref{prop:separatefh} and the triangle inequality either $\dist_{\CC}((x,x), h)=2^{-O(d^{25}(d+\lg H))}$ or $\dist_{\CC}(x,(f_1^\Delta,f_2^\Delta))=2^{-O(d^{25}(d+\lg H))}$.  In the latter case, we substitute $k=O(d^{25}(d+\lg H))$ into Lemma \ref{lem:oldProp22} to reach a separation bound of $2^{-O(d^{34}(d+\lg H))}$.  Collecting these results, we have the following theorem.
	
	\begin{theorem}\label{thm:nonadaptivebound}
		Let $f_1,f_2\in\ZZ[x_1,x_2]$ be smooth and of degree at most $d$.  Define $g_1,g_2\in\RR[x_1,x_2,y_1,y_2]$ as $g_i(x_1,x_2,y_1,y_2)=\langle \nabla f_i(x_1,x_2),\nabla f_i(y_1,y_2)\rangle$ and $h\in\RR[x_1,x_2,y_1,y_2]$ as $h(x_1,x_2,y_1,y_2)=\nabla f_1(x_1,x_2)\times \nabla f_2(y_1,y_2)$.  Suppose that $R$ is an axis-aligned square whose corners have integral coordinates.  Let $H$ be the maximum absolute value of the coefficients of $f_1$, $f_2$,  and the coordinates of the corners of $H$.  The number of terminal regions produced by Algorithm \ref{alg:subdivision} is $2^{O(d^{34}(d+\lg H))}.$
	\end{theorem}
	\begin{proof}
		By Equation \ref{eq:boxcount}, the number of boxes produced by the algorithm is $2^{O(2\lg(\Diam(R))-2\lg K-2\lg \delta)}$.  Since the diameter of $R$ is $O(H)$, $K=O(1/d)$ by Equation (\ref{eq:Kbound}), and $\delta = 2^{-O(d^{34}(d+\lg H))}$, the result follows.
	\end{proof}
	
	\subsection{Bit-Complexity}
	
	As above, the bit-complexity calculation is very similar to the bit-complexity calculation from \citet{BGT:2020}.  In fact, since $h$, $g_1$, and $g_2$ all have the same degree bounds and same bounds on the sizes of their coefficients, the cost to compute all of the predicates needed for Algorithm \ref{alg:subdivision} is no more than three times the cost of the predicates for the Plantinga and Vegter algorithm as discussed in \citet{BGT:2020}.  Using the upper bound on the number of boxes from Theorem \ref{thm:nonadaptivebound}, and the corresponding lower bound on the minimum width of boxes, the bit-complexity analysis from \citet{BGT:2020} carries over to this case.
	\begin{theorem}[{cf.} {\citet[Corollary 28]{BGT:2020}}]
		Let $f_1,f_2\in\ZZ[x_1,x_2]$ be smooth and of degree at most $d$.  Define $g_1,g_2\in\RR[x_1,x_2,y_1,y_2]$ as $g_i=\langle \nabla f_i,\nabla f_i\rangle$ and $h\in\RR[x_1,x_2,y_1,y_2]$ as $h(x_1,x_2,y_1,y_2)=\nabla f_1(x_1,x_2)\times \nabla f_2(x_1,x_2)$.  Suppose that $R$ is an axis-aligned square whose corners have integral coordinates.  Let $\tau=\lg H$ be the maximum bit-size of the coefficients of $f_1$, $f_2$, and the corners of $R$.  The overall bit-complexity of Algorithm \ref{alg:subdivision} is 
		$$
		2^{O(d^{34}(d+\tau))}\widetilde{O}_B(d^{37}(d+\tau)),
		$$
		where the $\sim$ indicates that logarithmic factors have been suppressed and the $B$ indicates that this is measuring a bit-complexity.
	\end{theorem}
	
	\begin{remark}
		We expect that the complexity of Algorithm \ref{alg:subdivision} dominates the complexity of Algorithm \ref{alg:full}, but we do not address the additional subdivisions required for boxes on the boundary and to ensure that snakes are disks.
	\end{remark}
	
	There are also adaptive bounds for the bit-complexity of Algorithm \ref{alg:subdivision}.  Once again, since the predicates needed for Algorithm \ref{alg:subdivision} have the same data as the tests in the Plantinga and Vegter algorithm, the complexity bounds from \citet{BGT:2020} apply directly.  In particular the cost functions for computing $C_0^{f_1}$ and $C_0^{f_2}$ is 
	$$
	h_0(y) = \left( d^3 \lg(w(I)) - \frac{d^3}{2}\lg(y) + d^2\tau \right)k_0\left(d, \tau, 2\right).
	$$
	On the other hand, the cost function for computing $C_1^{f_1}$, $C_1^{f_2}$, and $C_1^\times$ is 
	$$
	h_1(y) = \left( 2^4d^{5} \lg(w(I)) - \frac{2^4d^{5}}{2}\lg(y) + 2^{4}d^4\tau \right)k_1\left(d, \tau, 2\right).
	$$
	In these expressions $k_0$ and $k_1$ are the maximum values of the suppressed terms of $\widetilde{O}_B$ over the region $R$, see \citet[Section 5.3]{BGT:2020} for additional details.
	
	Since, on every step of the subdivision algorithm, all of $C_0^{f_1}$, $C_0^{f_2}$, $C_1^{f_1}$, $C_1^{f_2}$, and $C_1^\times$ are applied, we let $h=2h_0+3h_1$.  Then, the computations in \citet{BGT:2020} lead directly to an overall adaptive bit-complexity bound.
	\begin{theorem}
		Let $f_1, f_2 \in \ZZ[x,y]$ be both smooth and $R \subseteq \RR^2$ be an axis-aligned square whose corners have integral coordinates. Let $\tau$ be the maximum bit-size of the coefficients of $f_1, $ the coefficients of $f_2,$ and the corners of $R.$ The overall bit-complexity of Algorithm \ref{alg:subdivision} is the maximum of $h(w(R)^n) $ and \[ 4\int_R \min\left\lbrace \frac{h\left( \frac{1}{4} G_{00}(x) \right)}{G_{00}(x)}, \frac{h\left( \frac{1}{4} G_{01}(x) \right)}{G_{01}(x)}, \frac{h\left( \frac{1}{4} G_{10}(x) \right)}{G_{10}(x)}, \frac{h\left( \frac{1}{4} G_{11}(x) \right)}{G_{11}(x)} \right\rbrace \, dA,  \] where $h(x) = 2h_0(x) + 3h_1(x).$
	\end{theorem}
	
	\begin{remark}
		We expect that the analyses for the Plantinga and Vegter algorithm in \citet{CET:2022} and \citet{TT:2020} can be adapted to provide the complexity for the average case, but leave this step to future work.
	\end{remark}
	
	\section{Approximating more than two curves}\label{sec:arbitrarycurves}
	
	We now consider the main problem of the paper, approximating any number of curves in the plane.  In particular, suppose that $f_1,\dots,f_n\in\ZZ[x,y]$ define $n$ smooth curves in the plane such that these curves intersect transversely and no three of these curves intersect simultaneously.  A first attempt to simultaneously approximate these curves might be first to run the subdivision step, Algorithm \ref{alg:subdivision}, on each pair of curves.  Second, perform the remaining steps of Algorithm \ref{alg:full} on the common refinement of these subdivisions.  Unfortunately, we cannot guarantee that the output of this algorithm is topologically correct since the isotopies guaranteed by Theorem \ref{thm:twocurves} might not be compatible.
	
	The main challenge with approximating more than two curves is that the ambient isotopies are not guaranteed to preserve the order of intersections between multiple curves.  In order to avoid this, we require that the intersections are far enough apart so that two intersections cannot interfere with each other through the homotopy.  By Propositions \ref{prop:plus_sign} and \ref{prop:snake}, all intersections of the varieties are within the 1-neighborhoods of a box containing an intersection of the approximations or a snake containing an intersection of the approximations.  The final step to make sure that the topology is correct is to ensure that these neighborhoods are all disjoint so that the ambient isotopies constructed by Theorem \ref{thm:twocurves} do not allow the intersections of the varieties to interact.
	
	\begin{algorithm}[hbt]
		\caption{Simultaneous approximation algorithm for any number of curves}
		\label{alg:morecurves}
		\begin{algorithmic}[1]
			\Require polynomials $f_1,f_2,\dots,f_n\in\ZZ[x,y]$, with $\cV(f_i)$ a smooth curve for all $i$ such that $\cV(f_i)$ and $\cV(f_j)$ intersect transversely when $i\not=j$, and no three of $\cV(f_i)$, $\cV(f_j)$, and $\cV(f_k)$ simultaneously intersect when $i$, $j$, and $k$ are all distinct.  A square region $R$ with integral corners such that $\cV(f_i)$ and $\cV(f_j)$ do not intersect on the boundary of $R$ for $i\not=j$.
			\Ensure approximations $\cA(f_1),\dots,\cA(f_n)$ such that 
			$(\cA(f_1),\dots,\cA(f_n))$ approximates $(\cV(f_1),\dots,\cV(f_n))$.
			\State For each pair $(f_i,f_j)$, with $i\not=j$, subdivide $R$ using Algorithm \ref{alg:subdivision}.
			\State Further subdivide boxes $B$ intersecting $\partial R$ until they satisfy $C_0^{f_i}(B)=\true$ for at least $n-1$ curves.
			\State Further subdivide boxes until the side lengths of neighboring boxes differ by at most a factor of two.
			\State Compute the Plantinga and Vegter curve approximations.
			\State Subdivide boxes which have shared approximations, but fail the condition of Definition \ref{def:snake}.
			\item Subdivide boxes or snakes with intersections of the approximations whose $1$-neighborhoods intersect.
			\State Recompute the Plantinga and Vegter curve approximations, if necessary.
			\State For any snakes, apply Lemma \ref{lem:snakecrossings}.
			\renewcommand{\algorithmicif}{\hspace{\algorithmicindent}\textbf{if}}
			\If{there is no crossing}
			\renewcommand{\algorithmicif}{\textbf{if}}
			\State \hspace{\algorithmicindent}slightly separate the edges of the snake so that the common edges do not overlap.
			\renewcommand{\algorithmicelse}{\hspace{\algorithmicindent}\textbf{else}}
			\Else\hspace{\algorithmicindent}
			\renewcommand{\algorithmicelse}{\textbf{else}}
			\State \hspace{\algorithmicindent} slightly separate the ends of the snake and add an explicit crossing in the middle of the snake.
			\EndIf
			\State\Return {the approximations}
		\end{algorithmic}
	\end{algorithm}
	
	\begin{theorem}\label{thm:morecurves}
		Suppose that $f_1,\dots,f_n\in\ZZ[x,y]$, $R=[a,b]\times[c,d]\subseteq\RR^2$ a rectangle such that $a,b,c,d\in\ZZ$.  Suppose also that $\cV(f_i)$ is nonsingular within $R$ for all $i$, there does not exist a pair $i\not=j$ such that $\cV(f_i,f_j)$ includes a point on the boundary of $R$, for every pair $i\not= j$, $\cV(f_i)$ and $\cV(f_j)$ intersect transversely within $R$, and there does not exist a triple $i\not=j$, $i\not=k$, and $j\not=k$ such that $\cV(f_i,f_j,f_k)$ includes a point in $R$.  Let $\cA(f_1),\dots,\cA(f_n)$ be the output of Algorithm \ref{alg:morecurves}. There exists a set $U$ containing $R$ and an ambient isotopy defined on $U$ that simultaneously takes $\cV(f_i)$ to $\cA(f_i)$ for all $i$.
	\end{theorem}
	\begin{proof}
		We begin with the same simplification of Theorem \ref{thm:twocurves}.  This reduces the problem to studying the case where $R$ is a square and allows us to focus on the topological correctness in the interior of $R$.
		
		We begin the isotopy by deforming space so that each excursion is moved into a neighboring box.  This can be done as no intersection can involve two excursions, as discussed in Theorem \ref{thm:twocurves}.  Therefore, it is possible to apply this isotopy without introducing or removing any intersections between the curves.  
		
		After applying this first isotopy, let $\widetilde{\cV}(f_i)$ be the result. By the correctness of the Plantinga and Vegter algorithm each $\widetilde{\cV}(f_i)\cap B$ is ambient isotopic to $\cA(f_i)\cap B$ within each box $B$.  We must now show that these isotopies are compatible.  By the assumption on the 1-neighborhoods of crossings of approximations and snakes and appealing to the proof of Theorem \ref{thm:twocurves}, we see that in each box $B$, there is at most one of the following two situations: 
		(1) a pair of $i\not=j$ such that $\widetilde{\cV}(f_i)\cap B$ and $\widetilde{\cV}(f_j)\cap B$ cross or (2) a pair of $i\not=j$ with endpoints of $\widetilde{\cV}(f_i)\cap B$ and $\widetilde{\cV}(f_j)\cap B$ which are on the same edge of $B$.  These two cases correspond to a crossing of the approximations and a snake, respectively.  
		
		By the isotopy argument in Theorem \ref{thm:twocurves}, if there is a pair of $i\not=j$ as identified above, these curves will deform to an intersection or a snake.  On the other hand, since the remaining curves do not intersect and are on different edges of the subdivision, the proof shows that they can be deformed, without introducing new intersections into the straight edges of the approximation.
	\end{proof}
	
	\section{Examples}\label{sec:examples}
	We end with a gallery of examples, illustrating the output of our algorithm.
	
	\begin{example}\label{ex:flower}
		Consider the pair of curves given by $f_1 = x^2 + y^2 - 1$ and $f_2 =  2000y^8 + 8000x^2y^6 + 12000x^4y^4 + 8000x^6y^2 + 2000x^8 - 3000y^6 + 9000x^2y^4 - 21000x^4y^2 - 1000x^6 + 1$, as seen in Figure \ref{fig:flower}(a).  Our algorithm correctly finds the twelve intersections between the two curves.  In addition, in Table \ref{table:flower}, we compare the number and types of boxes produced by our algorithm and what would be produced by applying the Plantinga and Vegter algorithm twice, that is, without invoking $C_1^\times$.  We observe that the new predicates only introduce a modest number of new boxes in this example.
		\begin{table}[htb]
			\begin{center}
				\begin{tabular}{|c|c|c|c|c|c|}\hline
					&$C_0^{f_1},C_0^{f_2}$&$C_0^{f_1},C_1^{f_2}$&$C_1^{f_1},C_0^{f_2}$&$C_1^{f_1},C_1^{f_2},C_1^\times$&$C_1^{f_1},C_1^{f_2}$\\\hline
					Algorithm \ref{alg:full}&$1087$&$489$&$168$&$24$&\\\hline
					Plantinga and Vegter&$852$&$340$&$60$&&$36$\\\hline
				\end{tabular}
			\end{center}
			\caption{The number of boxes produced by Algorithm \ref{alg:full} and the Plantinga and Vegter algorithm on the curves in Example \ref{ex:flower}.}
			\label{table:flower}
		\end{table}
		\begin{figure}[hbt]
			\centering
			\includegraphics{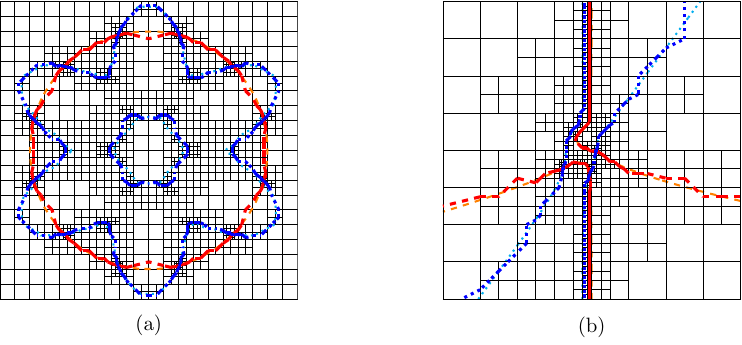}
			\caption{(a) Simultaneous approximation of $f_1 = x^2 + y^2 - 1$ (drawn in red) and $f_2 =  2000y^8 + 8000x^2y^6 + 12000x^4y^4 + 8000x^6y^2 + 2000x^8 - 3000y^6 + 9000x^2y^4 - 21000x^4y^2 - 1000x^6 + 1$ (drawn in blue). This example is inspired from a similar example in \citep{CWZ:2023}.  (b) Simultaneous approximation of $f_1 = 10xy^3 - x^2y^2 -6xy^2 + x^2y + 5x^3 + x^2 + 3y$ and $f_2 = 4xy^3 - x^2y^2 - 8x^4 + 2y^3 - 8x^3 + 4xy - 2x^2 - 7x + 2$. The predicate $C_1^\times$ induces intersections because the two curves are nearly parallel.}
			\label{fig:flower}
		\end{figure}  
	\end{example}
	
	\begin{example}\label{ex:cool_quartic}
		Consider the pair of curves $f_1 = 2y^4 + 3x^2y^2 + 8x^3y - 12y^2 - 18xy - 12x^2 + 18$ and $f_2 = y^4 + 4xy^3 + 3x^2y^2  + 4x^3y +  x^4 - 12y^2 - 18xy - 12x^2 + 18$, as seen in Figure \ref{fig:cool_quartic}.  There is a triple intersection of the two curves along each nearly parallel portion of the curves.  Our algorithm correctly finds and approximates all of these intersections.  In addition, in Table \ref{table:cool_quartic}, we compare the number and types of boxes produced by our algorithm and what would be produced by applying the Plantinga and Vegter algorithm twice, that is, without invoking $C_1^\times$.  We observe that the new predicates produce significantly more boxes due to the work needed to process the nearly parallel curves.
		\begin{table}[htb]
			\begin{center}
				\begin{tabular}{|c|c|c|c|c|c|}\hline
					&$C_0^{f_1},C_0^{f_2}$&$C_0^{f_1},C_1^{f_2}$&$C_1^{f_1},C_0^{f_2}$&$C_1^{f_1},C_1^{f_2},C_1^\times$&$C_1^{f_1},C_1^{f_2}$\\\hline
					Algorithm \ref{alg:full}&$1252$&$590$&$584$&$140$&\\\hline
					Plantinga and Vegter&$202$&$86$&$58$&&$72$\\\hline
				\end{tabular}
			\end{center}
			\caption{The number of boxes produced by Algorithm \ref{alg:full} and the Plantinga and Vegter algorithm on the curves in Example \ref{ex:cool_quartic}.}
			\label{table:cool_quartic}
		\end{table}
		\begin{figure}[htb]
			\centering
			\includegraphics{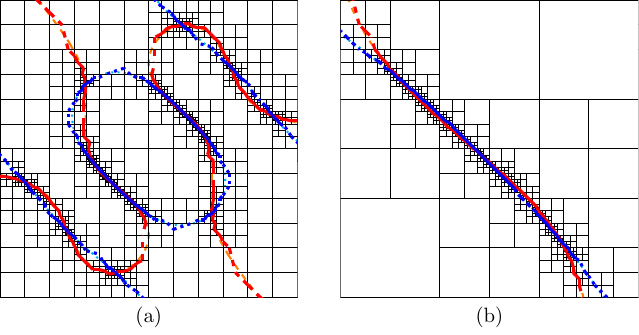}
			\caption{(a) Simultaneous approximation of $f_1 = 2y^4 + 3x^2y^2 + 8x^3y - 12y^2 - 18xy - 12x^2 + 18$ and $f_2 = y^4 + 4xy^3 + 3x^2y^2  + 4x^3y +  x^4 - 12y^2 - 18xy - 12x^2 + 18$. This example has several close curves, each with three intersections.  (b) An enlarged image of a portion of (a), illustrating the three intersections along one of the nearly parallel portions of the curves.}
			\label{fig:cool_quartic}
		\end{figure}
	\end{example}
	
	\begin{example}\label{ex:random_example}
		Consider the curves $f_1 = 10xy^3 - x^2y^2 -6xy^2 + x^2y + 5x^3 + x^2 + 3y$ and $f_2 = 4xy^3 - x^2y^2 - 8x^4 + 2y^3 - 8x^3 + 4xy - 2x^2 - 7x + 2$, as seen in Figure \ref{fig:flower}(b).  The two curves share an asymptote, and our algorithm correctly separates the two curves as they become nearly parallel.  In addition, in Table \ref{table:random_example}, we compare the number and types of boxes produced by our algorithm and what would be produced by applying the Plantinga and Vegter algorithm twice, that is, without invoking $C_1^\times$.  We observe that the new predicates produce significantly more boxes due to the work needed to separate the curves as they become nearly parallel.
		\begin{table}[htb]
			\begin{center}
				\begin{tabular}{|c|c|c|c|c|c|}\hline
					&$C_0^{f_1},C_0^{f_2}$&$C_0^{f_1},C_1^{f_2}$&$C_1^{f_1},C_0^{f_2}$&$C_1^{f_1},C_1^{f_2},C_1^\times$&$C_1^{f_1},C_1^{f_2}$\\\hline
					Algorithm \ref{alg:full}&$310$&$186$&$171$&$6$&\\\hline
					Plantinga and Vegter&$188$&$71$&$69$&&$39$\\\hline
				\end{tabular}
			\end{center}
			\caption{The number of boxes produced by Algorithm \ref{alg:full} and the Plantinga and Vegter algorithm on the curves in Example \ref{ex:random_example}.}
			\label{table:random_example}
		\end{table}

	\end{example}
	\begin{example}
		Consider the curves $f_1 = 16y^2 + 16x^2 - 24y - 16x - 3$, $f_2 = 16y^2 + 16x^2 - 24y + 16x - 3$ and $f_3 = 100y^2 + 100x^2 + 180y - 19$.  The output of Algorithm \ref{alg:morecurves} for these three curves appears in Figure \ref{fig:threecurves}.  Even though there are three nearby intersections, our algorithm correctly approximates both their locations and correctly computes the orders in which the curves intersect each other.
		
		\begin{figure}[hbt]
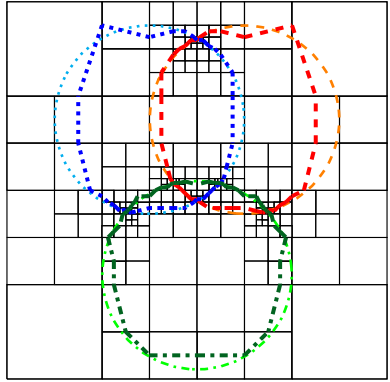

			\centering
			\include{figures_source/three_curve}
			\caption{Simultaneous approximation of $f_1 = 16y^2 + 16x^2 - 24y - 16x - 3$, $f_2 = 16y^2 + 16x^2 - 24y + 16x - 3$ and $f_3 = 100y^2 + 100x^2 + 180y - 19$. For approximating more than two curves, we see in this example that the intersections from distinct pairs are in separate neighborhoods.}
			\label{fig:threecurves}
		\end{figure}
	\end{example}
	\bibliographystyle{plainnat}
	\bibliography{approximation}
\end{document}

%% file: figures_source/missing_intersection_before.tex
\begin{tikzpicture}[scale = 1.25 ] 

\draw[thin, black] (2.0, 2.0) -- (4, 2.0); 
\draw[thin, black] (4, 2.0) -- (4, 4); 
\draw[thin, black] (4, 4) -- (2.0, 4); 
\draw[thin, black] (2.0, 4) -- (2.0, 2.0); 
\draw [] (2.0, 2.0) rectangle (4, 4);

\draw[thin, black] (0, 0) -- (2.0, 0); 
\draw[thin, black] (2.0, 0) -- (2.0, 2.0); 
\draw[thin, black] (2.0, 2.0) -- (0, 2.0); 
\draw[thin, black] (0, 2.0) -- (0, 0); 
\draw [] (0, 0) rectangle (2.0, 2.0);

\draw[thin, black] (0, 2.0) -- (2.0, 2.0); 
\draw[thin, black] (2.0, 2.0) -- (2.0, 4); 
\draw[thin, black] (2.0, 4) -- (0, 4); 
\draw[thin, black] (0, 4) -- (0, 2.0); 
\draw [] (0, 2.0) rectangle (2.0, 4);

\draw[thin, black] (2.0, 0) -- (4, 0); 
\draw[thin, black] (4, 0) -- (4, 2.0); 
\draw[thin, black] (4, 2.0) -- (2.0, 2.0); 
\draw[thin, black] (2.0, 2.0) -- (2.0, 0); 
\draw [] (2.0, 0) rectangle (4, 2.0);

\draw[red, line width = 1.5pt, dashed] (2.0, 3.0) -- (3.0, 4.0); 

\draw[red, line width = 1.5pt, dashed] (2.0, 3.0) -- (0.0, 3.0); 

\draw[blue, line width = 1.5pt, dotted] (2.0, 1.0) -- (0.0, 1.0); 

\draw[blue, line width = 1.5pt, dotted] (2.0, 1.0) -- (4.0, 1.0); 

\draw[color=orange, line width = 1pt, dashed] plot[id=fmissingred, raw gnuplot, smooth] function{
	f(x,y) =  -0.54*x + 0.3*x**2 + 2.143 - y;
	set xrange [0:4];
	set yrange [0:4];
	set view 0,0;
	set isosample 1000,1000;
	set size square;
	set cont base;
	set cntrparam levels incre 0,0.1,0;
	unset surface;
	splot f(x,y)};

	\draw[color=cyan, line width = 1pt, dotted] plot[id=fmissingblue, raw gnuplot, smooth] function{
	g(x,y) =  0.44*x - 0.2*x**2 + 1.858 - y;
	set xrange [0:4];
	set yrange [0:4];
	set view 0,0;
	set isosample 1000,1000;
	set size square;
	set cont base;
	set cntrparam levels incre 0,0.1,0;
	unset surface;
	splot g(x,y)};

\end{tikzpicture}

%% file: figures_source/missing_intersection_after.tex
\begin{tikzpicture}[scale = 1.25 ] 
\draw[color=orange, line width = 1, dashed] plot[id=fmissingred, raw gnuplot, smooth] function{
	f(x,y) =  -0.54*x + 0.3*x**2 + 2.143 - y;
	set xrange [0:4];
	set yrange [0:4];
	set view 0,0;
	set isosample 1000,1000;
	set size square;
	set cont base;
	set cntrparam levels incre 0,0.1,0;
	unset surface;
	splot f(x,y)};

	\draw[color=cyan, line width = 1, dotted] plot[id=fmissingblue, raw gnuplot, smooth] function{
	g(x,y) =  0.44*x - 0.2*x**2 + 1.858 - y;
	set xrange [0:4];
	set yrange [0:4];
	set view 0,0;
	set isosample 1000,1000;
	set size square;
	set cont base;
	set cntrparam levels incre 0,0.1,0;
	unset surface;
	splot g(x,y)};

\draw[thin, black] (3.0, 2.0) -- (4, 2.0) -- (4, 3.0) -- (3.0, 3.0) -- (3.0, 2.0); 
 \draw [fill = white, opacity = .3] (3.0, 2.0) rectangle (4, 3.0);

\draw[thin, black] (1.0, 3.0) -- (2.0, 3.0) -- (2.0, 4) -- (1.0, 4) -- (1.0, 3.0); 
 \draw [fill = white, opacity = .3] (1.0, 3.0) rectangle (2.0, 4);

\draw[thin, black] (0, 3.0) -- (1.0, 3.0) -- (1.0, 4) -- (0, 4) -- (0, 3.0); 
 \draw [fill = white, opacity = .3] (0, 3.0) rectangle (1.0, 4);

\draw[thin, black] (0, 0) -- (1.0, 0) -- (1.0, 1.0) -- (0, 1.0) -- (0, 0); 
 \draw [fill = white, opacity = .3] (0, 0) rectangle (1.0, 1.0);

\draw[thin, black] (1.0, 0) -- (2.0, 0) -- (2.0, 1.0) -- (1.0, 1.0) -- (1.0, 0); 
 \draw [fill = white, opacity = .3] (1.0, 0) rectangle (2.0, 1.0);

\draw[thin, black] (2.0, 0) -- (3.0, 0) -- (3.0, 1.0) -- (2.0, 1.0) -- (2.0, 0); 
 \draw [fill = white, opacity = .3] (2.0, 0) rectangle (3.0, 1.0);

\draw[thin, black] (0.5, 2.5) -- (1.0, 2.5) -- (1.0, 3.0) -- (0.5, 3.0) -- (0.5, 2.5); 
 \draw [fill = white, opacity = .3] (0.5, 2.5) rectangle (1.0, 3.0);

\draw[thin, black] (0, 2.5) -- (0.5, 2.5) -- (0.5, 3.0) -- (0, 3.0) -- (0, 2.5); 
 \draw [fill = white, opacity = .3] (0, 2.5) rectangle (0.5, 3.0);

\draw[thin, black] (1.5, 2.5) -- (2.0, 2.5) -- (2.0, 3.0) -- (1.5, 3.0) -- (1.5, 2.5); 
 \draw [fill = white, opacity = .3] (1.5, 2.5) rectangle (2.0, 3.0);

\draw[thin, black] (1.0, 2.5) -- (1.5, 2.5) -- (1.5, 3.0) -- (1.0, 3.0) -- (1.0, 2.5); 
 \draw [fill = white, opacity = .3] (1.0, 2.5) rectangle (1.5, 3.0);

\draw[thin, black] (1.0, 1.0) -- (1.5, 1.0) -- (1.5, 1.5) -- (1.0, 1.5) -- (1.0, 1.0); 
 \draw [fill = white, opacity = .3] (1.0, 1.0) rectangle (1.5, 1.5);

\draw[thin, black] (1.5, 1.0) -- (2.0, 1.0) -- (2.0, 1.5) -- (1.5, 1.5) -- (1.5, 1.0); 
 \draw [fill = white, opacity = .3] (1.5, 1.0) rectangle (2.0, 1.5);

\draw[thin, black] (0, 1.0) -- (0.5, 1.0) -- (0.5, 1.5) -- (0, 1.5) -- (0, 1.0); 
 \draw [fill = white, opacity = .3] (0, 1.0) rectangle (0.5, 1.5);

\draw[thin, black] (0.5, 1.0) -- (1.0, 1.0) -- (1.0, 1.5) -- (0.5, 1.5) -- (0.5, 1.0); 
 \draw [fill = white, opacity = .3] (0.5, 1.0) rectangle (1.0, 1.5);

\draw[thin, black] (0.25, 2.25) -- (0.5, 2.25) -- (0.5, 2.5) -- (0.25, 2.5) -- (0.25, 2.25); 
 \draw [fill = white, opacity = .3] (0.25, 2.25) rectangle (0.5, 2.5);

\draw[thin, black] (0, 2.25) -- (0.25, 2.25) -- (0.25, 2.5) -- (0, 2.5) -- (0, 2.25); 
 \draw [fill = white, opacity = .3] (0, 2.25) rectangle (0.25, 2.5);

\draw[thin, black] (1.25, 2.25) -- (1.5, 2.25) -- (1.5, 2.5) -- (1.25, 2.5) -- (1.25, 2.25); 
 \draw [fill = white, opacity = .3] (1.25, 2.25) rectangle (1.5, 2.5);

\draw[thin, black] (1.0, 2.25) -- (1.25, 2.25) -- (1.25, 2.5) -- (1.0, 2.5) -- (1.0, 2.25); 
 \draw [fill = white, opacity = .3] (1.0, 2.25) rectangle (1.25, 2.5);

\draw[thin, black] (1.5, 2.25) -- (1.75, 2.25) -- (1.75, 2.5) -- (1.5, 2.5) -- (1.5, 2.25); 
 \draw [fill = white, opacity = .3] (1.5, 2.25) rectangle (1.75, 2.5);

\draw[thin, black] (1.5, 1.75) -- (1.75, 1.75) -- (1.75, 2.0) -- (1.5, 2.0) -- (1.5, 1.75); 
 \draw [fill = white, opacity = .3] (1.5, 1.75) rectangle (1.75, 2.0);

\draw[thin, black] (1.5, 1.5) -- (1.75, 1.5) -- (1.75, 1.75) -- (1.5, 1.75) -- (1.5, 1.5); 
 \draw [fill = white, opacity = .3] (1.5, 1.5) rectangle (1.75, 1.75);

\draw[thin, black] (1.75, 1.5) -- (2.0, 1.5) -- (2.0, 1.75) -- (1.75, 1.75) -- (1.75, 1.5); 
 \draw [fill = white, opacity = .3] (1.75, 1.5) rectangle (2.0, 1.75);

\draw[thin, black] (0, 1.5) -- (0.25, 1.5) -- (0.25, 1.75) -- (0, 1.75) -- (0, 1.5); 
 \draw [fill = white, opacity = .3] (0, 1.5) rectangle (0.25, 1.75);

\draw[thin, black] (0.25, 1.5) -- (0.5, 1.5) -- (0.5, 1.75) -- (0.25, 1.75) -- (0.25, 1.5); 
 \draw [fill = white, opacity = .3] (0.25, 1.5) rectangle (0.5, 1.75);

\draw[thin, black] (0.375, 2.125) -- (0.5, 2.125) -- (0.5, 2.25) -- (0.375, 2.25) -- (0.375, 2.125); 
 \draw [fill = white, opacity = .3] (0.375, 2.125) rectangle (0.5, 2.25);

\draw[thin, black] (0.25, 2.125) -- (0.375, 2.125) -- (0.375, 2.25) -- (0.25, 2.25) -- (0.25, 2.125); 
 \draw [fill = white, opacity = .3] (0.25, 2.125) rectangle (0.375, 2.25);

\draw[thin, black] (1.375, 2.125) -- (1.5, 2.125) -- (1.5, 2.25) -- (1.375, 2.25) -- (1.375, 2.125); 
 \draw [fill = white, opacity = .3] (1.375, 2.125) rectangle (1.5, 2.25);

\draw[thin, black] (1.25, 2.125) -- (1.375, 2.125) -- (1.375, 2.25) -- (1.25, 2.25) -- (1.25, 2.125); 
 \draw [fill = white, opacity = .3] (1.25, 2.125) rectangle (1.375, 2.25);

\draw[thin, black] (1.625, 2.125) -- (1.75, 2.125) -- (1.75, 2.25) -- (1.625, 2.25) -- (1.625, 2.125); 
 \draw [fill = white, opacity = .3] (1.625, 2.125) rectangle (1.75, 2.25);

\draw[thin, black] (1.5, 2.125) -- (1.625, 2.125) -- (1.625, 2.25) -- (1.5, 2.25) -- (1.5, 2.125); 
 \draw [fill = white, opacity = .3] (1.5, 2.125) rectangle (1.625, 2.25);

\draw[thin, black] (1.875, 2.0) -- (2.0, 2.0) -- (2.0, 2.125) -- (1.875, 2.125) -- (1.875, 2.0); 
 \draw [fill = white, opacity = .3] (1.875, 2.0) rectangle (2.0, 2.125);

\draw[thin, black] (0.25, 1.75) -- (0.375, 1.75) -- (0.375, 1.875) -- (0.25, 1.875) -- (0.25, 1.75); 
 \draw [fill = white, opacity = .3] (0.25, 1.75) rectangle (0.375, 1.875);

\draw[thin, black] (0.375, 1.75) -- (0.5, 1.75) -- (0.5, 1.875) -- (0.375, 1.875) -- (0.375, 1.75); 
 \draw [fill = white, opacity = .3] (0.375, 1.75) rectangle (0.5, 1.875);

\draw[thin, black] (1.375, 2.0) -- (1.4375, 2.0) -- (1.4375, 2.0625) -- (1.375, 2.0625) -- (1.375, 2.0); 
 \draw [fill = white, opacity = .3] (1.375, 2.0) rectangle (1.4375, 2.0625);

\draw[thin, black] (1.5625, 2.0625) -- (1.625, 2.0625) -- (1.625, 2.125) -- (1.5625, 2.125) -- (1.5625, 2.0625); 
 \draw [fill = white, opacity = .3] (1.5625, 2.0625) rectangle (1.625, 2.125);

\draw[thin, black] (1.8125, 2.0625) -- (1.875, 2.0625) -- (1.875, 2.125) -- (1.8125, 2.125) -- (1.8125, 2.0625); 
 \draw [fill = white, opacity = .3] (1.8125, 2.0625) rectangle (1.875, 2.125);

\draw[thin, black] (1.8125, 2.0) -- (1.875, 2.0) -- (1.875, 2.0625) -- (1.8125, 2.0625) -- (1.8125, 2.0); 
 \draw [fill = white, opacity = .3] (1.8125, 2.0) rectangle (1.875, 2.0625);

\draw[thin, black] (0.375, 1.875) -- (0.4375, 1.875) -- (0.4375, 1.9375) -- (0.375, 1.9375) -- (0.375, 1.875); 
 \draw [fill = white, opacity = .3] (0.375, 1.875) rectangle (0.4375, 1.9375);

\draw[thin, black] (0.4375, 1.875) -- (0.5, 1.875) -- (0.5, 1.9375) -- (0.4375, 1.9375) -- (0.4375, 1.875); 
 \draw [fill = white, opacity = .3] (0.4375, 1.875) rectangle (0.5, 1.9375);

\draw[thin, black] (0.25, 1.875) -- (0.3125, 1.875) -- (0.3125, 1.9375) -- (0.25, 1.9375) -- (0.25, 1.875); 
 \draw [fill = white, opacity = .3] (0.25, 1.875) rectangle (0.3125, 1.9375);

\draw[thin, black] (0.3125, 1.875) -- (0.375, 1.875) -- (0.375, 1.9375) -- (0.3125, 1.9375) -- (0.3125, 1.875); 
 \draw [fill = white, opacity = .3] (0.3125, 1.875) rectangle (0.375, 1.9375);

\draw[thin, black] (1.75, 1.75) -- (1.875, 1.75) -- (1.875, 1.875) -- (1.75, 1.875) -- (1.75, 1.75); 
 \draw [fill = white, opacity = .3] (1.75, 1.75) rectangle (1.875, 1.875);

\draw[thin, black] (1.875, 1.75) -- (2.0, 1.75) -- (2.0, 1.875) -- (1.875, 1.875) -- (1.875, 1.75); 
 \draw [fill = white, opacity = .3] (1.875, 1.75) rectangle (2.0, 1.875);

\draw[thin, black] (0.5, 1.5) -- (0.75, 1.5) -- (0.75, 1.75) -- (0.5, 1.75) -- (0.5, 1.5); 
 \draw [fill = white, opacity = .3] (0.5, 1.5) rectangle (0.75, 1.75);

\draw[thin, black] (0.75, 1.5) -- (1.0, 1.5) -- (1.0, 1.75) -- (0.75, 1.75) -- (0.75, 1.5); 
 \draw [fill = white, opacity = .3] (0.75, 1.5) rectangle (1.0, 1.75);

\draw[thin, black] (0.5, 1.75) -- (0.625, 1.75) -- (0.625, 1.875) -- (0.5, 1.875) -- (0.5, 1.75); 
 \draw [fill = white, opacity = .3] (0.5, 1.75) rectangle (0.625, 1.875);

\draw[thin, black] (0.625, 1.75) -- (0.75, 1.75) -- (0.75, 1.875) -- (0.625, 1.875) -- (0.625, 1.75); 
 \draw [fill = white, opacity = .3] (0.625, 1.75) rectangle (0.75, 1.875);

\draw[thin, black] (1.0, 1.5) -- (1.25, 1.5) -- (1.25, 1.75) -- (1.0, 1.75) -- (1.0, 1.5); 
 \draw [fill = white, opacity = .3] (1.0, 1.5) rectangle (1.25, 1.75);

\draw[thin, black] (1.25, 1.5) -- (1.5, 1.5) -- (1.5, 1.75) -- (1.25, 1.75) -- (1.25, 1.5); 
 \draw [fill = white, opacity = .3] (1.25, 1.5) rectangle (1.5, 1.75);

\draw[thin, black] (1.25, 1.75) -- (1.375, 1.75) -- (1.375, 1.875) -- (1.25, 1.875) -- (1.25, 1.75); 
 \draw [fill = white, opacity = .3] (1.25, 1.75) rectangle (1.375, 1.875);

\draw[thin, black] (1.375, 1.75) -- (1.5, 1.75) -- (1.5, 1.875) -- (1.375, 1.875) -- (1.375, 1.75); 
 \draw [fill = white, opacity = .3] (1.375, 1.75) rectangle (1.5, 1.875);

\draw[thin, black] (0.125, 1.75) -- (0.25, 1.75) -- (0.25, 1.875) -- (0.125, 1.875) -- (0.125, 1.75); 
 \draw [fill = white, opacity = .3] (0.125, 1.75) rectangle (0.25, 1.875);

\draw[thin, black] (2.5, 2.0) -- (3.0, 2.0) -- (3.0, 2.5) -- (2.5, 2.5) -- (2.5, 2.0); 
 \draw [fill = white, opacity = .3] (2.5, 2.0) rectangle (3.0, 2.5);

\draw[thin, black] (2.25, 2.0) -- (2.5, 2.0) -- (2.5, 2.25) -- (2.25, 2.25) -- (2.25, 2.0); 
 \draw [fill = white, opacity = .3] (2.25, 2.0) rectangle (2.5, 2.25);

\draw[thin, black] (2.0, 1.0) -- (2.5, 1.0) -- (2.5, 1.5) -- (2.0, 1.5) -- (2.0, 1.0); 
 \draw [fill = white, opacity = .3] (2.0, 1.0) rectangle (2.5, 1.5);

\draw[thin, black] (2.0, 1.5) -- (2.25, 1.5) -- (2.25, 1.75) -- (2.0, 1.75) -- (2.0, 1.5); 
 \draw [fill = white, opacity = .3] (2.0, 1.5) rectangle (2.25, 1.75);

\draw[thin, black] (0.75, 2.25) -- (1.0, 2.25) -- (1.0, 2.5) -- (0.75, 2.5) -- (0.75, 2.25); 
 \draw [fill = white, opacity = .3] (0.75, 2.25) rectangle (1.0, 2.5);

\draw[thin, black] (0.5, 2.25) -- (0.75, 2.25) -- (0.75, 2.5) -- (0.5, 2.5) -- (0.5, 2.25); 
 \draw [fill = white, opacity = .3] (0.5, 2.25) rectangle (0.75, 2.5);

\draw[thin, black] (1.4375, 2.0625) -- (1.5, 2.0625) -- (1.5, 2.125) -- (1.4375, 2.125) -- (1.4375, 2.0625); 
 \draw [fill = white, opacity = .3] (1.4375, 2.0625) rectangle (1.5, 2.125);

\draw[thin, black] (1.4375, 2.0625) -- (1.5, 2.0625) -- (1.5, 2.125) -- (1.4375, 2.125) -- (1.4375, 2.0625); 
 \draw [fill = white, opacity = .3] (1.4375, 2.0625) rectangle (1.5, 2.125);

\draw[thin, black] (1.5, 2.0) -- (1.5625, 2.0) -- (1.5625, 2.0625) -- (1.5, 2.0625) -- (1.5, 2.0); 
 \draw [fill = white, opacity = .3] (1.5, 2.0) rectangle (1.5625, 2.0625);

\draw[thin, black] (1.5, 2.0) -- (1.5625, 2.0) -- (1.5625, 2.0625) -- (1.5, 2.0625) -- (1.5, 2.0); 
 \draw [fill = white, opacity = .3] (1.5, 2.0) rectangle (1.5625, 2.0625);

\draw[thin, black] (1.6875, 2.0625) -- (1.75, 2.0625) -- (1.75, 2.125) -- (1.6875, 2.125) -- (1.6875, 2.0625); 
 \draw [fill = white, opacity = .3] (1.6875, 2.0625) rectangle (1.75, 2.125);

\draw[thin, black] (1.6875, 2.0625) -- (1.75, 2.0625) -- (1.75, 2.125) -- (1.6875, 2.125) -- (1.6875, 2.0625); 
 \draw [fill = white, opacity = .3] (1.6875, 2.0625) rectangle (1.75, 2.125);

\draw[thin, black] (1.625, 2.0) -- (1.6875, 2.0) -- (1.6875, 2.0625) -- (1.625, 2.0625) -- (1.625, 2.0); 
 \draw [fill = white, opacity = .3] (1.625, 2.0) rectangle (1.6875, 2.0625);

\draw[thin, black] (1.625, 2.0) -- (1.6875, 2.0) -- (1.6875, 2.0625) -- (1.625, 2.0625) -- (1.625, 2.0); 
 \draw [fill = white, opacity = .3] (1.625, 2.0) rectangle (1.6875, 2.0625);

\draw[thin, black] (1.75, 2.0) -- (1.8125, 2.0) -- (1.8125, 2.0625) -- (1.75, 2.0625) -- (1.75, 2.0); 
 \draw [fill = white, opacity = .3] (1.75, 2.0) rectangle (1.8125, 2.0625);

\draw[thin, black] (1.75, 2.0) -- (1.8125, 2.0) -- (1.8125, 2.0625) -- (1.75, 2.0625) -- (1.75, 2.0); 
 \draw [fill = white, opacity = .3] (1.75, 2.0) rectangle (1.8125, 2.0625);

\draw[thin, black] (0.4375, 1.9375) -- (0.5, 1.9375) -- (0.5, 2.0) -- (0.4375, 2.0) -- (0.4375, 1.9375); 
 \draw [fill = white, opacity = .3] (0.4375, 1.9375) rectangle (0.5, 2.0);

\draw[thin, black] (0.4375, 1.9375) -- (0.5, 1.9375) -- (0.5, 2.0) -- (0.4375, 2.0) -- (0.4375, 1.9375); 
 \draw [fill = white, opacity = .3] (0.4375, 1.9375) rectangle (0.5, 2.0);

\draw[thin, black] (0.3125, 1.9375) -- (0.375, 1.9375) -- (0.375, 2.0) -- (0.3125, 2.0) -- (0.3125, 1.9375); 
 \draw [fill = white, opacity = .3] (0.3125, 1.9375) rectangle (0.375, 2.0);

\draw[thin, black] (0.3125, 1.9375) -- (0.375, 1.9375) -- (0.375, 2.0) -- (0.3125, 2.0) -- (0.3125, 1.9375); 
 \draw [fill = white, opacity = .3] (0.3125, 1.9375) rectangle (0.375, 2.0);

\draw[thin, black] (1.375, 2.0625) -- (1.4375, 2.0625) -- (1.4375, 2.125) -- (1.375, 2.125) -- (1.375, 2.0625); 
 \draw [fill = white, opacity = .3] (1.375, 2.0625) rectangle (1.4375, 2.125);

\draw[thin, black] (1.375, 2.0625) -- (1.4375, 2.0625) -- (1.4375, 2.125) -- (1.375, 2.125) -- (1.375, 2.0625); 
 \draw [fill = white, opacity = .3] (1.375, 2.0625) rectangle (1.4375, 2.125);

\draw[thin, black] (1.4375, 2.0) -- (1.5, 2.0) -- (1.5, 2.0625) -- (1.4375, 2.0625) -- (1.4375, 2.0); 
 \draw [fill = white, opacity = .3] (1.4375, 2.0) rectangle (1.5, 2.0625);

\draw[thin, black] (1.4375, 2.0) -- (1.5, 2.0) -- (1.5, 2.0625) -- (1.4375, 2.0625) -- (1.4375, 2.0); 
 \draw [fill = white, opacity = .3] (1.4375, 2.0) rectangle (1.5, 2.0625);

\draw[thin, black] (1.5, 2.0625) -- (1.5625, 2.0625) -- (1.5625, 2.125) -- (1.5, 2.125) -- (1.5, 2.0625); 
 \draw [fill = white, opacity = .3] (1.5, 2.0625) rectangle (1.5625, 2.125);

\draw[thin, black] (1.5, 2.0625) -- (1.5625, 2.0625) -- (1.5625, 2.125) -- (1.5, 2.125) -- (1.5, 2.0625); 
 \draw [fill = white, opacity = .3] (1.5, 2.0625) rectangle (1.5625, 2.125);

\draw[thin, black] (1.5625, 2.0) -- (1.625, 2.0) -- (1.625, 2.0625) -- (1.5625, 2.0625) -- (1.5625, 2.0); 
 \draw [fill = white, opacity = .3] (1.5625, 2.0) rectangle (1.625, 2.0625);

\draw[thin, black] (1.5625, 2.0) -- (1.625, 2.0) -- (1.625, 2.0625) -- (1.5625, 2.0625) -- (1.5625, 2.0); 
 \draw [fill = white, opacity = .3] (1.5625, 2.0) rectangle (1.625, 2.0625);

\draw[thin, black] (1.625, 2.0625) -- (1.6875, 2.0625) -- (1.6875, 2.125) -- (1.625, 2.125) -- (1.625, 2.0625); 
 \draw [fill = white, opacity = .3] (1.625, 2.0625) rectangle (1.6875, 2.125);

\draw[thin, black] (1.625, 2.0625) -- (1.6875, 2.0625) -- (1.6875, 2.125) -- (1.625, 2.125) -- (1.625, 2.0625); 
 \draw [fill = white, opacity = .3] (1.625, 2.0625) rectangle (1.6875, 2.125);

\draw[thin, black] (1.75, 2.0625) -- (1.8125, 2.0625) -- (1.8125, 2.125) -- (1.75, 2.125) -- (1.75, 2.0625); 
 \draw [fill = white, opacity = .3] (1.75, 2.0625) rectangle (1.8125, 2.125);

\draw[thin, black] (1.75, 2.0625) -- (1.8125, 2.0625) -- (1.8125, 2.125) -- (1.75, 2.125) -- (1.75, 2.0625); 
 \draw [fill = white, opacity = .3] (1.75, 2.0625) rectangle (1.8125, 2.125);

\draw[thin, black] (1.6875, 2.0) -- (1.75, 2.0) -- (1.75, 2.0625) -- (1.6875, 2.0625) -- (1.6875, 2.0); 
 \draw [fill = white, opacity = .3] (1.6875, 2.0) rectangle (1.75, 2.0625);

\draw[thin, black] (1.6875, 2.0) -- (1.75, 2.0) -- (1.75, 2.0625) -- (1.6875, 2.0625) -- (1.6875, 2.0); 
 \draw [fill = white, opacity = .3] (1.6875, 2.0) rectangle (1.75, 2.0625);

\draw[thin, black] (0.25, 1.9375) -- (0.3125, 1.9375) -- (0.3125, 2.0) -- (0.25, 2.0) -- (0.25, 1.9375); 
 \draw [fill = white, opacity = .3] (0.25, 1.9375) rectangle (0.3125, 2.0);

\draw[thin, black] (0.25, 1.9375) -- (0.3125, 1.9375) -- (0.3125, 2.0) -- (0.25, 2.0) -- (0.25, 1.9375); 
 \draw [fill = white, opacity = .3] (0.25, 1.9375) rectangle (0.3125, 2.0);

\draw[thin, black] (0.375, 1.9375) -- (0.4375, 1.9375) -- (0.4375, 2.0) -- (0.375, 2.0) -- (0.375, 1.9375); 
 \draw [fill = white, opacity = .3] (0.375, 1.9375) rectangle (0.4375, 2.0);

\draw[thin, black] (0.375, 1.9375) -- (0.4375, 1.9375) -- (0.4375, 2.0) -- (0.375, 2.0) -- (0.375, 1.9375); 
 \draw [fill = white, opacity = .3] (0.375, 1.9375) rectangle (0.4375, 2.0);

\draw[thin, black] (1.875, 2.125) -- (2.0, 2.125) -- (2.0, 2.25) -- (1.875, 2.25) -- (1.875, 2.125); 
 \draw [fill = white, opacity = .3] (1.875, 2.125) rectangle (2.0, 2.25);

\draw[thin, black] (1.875, 2.125) -- (2.0, 2.125) -- (2.0, 2.25) -- (1.875, 2.25) -- (1.875, 2.125); 
 \draw [fill = white, opacity = .3] (1.875, 2.125) rectangle (2.0, 2.25);

\draw[thin, black] (1.875, 1.875) -- (2.0, 1.875) -- (2.0, 2.0) -- (1.875, 2.0) -- (1.875, 1.875); 
 \draw [fill = white, opacity = .3] (1.875, 1.875) rectangle (2.0, 2.0);

\draw[thin, black] (1.875, 1.875) -- (2.0, 1.875) -- (2.0, 2.0) -- (1.875, 2.0) -- (1.875, 1.875); 
 \draw [fill = white, opacity = .3] (1.875, 1.875) rectangle (2.0, 2.0);

\draw[thin, black] (0.625, 1.875) -- (0.75, 1.875) -- (0.75, 2.0) -- (0.625, 2.0) -- (0.625, 1.875); 
 \draw [fill = white, opacity = .3] (0.625, 1.875) rectangle (0.75, 2.0);

\draw[thin, black] (0.625, 1.875) -- (0.75, 1.875) -- (0.75, 2.0) -- (0.625, 2.0) -- (0.625, 1.875); 
 \draw [fill = white, opacity = .3] (0.625, 1.875) rectangle (0.75, 2.0);

\draw[thin, black] (1.25, 2.0) -- (1.375, 2.0) -- (1.375, 2.125) -- (1.25, 2.125) -- (1.25, 2.0); 
 \draw [fill = white, opacity = .3] (1.25, 2.0) rectangle (1.375, 2.125);

\draw[thin, black] (1.25, 2.0) -- (1.375, 2.0) -- (1.375, 2.125) -- (1.25, 2.125) -- (1.25, 2.0); 
 \draw [fill = white, opacity = .3] (1.25, 2.0) rectangle (1.375, 2.125);

\draw[thin, black] (1.375, 1.875) -- (1.5, 1.875) -- (1.5, 2.0) -- (1.375, 2.0) -- (1.375, 1.875); 
 \draw [fill = white, opacity = .3] (1.375, 1.875) rectangle (1.5, 2.0);

\draw[thin, black] (1.375, 1.875) -- (1.5, 1.875) -- (1.5, 2.0) -- (1.375, 2.0) -- (1.375, 1.875); 
 \draw [fill = white, opacity = .3] (1.375, 1.875) rectangle (1.5, 2.0);

\draw[thin, black] (0.25, 2.0) -- (0.375, 2.0) -- (0.375, 2.125) -- (0.25, 2.125) -- (0.25, 2.0); 
 \draw [fill = white, opacity = .3] (0.25, 2.0) rectangle (0.375, 2.125);

\draw[thin, black] (0.25, 2.0) -- (0.375, 2.0) -- (0.375, 2.125) -- (0.25, 2.125) -- (0.25, 2.0); 
 \draw [fill = white, opacity = .3] (0.25, 2.0) rectangle (0.375, 2.125);

\draw[thin, black] (0, 1.875) -- (0.125, 1.875) -- (0.125, 2.0) -- (0, 2.0) -- (0, 1.875); 
 \draw [fill = white, opacity = .3] (0, 1.875) rectangle (0.125, 2.0);

\draw[thin, black] (0, 1.875) -- (0.125, 1.875) -- (0.125, 2.0) -- (0, 2.0) -- (0, 1.875); 
 \draw [fill = white, opacity = .3] (0, 1.875) rectangle (0.125, 2.0);

\draw[thin, black] (1.75, 2.125) -- (1.875, 2.125) -- (1.875, 2.25) -- (1.75, 2.25) -- (1.75, 2.125); 
 \draw [fill = white, opacity = .3] (1.75, 2.125) rectangle (1.875, 2.25);

\draw[thin, black] (1.75, 2.125) -- (1.875, 2.125) -- (1.875, 2.25) -- (1.75, 2.25) -- (1.75, 2.125); 
 \draw [fill = white, opacity = .3] (1.75, 2.125) rectangle (1.875, 2.25);

\draw[thin, black] (1.75, 1.875) -- (1.875, 1.875) -- (1.875, 2.0) -- (1.75, 2.0) -- (1.75, 1.875); 
 \draw [fill = white, opacity = .3] (1.75, 1.875) rectangle (1.875, 2.0);

\draw[thin, black] (1.75, 1.875) -- (1.875, 1.875) -- (1.875, 2.0) -- (1.75, 2.0) -- (1.75, 1.875); 
 \draw [fill = white, opacity = .3] (1.75, 1.875) rectangle (1.875, 2.0);

\draw[thin, black] (0.5, 1.875) -- (0.625, 1.875) -- (0.625, 2.0) -- (0.5, 2.0) -- (0.5, 1.875); 
 \draw [fill = white, opacity = .3] (0.5, 1.875) rectangle (0.625, 2.0);

\draw[thin, black] (0.5, 1.875) -- (0.625, 1.875) -- (0.625, 2.0) -- (0.5, 2.0) -- (0.5, 1.875); 
 \draw [fill = white, opacity = .3] (0.5, 1.875) rectangle (0.625, 2.0);

\draw[thin, black] (1.25, 1.875) -- (1.375, 1.875) -- (1.375, 2.0) -- (1.25, 2.0) -- (1.25, 1.875); 
 \draw [fill = white, opacity = .3] (1.25, 1.875) rectangle (1.375, 2.0);

\draw[thin, black] (1.25, 1.875) -- (1.375, 1.875) -- (1.375, 2.0) -- (1.25, 2.0) -- (1.25, 1.875); 
 \draw [fill = white, opacity = .3] (1.25, 1.875) rectangle (1.375, 2.0);

\draw[thin, black] (0.375, 2.0) -- (0.5, 2.0) -- (0.5, 2.125) -- (0.375, 2.125) -- (0.375, 2.0); 
 \draw [fill = white, opacity = .3] (0.375, 2.0) rectangle (0.5, 2.125);

\draw[thin, black] (0.375, 2.0) -- (0.5, 2.0) -- (0.5, 2.125) -- (0.375, 2.125) -- (0.375, 2.0); 
 \draw [fill = white, opacity = .3] (0.375, 2.0) rectangle (0.5, 2.125);

\draw[thin, black] (0, 1.75) -- (0.125, 1.75) -- (0.125, 1.875) -- (0, 1.875) -- (0, 1.75); 
 \draw [fill = white, opacity = .3] (0, 1.75) rectangle (0.125, 1.875);

\draw[thin, black] (0, 1.75) -- (0.125, 1.75) -- (0.125, 1.875) -- (0, 1.875) -- (0, 1.75); 
 \draw [fill = white, opacity = .3] (0, 1.75) rectangle (0.125, 1.875);

\draw[thin, black] (0.125, 1.875) -- (0.25, 1.875) -- (0.25, 2.0) -- (0.125, 2.0) -- (0.125, 1.875); 
 \draw [fill = white, opacity = .3] (0.125, 1.875) rectangle (0.25, 2.0);

\draw[thin, black] (0.125, 1.875) -- (0.25, 1.875) -- (0.25, 2.0) -- (0.125, 2.0) -- (0.125, 1.875); 
 \draw [fill = white, opacity = .3] (0.125, 1.875) rectangle (0.25, 2.0);

\draw[thin, black] (2.25, 2.25) -- (2.5, 2.25) -- (2.5, 2.5) -- (2.25, 2.5) -- (2.25, 2.25); 
 \draw [fill = white, opacity = .3] (2.25, 2.25) rectangle (2.5, 2.5);

\draw[thin, black] (2.25, 2.25) -- (2.5, 2.25) -- (2.5, 2.5) -- (2.25, 2.5) -- (2.25, 2.25); 
 \draw [fill = white, opacity = .3] (2.25, 2.25) rectangle (2.5, 2.5);

\draw[thin, black] (2.0, 2.0) -- (2.25, 2.0) -- (2.25, 2.25) -- (2.0, 2.25) -- (2.0, 2.0); 
 \draw [fill = white, opacity = .3] (2.0, 2.0) rectangle (2.25, 2.25);

\draw[thin, black] (2.0, 2.0) -- (2.25, 2.0) -- (2.25, 2.25) -- (2.0, 2.25) -- (2.0, 2.0); 
 \draw [fill = white, opacity = .3] (2.0, 2.0) rectangle (2.25, 2.25);

\draw[thin, black] (2.25, 1.75) -- (2.5, 1.75) -- (2.5, 2.0) -- (2.25, 2.0) -- (2.25, 1.75); 
 \draw [fill = white, opacity = .3] (2.25, 1.75) rectangle (2.5, 2.0);

\draw[thin, black] (2.25, 1.75) -- (2.5, 1.75) -- (2.5, 2.0) -- (2.25, 2.0) -- (2.25, 1.75); 
 \draw [fill = white, opacity = .3] (2.25, 1.75) rectangle (2.5, 2.0);

\draw[thin, black] (0.75, 1.75) -- (1.0, 1.75) -- (1.0, 2.0) -- (0.75, 2.0) -- (0.75, 1.75); 
 \draw [fill = white, opacity = .3] (0.75, 1.75) rectangle (1.0, 2.0);

\draw[thin, black] (0.75, 1.75) -- (1.0, 1.75) -- (1.0, 2.0) -- (0.75, 2.0) -- (0.75, 1.75); 
 \draw [fill = white, opacity = .3] (0.75, 1.75) rectangle (1.0, 2.0);

\draw[thin, black] (1.0, 2.0) -- (1.25, 2.0) -- (1.25, 2.25) -- (1.0, 2.25) -- (1.0, 2.0); 
 \draw [fill = white, opacity = .3] (1.0, 2.0) rectangle (1.25, 2.25);

\draw[thin, black] (1.0, 2.0) -- (1.25, 2.0) -- (1.25, 2.25) -- (1.0, 2.25) -- (1.0, 2.0); 
 \draw [fill = white, opacity = .3] (1.0, 2.0) rectangle (1.25, 2.25);

\draw[thin, black] (1.75, 2.25) -- (2.0, 2.25) -- (2.0, 2.5) -- (1.75, 2.5) -- (1.75, 2.25); 
 \draw [fill = white, opacity = .3] (1.75, 2.25) rectangle (2.0, 2.5);

\draw[thin, black] (1.75, 2.25) -- (2.0, 2.25) -- (2.0, 2.5) -- (1.75, 2.5) -- (1.75, 2.25); 
 \draw [fill = white, opacity = .3] (1.75, 2.25) rectangle (2.0, 2.5);

\draw[thin, black] (0.5, 2.0) -- (0.75, 2.0) -- (0.75, 2.25) -- (0.5, 2.25) -- (0.5, 2.0); 
 \draw [fill = white, opacity = .3] (0.5, 2.0) rectangle (0.75, 2.25);

\draw[thin, black] (0.5, 2.0) -- (0.75, 2.0) -- (0.75, 2.25) -- (0.5, 2.25) -- (0.5, 2.0); 
 \draw [fill = white, opacity = .3] (0.5, 2.0) rectangle (0.75, 2.25);

\draw[thin, black] (0, 2.0) -- (0.25, 2.0) -- (0.25, 2.25) -- (0, 2.25) -- (0, 2.0); 
 \draw [fill = white, opacity = .3] (0, 2.0) rectangle (0.25, 2.25);

\draw[thin, black] (0, 2.0) -- (0.25, 2.0) -- (0.25, 2.25) -- (0, 2.25) -- (0, 2.0); 
 \draw [fill = white, opacity = .3] (0, 2.0) rectangle (0.25, 2.25);

\draw[thin, black] (2.25, 1.5) -- (2.5, 1.5) -- (2.5, 1.75) -- (2.25, 1.75) -- (2.25, 1.5); 
 \draw [fill = white, opacity = .3] (2.25, 1.5) rectangle (2.5, 1.75);

\draw[thin, black] (2.25, 1.5) -- (2.5, 1.5) -- (2.5, 1.75) -- (2.25, 1.75) -- (2.25, 1.5); 
 \draw [fill = white, opacity = .3] (2.25, 1.5) rectangle (2.5, 1.75);

\draw[thin, black] (2.0, 1.75) -- (2.25, 1.75) -- (2.25, 2.0) -- (2.0, 2.0) -- (2.0, 1.75); 
 \draw [fill = white, opacity = .3] (2.0, 1.75) rectangle (2.25, 2.0);

\draw[thin, black] (2.0, 1.75) -- (2.25, 1.75) -- (2.25, 2.0) -- (2.0, 2.0) -- (2.0, 1.75); 
 \draw [fill = white, opacity = .3] (2.0, 1.75) rectangle (2.25, 2.0);

\draw[thin, black] (1.0, 1.75) -- (1.25, 1.75) -- (1.25, 2.0) -- (1.0, 2.0) -- (1.0, 1.75); 
 \draw [fill = white, opacity = .3] (1.0, 1.75) rectangle (1.25, 2.0);

\draw[thin, black] (1.0, 1.75) -- (1.25, 1.75) -- (1.25, 2.0) -- (1.0, 2.0) -- (1.0, 1.75); 
 \draw [fill = white, opacity = .3] (1.0, 1.75) rectangle (1.25, 2.0);

\draw[thin, black] (2.0, 2.25) -- (2.25, 2.25) -- (2.25, 2.5) -- (2.0, 2.5) -- (2.0, 2.25); 
 \draw [fill = white, opacity = .3] (2.0, 2.25) rectangle (2.25, 2.5);

\draw[thin, black] (2.0, 2.25) -- (2.25, 2.25) -- (2.25, 2.5) -- (2.0, 2.5) -- (2.0, 2.25); 
 \draw [fill = white, opacity = .3] (2.0, 2.25) rectangle (2.25, 2.5);

\draw[thin, black] (0.75, 2.0) -- (1.0, 2.0) -- (1.0, 2.25) -- (0.75, 2.25) -- (0.75, 2.0); 
 \draw [fill = white, opacity = .3] (0.75, 2.0) rectangle (1.0, 2.25);

\draw[thin, black] (0.75, 2.0) -- (1.0, 2.0) -- (1.0, 2.25) -- (0.75, 2.25) -- (0.75, 2.0); 
 \draw [fill = white, opacity = .3] (0.75, 2.0) rectangle (1.0, 2.25);

\draw[thin, black] (2.5, 2.5) -- (3.0, 2.5) -- (3.0, 3.0) -- (2.5, 3.0) -- (2.5, 2.5); 
 \draw [fill = white, opacity = .3] (2.5, 2.5) rectangle (3.0, 3.0);

\draw[thin, black] (2.5, 2.5) -- (3.0, 2.5) -- (3.0, 3.0) -- (2.5, 3.0) -- (2.5, 2.5); 
 \draw [fill = white, opacity = .3] (2.5, 2.5) rectangle (3.0, 3.0);

\draw[thin, black] (2.5, 1.0) -- (3.0, 1.0) -- (3.0, 1.5) -- (2.5, 1.5) -- (2.5, 1.0); 
 \draw [fill = white, opacity = .3] (2.5, 1.0) rectangle (3.0, 1.5);

\draw[thin, black] (2.5, 1.0) -- (3.0, 1.0) -- (3.0, 1.5) -- (2.5, 1.5) -- (2.5, 1.0); 
 \draw [fill = white, opacity = .3] (2.5, 1.0) rectangle (3.0, 1.5);

\draw[thin, black] (2.0, 2.5) -- (2.5, 2.5) -- (2.5, 3.0) -- (2.0, 3.0) -- (2.0, 2.5); 
 \draw [fill = white, opacity = .3] (2.0, 2.5) rectangle (2.5, 3.0);

\draw[thin, black] (2.0, 2.5) -- (2.5, 2.5) -- (2.5, 3.0) -- (2.0, 3.0) -- (2.0, 2.5); 
 \draw [fill = white, opacity = .3] (2.0, 2.5) rectangle (2.5, 3.0);

\draw[thin, black] (2.5, 1.5) -- (3.0, 1.5) -- (3.0, 2.0) -- (2.5, 2.0) -- (2.5, 1.5); 
 \draw [fill = white, opacity = .3] (2.5, 1.5) rectangle (3.0, 2.0);

\draw[thin, black] (2.5, 1.5) -- (3.0, 1.5) -- (3.0, 2.0) -- (2.5, 2.0) -- (2.5, 1.5); 
 \draw [fill = white, opacity = .3] (2.5, 1.5) rectangle (3.0, 2.0);

\draw[thin, black] (3.0, 3.0) -- (4, 3.0) -- (4, 4) -- (3.0, 4) -- (3.0, 3.0); 
 \draw [fill = white, opacity = .3] (3.0, 3.0) rectangle (4, 4);

\draw[thin, black] (3.0, 3.0) -- (4, 3.0) -- (4, 4) -- (3.0, 4) -- (3.0, 3.0); 
 \draw [fill = white, opacity = .3] (3.0, 3.0) rectangle (4, 4);

\draw[thin, black] (3.0, 0) -- (4, 0) -- (4, 1.0) -- (3.0, 1.0) -- (3.0, 0); 
 \draw [fill = white, opacity = .3] (3.0, 0) rectangle (4, 1.0);

\draw[thin, black] (3.0, 0) -- (4, 0) -- (4, 1.0) -- (3.0, 1.0) -- (3.0, 0); 
 \draw [fill = white, opacity = .3] (3.0, 0) rectangle (4, 1.0);

\draw[thin, black] (2.0, 3.0) -- (3.0, 3.0) -- (3.0, 4) -- (2.0, 4) -- (2.0, 3.0); 
 \draw [fill = white, opacity = .3] (2.0, 3.0) rectangle (3.0, 4);

\draw[thin, black] (2.0, 3.0) -- (3.0, 3.0) -- (3.0, 4) -- (2.0, 4) -- (2.0, 3.0); 
 \draw [fill = white, opacity = .3] (2.0, 3.0) rectangle (3.0, 4);

\draw[thin, black] (3.0, 1.0) -- (4, 1.0) -- (4, 2.0) -- (3.0, 2.0) -- (3.0, 1.0); 
 \draw [fill = white, opacity = .3] (3.0, 1.0) rectangle (4, 2.0);

\draw[thin, black] (3.0, 1.0) -- (4, 1.0) -- (4, 2.0) -- (3.0, 2.0) -- (3.0, 1.0); 
 \draw [fill = white, opacity = .3] (3.0, 1.0) rectangle (4, 2.0);

\draw[line width=1.5pt, red, dashed] (1.5, 2.03125) -- (1.5625, 2.0234375); 

\draw[line width=1.5pt, red, dashed] (1.6875, 2.09375) -- (1.75, 2.09375); 

\draw[line width=1.5pt, red, dashed] (1.625, 2.0390625) -- (1.65625, 2.0625); 

\draw[line width=1.5pt, red, dashed] (0.4375, 1.96875) -- (0.5, 1.96875); 

\draw[line width=1.5pt, red, dashed] (0.375, 1.96875) -- (0.34375, 2.0); 

\draw[line width=1.5pt, red, dashed] (1.5, 2.03125) -- (1.46875, 2.0); 

\draw[line width=1.5pt, red, dashed] (1.5625, 2.0234375) -- (1.625, 2.0390625); 

\draw[line width=1.5pt, red, dashed] (1.6875, 2.09375) -- (1.65625, 2.0625); 

\draw[line width=1.5pt, red, dashed] (1.75, 2.09375) -- (1.78125, 2.125); 

\draw[line width=1.5pt, red, dashed] (0.4375, 1.96875) -- (0.375, 1.96875); 

\draw[line width=1.5pt, red, dashed] (1.875, 2.1875) -- (1.9375, 2.25); 

\draw[line width=1.5pt, red, dashed] (0.625, 1.9375) -- (0.75, 1.9375); 

\draw[line width=1.5pt, red, dashed] (1.46875, 2.0) -- (1.375, 1.9375); 

\draw[line width=1.5pt, red, dashed] (0.34375, 2.0) -- (0.25, 2.0625); 

\draw[line width=1.5pt, red, dashed] (1.78125, 2.125) -- (1.875, 2.1875); 

\draw[line width=1.5pt, red, dashed] (0.5, 1.96875) -- (0.625, 1.9375); 

\draw[line width=1.5pt, red, dashed] (1.375, 1.9375) -- (1.25, 1.9375); 

\draw[line width=1.5pt, red, dashed] (2.25, 2.375) -- (2.375, 2.5); 

\draw[line width=1.5pt, red, dashed] (0.75, 1.9375) -- (1.0, 1.875); 

\draw[line width=1.5pt, red, dashed] (1.9375, 2.25) -- (2.0, 2.375); 

\draw[line width=1.5pt, red, dashed] (0.25, 2.0625) -- (0.0, 2.125); 

\draw[line width=1.5pt, red, dashed] (1.25, 1.9375) -- (1.0, 1.875); 

\draw[line width=1.5pt, red, dashed] (2.25, 2.375) -- (2.0, 2.375); 

\draw[line width=1.5pt, red, dashed] (2.5, 2.75) -- (2.75, 3.0); 

\draw[line width=1.5pt, red, dashed] (2.375, 2.5) -- (2.5, 2.75); 

\draw[line width=1.5pt, red, dashed] (3.0, 3.5) -- (3.5, 4.0); 

\draw[line width=1.5pt, red, dashed] (2.75, 3.0) -- (3.0, 3.5); 

\draw[line width=1.5pt, blue, dotted] (1.4375, 2.09375) -- (1.5, 2.09375); 

\draw[line width=1.5pt, blue, dotted] (1.53125, 2.0625) -- (1.5625, 2.0390625); 

\draw[line width=1.5pt, blue, dotted] (1.625, 2.0234375) -- (1.6875, 2.03125); 

\draw[line width=1.5pt, blue, dotted] (1.75, 2.03125) -- (1.78125, 2.0); 

\draw[line width=1.5pt, blue, dotted] (0.3125, 1.96875) -- (0.375, 1.96875); 

\draw[line width=1.5pt, blue, dotted] (1.4375, 2.09375) -- (1.375, 2.09375); 

\draw[line width=1.5pt, blue, dotted] (1.5, 2.09375) -- (1.53125, 2.0625); 

\draw[line width=1.5pt, blue, dotted] (1.5625, 2.0390625) -- (1.625, 2.0234375); 

\draw[line width=1.5pt, blue, dotted] (1.6875, 2.03125) -- (1.75, 2.03125); 

\draw[line width=1.5pt, blue, dotted] (0.3125, 1.96875) -- (0.25, 1.96875); 

\draw[line width=1.5pt, blue, dotted] (0.375, 1.96875) -- (0.40625, 2.0); 

\draw[line width=1.5pt, blue, dotted] (1.875, 1.9375) -- (2.0, 1.9375); 

\draw[line width=1.5pt, blue, dotted] (1.375, 2.09375) -- (1.25, 2.0625); 

\draw[line width=1.5pt, blue, dotted] (0.0625, 1.875) -- (0.125, 1.9375); 

\draw[line width=1.5pt, blue, dotted] (1.78125, 2.0) -- (1.875, 1.9375); 

\draw[line width=1.5pt, blue, dotted] (0.40625, 2.0) -- (0.5, 2.0625); 

\draw[line width=1.5pt, blue, dotted] (0.0625, 1.875) -- (0.0, 1.8125); 

\draw[line width=1.5pt, blue, dotted] (0.25, 1.96875) -- (0.125, 1.9375); 

\draw[line width=1.5pt, blue, dotted] (2.375, 1.75) -- (2.25, 1.875); 

\draw[line width=1.5pt, blue, dotted] (1.25, 2.0625) -- (1.0, 2.125); 

\draw[line width=1.5pt, blue, dotted] (0.5, 2.0625) -- (0.75, 2.125); 

\draw[line width=1.5pt, blue, dotted] (2.375, 1.75) -- (2.5, 1.625); 

\draw[line width=1.5pt, blue, dotted] (2.0, 1.9375) -- (2.25, 1.875); 

\draw[line width=1.5pt, blue, dotted] (1.0, 2.125) -- (0.75, 2.125); 

\draw[line width=1.5pt, blue, dotted] (2.75, 1.5) -- (3.0, 1.25); 

\draw[line width=1.5pt, blue, dotted] (2.5, 1.625) -- (2.75, 1.5); 

\draw[line width=1.5pt, blue, dotted] (3.5, 1.0) -- (4.0, 0.5); 

\draw[line width=1.5pt, blue, dotted] (3.0, 1.25) -- (3.5, 1.0); 

\end{tikzpicture}

%% file: figures_source/lemma_schematics.tex
\begin{tikzpicture}[scale=4]
    \begin{scope}
			\draw[black] (0,0) -- (0,1) -- (1, 1) -- (1, 0) -- (0, 0) ;
			\draw[color=cyan, line width = 1.5] plot[smooth, tension = .6] coordinates{(.1, 0) (.3, .6) (.7, .6) (.8, 0.)}; 
			
			\draw[line width=1pt,black,-stealth](0.436, 0.399)--(0.32, 0.515) node[anchor=south west]{};
			\draw[line width=1pt,black,-stealth](0.59, 0.34)--(0.68, 0.43) node[anchor=south west]{};	
			
			\draw[blue, line width = 1.25] (.12, .093) -- (.667, .622);
			\node[circle, draw = blue, fill = blue, inner sep = 1.1pt, label=below:] at (.12, .093) {};
			\node[circle, draw = blue, fill = blue, inner sep = 1.1pt, label=below:] at (.667, .622) {};

			\draw[red, line width = 1.25](.79, .122) -- (.318, .61);
			
			\node[circle, draw = red, fill = red, inner sep = 1.1pt, label=below:] at (.79, .122) {};
			\node[circle, draw = red, fill = red, inner sep = 1.1pt, label=below:] at (.318, .61) {};
			
			\node[label=below:] at (.33, .23) {$s_1$};
			\node[label=below:] at (.65, .2) {$s_2$};
			\node[label=below:] at (.67, .48) {$\nabla f_1$};
			\node[label=below:] at (.24, .46) {$\nabla f_1$};
   \node[below] at (.5,0) {(a)};
		\end{scope}
    \begin{scope}[shift = {(1.5,0)}]
		\draw[black] (0,0) -- (0,1) -- (1, 1) -- (1, 0) -- (0, 0) ;
		\draw[color=cyan, line width = 1.5] plot[smooth, tension = .6] coordinates{(.1, 0) (.3, .6) (.7, .6) (.8, 0.)}; 
		\draw[color=orange, line width = 1.5] plot[smooth, tension = .6] coordinates{(.2, 0) (.4, .5) (.7, .7) (1, 0.5)};

		\draw[red, line width = 1.25](.298, .287) -- (.77, .678);
        \draw[blue, line width = 1.25](.118, .089) -- (.7135, .5823);
		\node[circle, draw = blue, fill = blue, inner sep = 1.1pt, label=below:] at (.118, .089) {};
		\node[circle, draw = blue, fill = blue, inner sep = 1.1pt, label=below:] at (.7135, .5823) {};

		\node[circle, draw = red, fill = red, inner sep = 1.1pt, label=below:] at (.298, .287) {};
		\node[circle, draw = red, fill = red, inner sep = 1.1pt, label=below:] at (.77, .678) {};
        \draw[line width=1pt,black,-stealth](0.636, 0.567)--(0.55, 0.671) node[anchor=south west]{};
		\draw[line width=1pt,black,-stealth](0.467, 0.378)--(0.35, 0.519) node[anchor=south west]{};	
		\node[label=below:] at (.3, .1) {$f_2$};
		\node[label=below:] at (.85, .2) {$f_1$};
		\node[label=below:] at (.5, .74) {$\nabla f_2$};
		\node[label=below:] at (.28, .57) {$\nabla f_1$};
        \node[label=below:] at (.38, .2) {$s_1$};
        \node[label=below:] at (.8, .6) {$s_2$};
         \node[below] at (.5,0) {(b)};
    \end{scope}
\end{tikzpicture}%

%% file: figures_source/two_ellipse_tangential_nonintersect_before.tex
\begin{tikzpicture}[scale = 1.25 ]
\begin{scope}
\clip (-1.135, -2.262) rectangle (4.508,1.135);
\draw[thin, black] (2.25, 2.25) -- (4.5, 2.25) -- (4.5, 4.5) -- (2.25, 4.5) -- (2.25, 2.25);
 \draw [fill = white, opacity = .25] (2.25, 2.25) rectangle (4.5, 4.5);

\draw[thin, black] (0.0, 2.25) -- (2.25, 2.25) -- (2.25, 4.5) -- (0.0, 4.5) -- (0.0, 2.25);
 \draw [fill = white, opacity = .25] (0.0, 2.25) rectangle (2.25, 4.5);

\draw[thin, black] (-4.5, -4.5) -- (-2.25, -4.5) -- (-2.25, -2.25) -- (-4.5, -2.25) -- (-4.5, -4.5);
 \draw [fill = white, opacity = .25] (-4.5, -4.5) rectangle (-2.25, -2.25);

\draw[thin, black] (1.125, 1.125) -- (2.25, 1.125) -- (2.25, 2.25) -- (1.125, 2.25) -- (1.125, 1.125);
 \draw [fill = white, opacity = .25] (1.125, 1.125) rectangle (2.25, 2.25);

\draw[thin, black] (0.0, 1.125) -- (1.125, 1.125) -- (1.125, 2.25) -- (0.0, 2.25) -- (0.0, 1.125);
 \draw [fill = white, opacity = .25] (0.0, 1.125) rectangle (1.125, 2.25);

\draw[thin, black] (-2.25, -1.125) -- (-1.125, -1.125) -- (-1.125, 0.0) -- (-2.25, 0.0) -- (-2.25, -1.125);
 \draw [fill = white, opacity = .25] (-2.25, -1.125) rectangle (-1.125, 0.0);

\draw[thin, black] (-2.25, -2.25) -- (-1.125, -2.25) -- (-1.125, -1.125) -- (-2.25, -1.125) -- (-2.25, -2.25);
 \draw [fill = white, opacity = .25] (-2.25, -2.25) rectangle (-1.125, -1.125);

\draw[thin, black] (-3.375, -1.125) -- (-2.25, -1.125) -- (-2.25, 0.0) -- (-3.375, 0.0) -- (-3.375, -1.125);
 \draw [fill = white, opacity = .25] (-3.375, -1.125) rectangle (-2.25, 0.0);

\draw[thin, black] (-4.5, -1.125) -- (-3.375, -1.125) -- (-3.375, 0.0) -- (-4.5, 0.0) -- (-4.5, -1.125);
 \draw [fill = white, opacity = .25] (-4.5, -1.125) rectangle (-3.375, 0.0);

\draw[thin, black] (-4.5, -2.25) -- (-3.375, -2.25) -- (-3.375, -1.125) -- (-4.5, -1.125) -- (-4.5, -2.25);
 \draw [fill = white, opacity = .25] (-4.5, -2.25) rectangle (-3.375, -1.125);

\draw[thin, black] (-3.375, -2.25) -- (-2.25, -2.25) -- (-2.25, -1.125) -- (-3.375, -1.125) -- (-3.375, -2.25);
 \draw [fill = white, opacity = .25] (-3.375, -2.25) rectangle (-2.25, -1.125);

\draw[thin, black] (1.125, 0.0) -- (1.6875, 0.0) -- (1.6875, 0.5625) -- (1.125, 0.5625) -- (1.125, 0.0);
 \draw [fill = white, opacity = .25] (1.125, 0.0) rectangle (1.6875, 0.5625);

\draw[thin, black] (1.6875, 0.0) -- (2.25, 0.0) -- (2.25, 0.5625) -- (1.6875, 0.5625) -- (1.6875, 0.0);
 \draw [fill = white, opacity = .25] (1.6875, 0.0) rectangle (2.25, 0.5625);

\draw[thin, black] (-1.125, -1.6875) -- (-0.5625, -1.6875) -- (-0.5625, -1.125) -- (-1.125, -1.125) -- (-1.125, -1.6875);
 \draw [fill = white, opacity = .25] (-1.125, -1.6875) rectangle (-0.5625, -1.125);

\draw[thin, black] (-1.125, -2.25) -- (-0.5625, -2.25) -- (-0.5625, -1.6875) -- (-1.125, -1.6875) -- (-1.125, -2.25);
 \draw [fill = white, opacity = .25] (-1.125, -2.25) rectangle (-0.5625, -1.6875);

\draw[thin, black] (2.25, -1.6875) -- (2.8125, -1.6875) -- (2.8125, -1.125) -- (2.25, -1.125) -- (2.25, -1.6875);
 \draw [fill = white, opacity = .25] (2.25, -1.6875) rectangle (2.8125, -1.125);

\draw[thin, black] (3.9375, -1.6875) -- (4.5, -1.6875) -- (4.5, -1.125) -- (3.9375, -1.125) -- (3.9375, -1.6875);
 \draw [fill = white, opacity = .25] (3.9375, -1.6875) rectangle (4.5, -1.125);

\draw[thin, black] (3.9375, -2.25) -- (4.5, -2.25) -- (4.5, -1.6875) -- (3.9375, -1.6875) -- (3.9375, -2.25);
 \draw [fill = white, opacity = .25] (3.9375, -2.25) rectangle (4.5, -1.6875);

\draw[thin, black] (1.6875, -0.5625) -- (2.25, -0.5625) -- (2.25, 0.0) -- (1.6875, 0.0) -- (1.6875, -0.5625);
 \draw [fill = white, opacity = .25] (1.6875, -0.5625) rectangle (2.25, 0.0);

\draw[thin, black] (1.125, -0.5625) -- (1.6875, -0.5625) -- (1.6875, 0.0) -- (1.125, 0.0) -- (1.125, -0.5625);
 \draw [fill = white, opacity = .25] (1.125, -0.5625) rectangle (1.6875, 0.0);

\draw[thin, black] (0.5625, -1.6875) -- (1.125, -1.6875) -- (1.125, -1.125) -- (0.5625, -1.125) -- (0.5625, -1.6875);
 \draw [fill = white, opacity = .25] (0.5625, -1.6875) rectangle (1.125, -1.125);

\draw[thin, black] (0.0, -1.6875) -- (0.5625, -1.6875) -- (0.5625, -1.125) -- (0.0, -1.125) -- (0.0, -1.6875);
 \draw [fill = white, opacity = .25] (0.0, -1.6875) rectangle (0.5625, -1.125);

\draw[thin, black] (1.6875, -1.6875) -- (2.25, -1.6875) -- (2.25, -1.125) -- (1.6875, -1.125) -- (1.6875, -1.6875);
 \draw [fill = white, opacity = .25] (1.6875, -1.6875) rectangle (2.25, -1.125);

\draw[thin, black] (1.125, -1.6875) -- (1.6875, -1.6875) -- (1.6875, -1.125) -- (1.125, -1.125) -- (1.125, -1.6875);
 \draw [fill = white, opacity = .25] (1.125, -1.6875) rectangle (1.6875, -1.125);

\draw[thin, black] (2.8125, -1.40625) -- (3.09375, -1.40625) -- (3.09375, -1.125) -- (2.8125, -1.125) -- (2.8125, -1.40625);
 \draw [fill = white, opacity = .25] (2.8125, -1.40625) rectangle (3.09375, -1.125);

\draw[thin, black] (2.8125, -1.6875) -- (3.09375, -1.6875) -- (3.09375, -1.40625) -- (2.8125, -1.40625) -- (2.8125, -1.6875);
 \draw [fill = white, opacity = .25] (2.8125, -1.6875) rectangle (3.09375, -1.40625);

\draw[thin, black] (3.65625, -1.40625) -- (3.9375, -1.40625) -- (3.9375, -1.125) -- (3.65625, -1.125) -- (3.65625, -1.40625);
 \draw [fill = white, opacity = .25] (3.65625, -1.40625) rectangle (3.9375, -1.125);

\draw[thin, black] (3.65625, -1.6875) -- (3.9375, -1.6875) -- (3.9375, -1.40625) -- (3.65625, -1.40625) -- (3.65625, -1.6875);
 \draw [fill = white, opacity = .25] (3.65625, -1.6875) rectangle (3.9375, -1.40625);

\draw[thin, black] (-1.125, -0.5625) -- (-0.5625, -0.5625) -- (-0.5625, 0.0) -- (-1.125, 0.0) -- (-1.125, -0.5625);
 \draw [fill = white, opacity = .25] (-1.125, -0.5625) rectangle (-0.5625, 0.0);

\draw[thin, black] (-1.125, -1.125) -- (-0.5625, -1.125) -- (-0.5625, -0.5625) -- (-1.125, -0.5625) -- (-1.125, -1.125);
 \draw [fill = white, opacity = .25] (-1.125, -1.125) rectangle (-0.5625, -0.5625);

\draw[thin, black] (3.9375, -0.5625) -- (4.5, -0.5625) -- (4.5, 0.0) -- (3.9375, 0.0) -- (3.9375, -0.5625);
 \draw [fill = white, opacity = .25] (3.9375, -0.5625) rectangle (4.5, 0.0);

\draw[thin, black] (3.9375, -1.125) -- (4.5, -1.125) -- (4.5, -0.5625) -- (3.9375, -0.5625) -- (3.9375, -1.125);
 \draw [fill = white, opacity = .25] (3.9375, -1.125) rectangle (4.5, -0.5625);

\draw[thin, black] (3.375, 1.125) -- (4.5, 1.125) -- (4.5, 2.25) -- (3.375, 2.25) -- (3.375, 1.125);
 \draw [fill = white, opacity = .25] (3.375, 1.125) rectangle (4.5, 2.25);

\draw[thin, black] (2.25, 1.125) -- (3.375, 1.125) -- (3.375, 2.25) -- (2.25, 2.25) -- (2.25, 1.125);
 \draw [fill = white, opacity = .25] (2.25, 1.125) rectangle (3.375, 2.25);

\draw[thin, black] (3.375, -3.375) -- (4.5, -3.375) -- (4.5, -2.25) -- (3.375, -2.25) -- (3.375, -3.375);
 \draw [fill = white, opacity = .25] (3.375, -3.375) rectangle (4.5, -2.25);

\draw[thin, black] (2.25, -4.5) -- (3.375, -4.5) -- (3.375, -3.375) -- (2.25, -3.375) -- (2.25, -4.5);
 \draw [fill = white, opacity = .25] (2.25, -4.5) rectangle (3.375, -3.375);

\draw[thin, black] (3.375, -4.5) -- (4.5, -4.5) -- (4.5, -3.375) -- (3.375, -3.375) -- (3.375, -4.5);
 \draw [fill = white, opacity = .25] (3.375, -4.5) rectangle (4.5, -3.375);

\draw[thin, black] (1.125, -3.375) -- (2.25, -3.375) -- (2.25, -2.25) -- (1.125, -2.25) -- (1.125, -3.375);
 \draw [fill = white, opacity = .25] (1.125, -3.375) rectangle (2.25, -2.25);

\draw[thin, black] (0.0, -4.5) -- (1.125, -4.5) -- (1.125, -3.375) -- (0.0, -3.375) -- (0.0, -4.5);
 \draw [fill = white, opacity = .25] (0.0, -4.5) rectangle (1.125, -3.375);

\draw[thin, black] (1.125, -4.5) -- (2.25, -4.5) -- (2.25, -3.375) -- (1.125, -3.375) -- (1.125, -4.5);
 \draw [fill = white, opacity = .25] (1.125, -4.5) rectangle (2.25, -3.375);

\draw[thin, black] (-2.25, 2.25) -- (0.0, 2.25) -- (0.0, 4.5) -- (-2.25, 4.5) -- (-2.25, 2.25);
 \draw [fill = white, opacity = .25] (-2.25, 2.25) rectangle (0.0, 4.5);

\draw[thin, black] (-4.5, 2.25) -- (-2.25, 2.25) -- (-2.25, 4.5) -- (-4.5, 4.5) -- (-4.5, 2.25);
 \draw [fill = white, opacity = .25] (-4.5, 2.25) rectangle (-2.25, 4.5);

\draw[thin, black] (-1.125, 1.125) -- (0.0, 1.125) -- (0.0, 2.25) -- (-1.125, 2.25) -- (-1.125, 1.125);
 \draw [fill = white, opacity = .25] (-1.125, 1.125) rectangle (0.0, 2.25);

\draw[thin, black] (-2.25, 1.125) -- (-1.125, 1.125) -- (-1.125, 2.25) -- (-2.25, 2.25) -- (-2.25, 1.125);
 \draw [fill = white, opacity = .25] (-2.25, 1.125) rectangle (-1.125, 2.25);

\draw[thin, black] (-2.25, 0.0) -- (-1.125, 0.0) -- (-1.125, 1.125) -- (-2.25, 1.125) -- (-2.25, 0.0);
 \draw [fill = white, opacity = .25] (-2.25, 0.0) rectangle (-1.125, 1.125);

\draw[thin, black] (-2.25, -3.375) -- (-1.125, -3.375) -- (-1.125, -2.25) -- (-2.25, -2.25) -- (-2.25, -3.375);
 \draw [fill = white, opacity = .25] (-2.25, -3.375) rectangle (-1.125, -2.25);

\draw[thin, black] (-2.25, -4.5) -- (-1.125, -4.5) -- (-1.125, -3.375) -- (-2.25, -3.375) -- (-2.25, -4.5);
 \draw [fill = white, opacity = .25] (-2.25, -4.5) rectangle (-1.125, -3.375);

\draw[thin, black] (-1.125, -4.5) -- (0.0, -4.5) -- (0.0, -3.375) -- (-1.125, -3.375) -- (-1.125, -4.5);
 \draw [fill = white, opacity = .25] (-1.125, -4.5) rectangle (0.0, -3.375);

\draw[thin, black] (-0.28125, -1.40625) -- (0.0, -1.40625) -- (0.0, -1.125) -- (-0.28125, -1.125) -- (-0.28125, -1.40625);
 \draw [fill = white, opacity = .25] (-0.28125, -1.40625) rectangle (0.0, -1.125);

\draw[thin, black] (-0.5625, -1.6875) -- (-0.28125, -1.6875) -- (-0.28125, -1.40625) -- (-0.5625, -1.40625) -- (-0.5625, -1.6875);
 \draw [fill = white, opacity = .25] (-0.5625, -1.6875) rectangle (-0.28125, -1.40625);

\draw[thin, black] (3.09375, -1.40625) -- (3.375, -1.40625) -- (3.375, -1.125) -- (3.09375, -1.125) -- (3.09375, -1.40625);
 \draw [fill = white, opacity = .25] (3.09375, -1.40625) rectangle (3.375, -1.125);

\draw[thin, black] (3.375, -1.6875) -- (3.65625, -1.6875) -- (3.65625, -1.40625) -- (3.375, -1.40625) -- (3.375, -1.6875);
 \draw [fill = white, opacity = .25] (3.375, -1.6875) rectangle (3.65625, -1.40625);

\draw[thin, black] (-0.5625, -1.40625) -- (-0.28125, -1.40625) -- (-0.28125, -1.125) -- (-0.5625, -1.125) -- (-0.5625, -1.40625);
 \draw [fill = white, opacity = .25] (-0.5625, -1.40625) rectangle (-0.28125, -1.125);

\draw[thin, black] (-0.28125, -1.6875) -- (0.0, -1.6875) -- (0.0, -1.40625) -- (-0.28125, -1.40625) -- (-0.28125, -1.6875);
 \draw [fill = white, opacity = .25] (-0.28125, -1.6875) rectangle (0.0, -1.40625);

\draw[thin, black] (3.09375, -1.6875) -- (3.375, -1.6875) -- (3.375, -1.40625) -- (3.09375, -1.40625) -- (3.09375, -1.6875);
 \draw [fill = white, opacity = .25] (3.09375, -1.6875) rectangle (3.375, -1.40625);

\draw[thin, black] (3.375, -1.40625) -- (3.65625, -1.40625) -- (3.65625, -1.125) -- (3.375, -1.125) -- (3.375, -1.40625);
 \draw [fill = white, opacity = .25] (3.375, -1.40625) rectangle (3.65625, -1.125);

\draw[thin, black] (1.6875, 0.5625) -- (2.25, 0.5625) -- (2.25, 1.125) -- (1.6875, 1.125) -- (1.6875, 0.5625);
 \draw [fill = white, opacity = .25] (1.6875, 0.5625) rectangle (2.25, 1.125);

\draw[thin, black] (2.25, -2.25) -- (2.8125, -2.25) -- (2.8125, -1.6875) -- (2.25, -1.6875) -- (2.25, -2.25);
 \draw [fill = white, opacity = .25] (2.25, -2.25) rectangle (2.8125, -1.6875);

\draw[thin, black] (1.125, -1.125) -- (1.6875, -1.125) -- (1.6875, -0.5625) -- (1.125, -0.5625) -- (1.125, -1.125);
 \draw [fill = white, opacity = .25] (1.125, -1.125) rectangle (1.6875, -0.5625);

\draw[thin, black] (0.0, -2.25) -- (0.5625, -2.25) -- (0.5625, -1.6875) -- (0.0, -1.6875) -- (0.0, -2.25);
 \draw [fill = white, opacity = .25] (0.0, -2.25) rectangle (0.5625, -1.6875);

\draw[thin, black] (1.125, -2.25) -- (1.6875, -2.25) -- (1.6875, -1.6875) -- (1.125, -1.6875) -- (1.125, -2.25);
 \draw [fill = white, opacity = .25] (1.125, -2.25) rectangle (1.6875, -1.6875);

\draw[thin, black] (-0.5625, -0.5625) -- (0.0, -0.5625) -- (0.0, 0.0) -- (-0.5625, 0.0) -- (-0.5625, -0.5625);
 \draw [fill = white, opacity = .25] (-0.5625, -0.5625) rectangle (0.0, 0.0);

\draw[thin, black] (2.8125, -0.5625) -- (3.375, -0.5625) -- (3.375, 0.0) -- (2.8125, 0.0) -- (2.8125, -0.5625);
 \draw [fill = white, opacity = .25] (2.8125, -0.5625) rectangle (3.375, 0.0);

\draw[thin, black] (2.25, -1.125) -- (2.8125, -1.125) -- (2.8125, -0.5625) -- (2.25, -0.5625) -- (2.25, -1.125);
 \draw [fill = white, opacity = .25] (2.25, -1.125) rectangle (2.8125, -0.5625);

\draw[thin, black] (3.375, -2.25) -- (3.9375, -2.25) -- (3.9375, -1.6875) -- (3.375, -1.6875) -- (3.375, -2.25);
 \draw [fill = white, opacity = .25] (3.375, -2.25) rectangle (3.9375, -1.6875);

\draw[thin, black] (3.375, -1.125) -- (3.9375, -1.125) -- (3.9375, -0.5625) -- (3.375, -0.5625) -- (3.375, -1.125);
 \draw [fill = white, opacity = .25] (3.375, -1.125) rectangle (3.9375, -0.5625);

\draw[thin, black] (1.125, 0.5625) -- (1.6875, 0.5625) -- (1.6875, 1.125) -- (1.125, 1.125) -- (1.125, 0.5625);
 \draw [fill = white, opacity = .25] (1.125, 0.5625) rectangle (1.6875, 1.125);

\draw[thin, black] (-0.5625, -2.25) -- (0.0, -2.25) -- (0.0, -1.6875) -- (-0.5625, -1.6875) -- (-0.5625, -2.25);
 \draw [fill = white, opacity = .25] (-0.5625, -2.25) rectangle (0.0, -1.6875);

\draw[thin, black] (1.6875, -2.25) -- (2.25, -2.25) -- (2.25, -1.6875) -- (1.6875, -1.6875) -- (1.6875, -2.25);
 \draw [fill = white, opacity = .25] (1.6875, -2.25) rectangle (2.25, -1.6875);

\draw[thin, black] (0.5625, -2.25) -- (1.125, -2.25) -- (1.125, -1.6875) -- (0.5625, -1.6875) -- (0.5625, -2.25);
 \draw [fill = white, opacity = .25] (0.5625, -2.25) rectangle (1.125, -1.6875);

\draw[thin, black] (-0.5625, -1.125) -- (0.0, -1.125) -- (0.0, -0.5625) -- (-0.5625, -0.5625) -- (-0.5625, -1.125);
 \draw [fill = white, opacity = .25] (-0.5625, -1.125) rectangle (0.0, -0.5625);

\draw[thin, black] (1.6875, -1.125) -- (2.25, -1.125) -- (2.25, -0.5625) -- (1.6875, -0.5625) -- (1.6875, -1.125);
 \draw [fill = white, opacity = .25] (1.6875, -1.125) rectangle (2.25, -0.5625);

\draw[thin, black] (2.25, -0.5625) -- (2.8125, -0.5625) -- (2.8125, 0.0) -- (2.25, 0.0) -- (2.25, -0.5625);
 \draw [fill = white, opacity = .25] (2.25, -0.5625) rectangle (2.8125, 0.0);

\draw[thin, black] (2.8125, -2.25) -- (3.375, -2.25) -- (3.375, -1.6875) -- (2.8125, -1.6875) -- (2.8125, -2.25);
 \draw [fill = white, opacity = .25] (2.8125, -2.25) rectangle (3.375, -1.6875);

\draw[thin, black] (3.375, -0.5625) -- (3.9375, -0.5625) -- (3.9375, 0.0) -- (3.375, 0.0) -- (3.375, -0.5625);
 \draw [fill = white, opacity = .25] (3.375, -0.5625) rectangle (3.9375, 0.0);

\draw[thin, black] (2.8125, -1.125) -- (3.375, -1.125) -- (3.375, -0.5625) -- (2.8125, -0.5625) -- (2.8125, -1.125);
 \draw [fill = white, opacity = .25] (2.8125, -1.125) rectangle (3.375, -0.5625);

\draw[thin, black] (-1.125, 0.0) -- (0.0, 0.0) -- (0.0, 1.125) -- (-1.125, 1.125) -- (-1.125, 0.0);
 \draw [fill = white, opacity = .25] (-1.125, 0.0) rectangle (0.0, 1.125);

\draw[thin, black] (-1.125, -3.375) -- (0.0, -3.375) -- (0.0, -2.25) -- (-1.125, -2.25) -- (-1.125, -3.375);
 \draw [fill = white, opacity = .25] (-1.125, -3.375) rectangle (0.0, -2.25);

\draw[thin, black] (0.0, -1.125) -- (1.125, -1.125) -- (1.125, 0.0) -- (0.0, 0.0) -- (0.0, -1.125);
 \draw [fill = white, opacity = .25] (0.0, -1.125) rectangle (1.125, 0.0);

\draw[thin, black] (2.25, 0.0) -- (3.375, 0.0) -- (3.375, 1.125) -- (2.25, 1.125) -- (2.25, 0.0);
 \draw [fill = white, opacity = .25] (2.25, 0.0) rectangle (3.375, 1.125);

\draw[thin, black] (2.25, -3.375) -- (3.375, -3.375) -- (3.375, -2.25) -- (2.25, -2.25) -- (2.25, -3.375);
 \draw [fill = white, opacity = .25] (2.25, -3.375) rectangle (3.375, -2.25);

\draw[thin, black] (0.0, -3.375) -- (1.125, -3.375) -- (1.125, -2.25) -- (0.0, -2.25) -- (0.0, -3.375);
 \draw [fill = white, opacity = .25] (0.0, -3.375) rectangle (1.125, -2.25);

\draw[thin, black] (0.0, 0.0) -- (1.125, 0.0) -- (1.125, 1.125) -- (0.0, 1.125) -- (0.0, 0.0);
 \draw [fill = white, opacity = .25] (0.0, 0.0) rectangle (1.125, 1.125);

\draw[thin, black] (3.375, 0.0) -- (4.5, 0.0) -- (4.5, 1.125) -- (3.375, 1.125) -- (3.375, 0.0);
 \draw [fill = white, opacity = .25] (3.375, 0.0) rectangle (4.5, 1.125);

\draw[thin, black] (-4.5, 0.0) -- (-2.25, 0.0) -- (-2.25, 2.25) -- (-4.5, 2.25) -- (-4.5, 0.0);
 \draw [fill = white, opacity = .25] (-4.5, 0.0) rectangle (-2.25, 2.25);

\draw[color=orange, line width = 1, dashed] plot[id=posalpha, raw gnuplot, smooth] function{
  f(x,y) = -0.4 - 1.3*x + 0.4*x**2 + 4*y**2 ;
  set xrange [-8:8];
  set yrange [-8:8];
  set view 0,0;
  set isosample 1000,1000;
  set size square;
  set cont base;
  set cntrparam levels incre 0,0.1,0;
  unset surface;
  splot f(x,y)};

\draw[color=cyan, line width = 1, dotted] plot[id=posalpha, raw gnuplot, smooth] function{
  g(x,y) = 7.44 - 1.3*x + 11.2*y + 0.4*x**2 + 4*y**2;
  set xrange [-8:8];
  set yrange [-8:8];
  set view 0,0;
  set isosample 1000,1000;
  set size square;
  set cont base;
  set cntrparam levels incre 0,0.1,0;
  unset surface;
  splot g(x,y)};

\draw[red, line width = 1.5pt, dashed] (1.6875, 0.84375) -- (2.25, 0.84375);

\draw[red, line width = 1.5pt, dashed] (1.6875, -0.84375) -- (1.125, -0.84375);

\draw[red, line width = 1.5pt, dashed] (-0.28125, 0.0) -- (0.0, -0.28125);

\draw[red, line width = 1.5pt, dashed] (2.8125, -0.28125) -- (3.375, -0.28125);

\draw[red, line width = 1.5pt, dashed] (2.25, -0.84375) -- (2.53125, -0.5625);

\draw[red, line width = 1.5pt, dashed] (1.6875, 0.84375) -- (1.125, 0.84375);

\draw[red, line width = 1.5pt, dashed] (1.6875, -0.84375) -- (2.25, -0.84375);

\draw[red, line width = 1.5pt, dashed] (2.8125, -0.28125) -- (2.53125, -0.5625);

\draw[red, line width = 1.5pt, dashed] (3.375, -0.28125) -- (3.65625, 0.0);

\draw[red, line width = 1.5pt, dashed] (-0.28125, 0.0) -- (0.0, 0.5625);

\draw[red, line width = 1.5pt, dashed] (1.125, -0.84375) -- (0.0, -0.28125);

\draw[red, line width = 1.5pt, dashed] (2.25, 0.84375) -- (3.375, 0.5625);

\draw[red, line width = 1.5pt, dashed] (1.125, 0.84375) -- (0.0, 0.5625);

\draw[red, line width = 1.5pt, dashed] (3.65625, 0.0) -- (3.375, 0.5625);

\draw[blue, line width = 1.5pt, dotted] (-0.28125, -1.265625) -- (-0.140625, -1.125);

\draw[blue, line width = 1.5pt, dotted] (-0.421875, -1.40625) -- (-0.28125, -1.546875);

\draw[blue, line width = 1.5pt, dotted] (3.375, -1.265625) -- (3.234375, -1.125);

\draw[blue, line width = 1.5pt, dotted] (3.375, -1.546875) -- (3.515625, -1.40625);

\draw[blue, line width = 1.5pt, dotted] (-0.28125, -1.265625) -- (-0.421875, -1.40625);

\draw[blue, line width = 1.5pt, dotted] (-0.28125, -1.546875) -- (-0.140625, -1.6875);

\draw[blue, line width = 1.5pt, dotted] (3.375, -1.546875) -- (3.234375, -1.6875);

\draw[blue, line width = 1.5pt, dotted] (3.375, -1.265625) -- (3.515625, -1.40625);

\draw[blue, line width = 1.5pt, dotted] (2.25, -1.96875) -- (2.8125, -1.96875);

\draw[blue, line width = 1.5pt, dotted] (1.6875, -0.84375) -- (1.125, -0.84375);

\draw[blue, line width = 1.5pt, dotted] (0.0, -1.96875) -- (0.5625, -1.96875);

\draw[blue, line width = 1.5pt, dotted] (1.6875, -1.96875) -- (1.125, -1.96875);

\draw[blue, line width = 1.5pt, dotted] (2.25, -0.84375) -- (2.8125, -0.84375);

\draw[blue, line width = 1.5pt, dotted] (-0.140625, -1.6875) -- (0.0, -1.96875);

\draw[blue, line width = 1.5pt, dotted] (2.25, -1.96875) -- (1.6875, -1.96875);

\draw[blue, line width = 1.5pt, dotted] (0.5625, -1.96875) -- (1.125, -1.96875);

\draw[blue, line width = 1.5pt, dotted] (-0.140625, -1.125) -- (0.0, -0.84375);

\draw[blue, line width = 1.5pt, dotted] (1.6875, -0.84375) -- (2.25, -0.84375);

\draw[blue, line width = 1.5pt, dotted] (3.234375, -1.6875) -- (2.8125, -1.96875);

\draw[blue, line width = 1.5pt, dotted] (3.234375, -1.125) -- (2.8125, -0.84375);

\draw[blue, line width = 1.5pt, dotted] (1.125, -0.84375) -- (0.0, -0.84375);

\end{scope}

\begin{scope}[shift = {(-6.5,-0.55)},scale = 1.2]
		\clip (-1.0, -1.0) rectangle (4,1.0);
\draw[thin, black] (2.0, 2.0) -- (4, 2.0) -- (4, 4) -- (2.0, 4) -- (2.0, 2.0);
 \draw [] (2.0, 2.0) rectangle (4, 4);

\draw[thin, black] (0.0, 2.0) -- (2.0, 2.0) -- (2.0, 4) -- (0.0, 4) -- (0.0, 2.0);
 \draw [] (0.0, 2.0) rectangle (2.0, 4);

\draw[thin, black] (0.0, -4) -- (2.0, -4) -- (2.0, -2.0) -- (0.0, -2.0) -- (0.0, -4);
 \draw [] (0.0, -4) rectangle (2.0, -2.0);

\draw[thin, black] (2.0, -4) -- (4, -4) -- (4, -2.0) -- (2.0, -2.0) -- (2.0, -4);
 \draw [] (2.0, -4) rectangle (4, -2.0);

\draw[thin, black] (0.0, 0.0) -- (0.5, 0.0) -- (0.5, 0.5) -- (0.0, 0.5) -- (0.0, 0.0);
 \draw [] (0.0, 0.0) rectangle (0.5, 0.5);

\draw[thin, black] (0.5, 0.0) -- (1.0, 0.0) -- (1.0, 0.5) -- (0.5, 0.5) -- (0.5, 0.0);
 \draw [] (0.5, 0.0) rectangle (1.0, 0.5);

\draw[thin, black] (2.0, 0.0) -- (2.5, 0.0) -- (2.5, 0.5) -- (2.0, 0.5) -- (2.0, 0.0);
 \draw [] (2.0, 0.0) rectangle (2.5, 0.5);

\draw[thin, black] (2.5, 0.0) -- (3.0, 0.0) -- (3.0, 0.5) -- (2.5, 0.5) -- (2.5, 0.0);
 \draw [] (2.5, 0.0) rectangle (3.0, 0.5);

\draw[thin, black] (2.5, -0.5) -- (3.0, -0.5) -- (3.0, 0.0) -- (2.5, 0.0) -- (2.5, -0.5);
 \draw [] (2.5, -0.5) rectangle (3.0, 0.0);

\draw[thin, black] (2.0, -0.5) -- (2.5, -0.5) -- (2.5, 0.0) -- (2.0, 0.0) -- (2.0, -0.5);
 \draw [] (2.0, -0.5) rectangle (2.5, 0.0);

\draw[thin, black] (0.5, -0.5) -- (1.0, -0.5) -- (1.0, 0.0) -- (0.5, 0.0) -- (0.5, -0.5);
 \draw [] (0.5, -0.5) rectangle (1.0, 0.0);

\draw[thin, black] (0.0, -0.5) -- (0.5, -0.5) -- (0.5, 0.0) -- (0.0, 0.0) -- (0.0, -0.5);
 \draw [] (0.0, -0.5) rectangle (0.5, 0.0);

\draw[thin, black] (-4, -2.0) -- (-2.0, -2.0) -- (-2.0, 0.0) -- (-4, 0.0) -- (-4, -2.0);
 \draw [] (-4, -2.0) rectangle (-2.0, 0.0);

\draw[thin, black] (-4, -4) -- (-2.0, -4) -- (-2.0, -2.0) -- (-4, -2.0) -- (-4, -4);
 \draw [] (-4, -4) rectangle (-2.0, -2.0);

\draw[thin, black] (-2.0, -4) -- (0.0, -4) -- (0.0, -2.0) -- (-2.0, -2.0) -- (-2.0, -4);
 \draw [] (-2.0, -4) rectangle (0.0, -2.0);

\draw[thin, black] (-2.0, -1.0) -- (-1.0, -1.0) -- (-1.0, 0.0) -- (-2.0, 0.0) -- (-2.0, -1.0);
 \draw [] (-2.0, -1.0) rectangle (-1.0, 0.0);

\draw[thin, black] (-2.0, -2.0) -- (-1.0, -2.0) -- (-1.0, -1.0) -- (-2.0, -1.0) -- (-2.0, -2.0);
 \draw [] (-2.0, -2.0) rectangle (-1.0, -1.0);

\draw[thin, black] (-2.0, 2.0) -- (0.0, 2.0) -- (0.0, 4) -- (-2.0, 4) -- (-2.0, 2.0);
 \draw [] (-2.0, 2.0) rectangle (0.0, 4);

\draw[thin, black] (-4, 2.0) -- (-2.0, 2.0) -- (-2.0, 4) -- (-4, 4) -- (-4, 2.0);
 \draw [] (-4, 2.0) rectangle (-2.0, 4);

\draw[thin, black] (-4, 0.0) -- (-2.0, 0.0) -- (-2.0, 2.0) -- (-4, 2.0) -- (-4, 0.0);
 \draw [] (-4, 0.0) rectangle (-2.0, 2.0);

\draw[thin, black] (-2.0, 1.0) -- (-1.0, 1.0) -- (-1.0, 2.0) -- (-2.0, 2.0) -- (-2.0, 1.0);
 \draw [] (-2.0, 1.0) rectangle (-1.0, 2.0);

\draw[thin, black] (-2.0, 0.0) -- (-1.0, 0.0) -- (-1.0, 1.0) -- (-2.0, 1.0) -- (-2.0, 0.0);
 \draw [] (-2.0, 0.0) rectangle (-1.0, 1.0);

\draw[thin, black] (0.5, 0.5) -- (1.0, 0.5) -- (1.0, 1.0) -- (0.5, 1.0) -- (0.5, 0.5);
 \draw [] (0.5, 0.5) rectangle (1.0, 1.0);

\draw[thin, black] (2.5, 0.5) -- (3.0, 0.5) -- (3.0, 1.0) -- (2.5, 1.0) -- (2.5, 0.5);
 \draw [] (2.5, 0.5) rectangle (3.0, 1.0);

\draw[thin, black] (2.0, -1.0) -- (2.5, -1.0) -- (2.5, -0.5) -- (2.0, -0.5) -- (2.0, -1.0);
 \draw [] (2.0, -1.0) rectangle (2.5, -0.5);

\draw[thin, black] (0.0, -1.0) -- (0.5, -1.0) -- (0.5, -0.5) -- (0.0, -0.5) -- (0.0, -1.0);
 \draw [] (0.0, -1.0) rectangle (0.5, -0.5);

\draw[thin, black] (0.0, 0.5) -- (0.5, 0.5) -- (0.5, 1.0) -- (0.0, 1.0) -- (0.0, 0.5);
 \draw [] (0.0, 0.5) rectangle (0.5, 1.0);

\draw[thin, black] (2.0, 0.5) -- (2.5, 0.5) -- (2.5, 1.0) -- (2.0, 1.0) -- (2.0, 0.5);
 \draw [] (2.0, 0.5) rectangle (2.5, 1.0);

\draw[thin, black] (2.5, -1.0) -- (3.0, -1.0) -- (3.0, -0.5) -- (2.5, -0.5) -- (2.5, -1.0);
 \draw [] (2.5, -1.0) rectangle (3.0, -0.5);

\draw[thin, black] (0.5, -1.0) -- (1.0, -1.0) -- (1.0, -0.5) -- (0.5, -0.5) -- (0.5, -1.0);
 \draw [] (0.5, -1.0) rectangle (1.0, -0.5);

\draw[thin, black] (1.0, 1.0) -- (2.0, 1.0) -- (2.0, 2.0) -- (1.0, 2.0) -- (1.0, 1.0);
 \draw [] (1.0, 1.0) rectangle (2.0, 2.0);

\draw[thin, black] (3.0, 1.0) -- (4, 1.0) -- (4, 2.0) -- (3.0, 2.0) -- (3.0, 1.0);
 \draw [] (3.0, 1.0) rectangle (4, 2.0);

\draw[thin, black] (3.0, -2.0) -- (4, -2.0) -- (4, -1.0) -- (3.0, -1.0) -- (3.0, -2.0);
 \draw [] (3.0, -2.0) rectangle (4, -1.0);

\draw[thin, black] (1.0, -2.0) -- (2.0, -2.0) -- (2.0, -1.0) -- (1.0, -1.0) -- (1.0, -2.0);
 \draw [] (1.0, -2.0) rectangle (2.0, -1.0);

\draw[thin, black] (-1.0, -2.0) -- (0.0, -2.0) -- (0.0, -1.0) -- (-1.0, -1.0) -- (-1.0, -2.0);
 \draw [] (-1.0, -2.0) rectangle (0.0, -1.0);

\draw[thin, black] (-1.0, 1.0) -- (0.0, 1.0) -- (0.0, 2.0) -- (-1.0, 2.0) -- (-1.0, 1.0);
 \draw [] (-1.0, 1.0) rectangle (0.0, 2.0);

\draw[thin, black] (1.0, 0.0) -- (2.0, 0.0) -- (2.0, 1.0) -- (1.0, 1.0) -- (1.0, 0.0);
 \draw [] (1.0, 0.0) rectangle (2.0, 1.0);

\draw[thin, black] (3.0, 0.0) -- (4, 0.0) -- (4, 1.0) -- (3.0, 1.0) -- (3.0, 0.0);
 \draw [] (3.0, 0.0) rectangle (4, 1.0);

\draw[thin, black] (2.0, 1.0) -- (3.0, 1.0) -- (3.0, 2.0) -- (2.0, 2.0) -- (2.0, 1.0);
 \draw [] (2.0, 1.0) rectangle (3.0, 2.0);

\draw[thin, black] (2.0, -2.0) -- (3.0, -2.0) -- (3.0, -1.0) -- (2.0, -1.0) -- (2.0, -2.0);
 \draw [] (2.0, -2.0) rectangle (3.0, -1.0);

\draw[thin, black] (-1.0, -1.0) -- (0.0, -1.0) -- (0.0, 0.0) -- (-1.0, 0.0) -- (-1.0, -1.0);
 \draw [] (-1.0, -1.0) rectangle (0.0, 0.0);

\draw[thin, black] (0.0, -2.0) -- (1.0, -2.0) -- (1.0, -1.0) -- (0.0, -1.0) -- (0.0, -2.0);
 \draw [] (0.0, -2.0) rectangle (1.0, -1.0);

\draw[thin, black] (0.0, 1.0) -- (1.0, 1.0) -- (1.0, 2.0) -- (0.0, 2.0) -- (0.0, 1.0);
 \draw [] (0.0, 1.0) rectangle (1.0, 2.0);

\draw[thin, black] (1.0, -1.0) -- (2.0, -1.0) -- (2.0, 0.0) -- (1.0, 0.0) -- (1.0, -1.0);
 \draw [] (1.0, -1.0) rectangle (2.0, 0.0);

\draw[thin, black] (3.0, -1.0) -- (4, -1.0) -- (4, 0.0) -- (3.0, 0.0) -- (3.0, -1.0);
 \draw [] (3.0, -1.0) rectangle (4, 0.0);

\draw[thin, black] (-1.0, 0.0) -- (0.0, 0.0) -- (0.0, 1.0) -- (-1.0, 1.0) -- (-1.0, 0.0);
 \draw [] (-1.0, 0.0) rectangle (0.0, 1.0);

\draw[color=orange, line width = 1, dashed] plot[id=two_circle_nonintersect_before_f, raw gnuplot, smooth] function{
  f(x,y) = -0.76 - 0.8*x + x**2 + y**2 ;
  set xrange [-8:8];
  set yrange [-8:8];
  set view 0,0;
  set isosample 1000,1000;
  set size square;
  set cont base;
  set cntrparam levels incre 0,0.1,0;
  unset surface;
  splot f(x,y)};

\draw[color=cyan, line width = 1, dotted] plot[id=two_circle_nonintersect_before_g, raw gnuplot, smooth] function{
  g(x,y) = 5.885  - 5.2*x + x**2 + y**2;
  set xrange [-8:8];
  set yrange [-8:8];
  set view 0,0;
  set isosample 1000,1000;
  set size square;
  set cont base;
  set cntrparam levels incre 0,0.1,0;
  unset surface;
  splot g(x,y)};

\draw[red, line width = 1.5pt, dashed] (0.5, 0.75) -- (1.0, 0.75);

\draw[red, line width = 1.5pt, dashed] (0.5, -0.75) -- (0.0, -0.75);

\draw[red, line width = 1.5pt, dashed] (0.5, 0.75) -- (0.0, 0.75);

\draw[red, line width = 1.5pt, dashed] (0.5, -0.75) -- (1.0, -0.75);

\draw[red, line width = 1.5pt, dashed] (1.0, 0.75) -- (1.5, 0.0);

\draw[red, line width = 1.5pt, dashed] (0.0, -0.75) -- (-0.5, 0.0);

\draw[red, line width = 1.5pt, dashed] (1.0, -0.75) -- (1.5, 0.0);

\draw[red, line width = 1.5pt, dashed] (0.0, 0.75) -- (-0.5, 0.0);

\draw[blue, line width = 1.5pt, dotted] (2.5, 0.75) -- (3.0, 0.75);

\draw[blue, line width = 1.5pt, dotted] (2.5, -0.75) -- (2.0, -0.75);

\draw[blue, line width = 1.5pt, dotted] (2.5, 0.75) -- (2.0, 0.75);

\draw[blue, line width = 1.5pt, dotted] (2.5, -0.75) -- (3.0, -0.75);

\draw[blue, line width = 1.5pt, dotted] (2.0, 0.75) -- (1.5, 0.0);

\draw[blue, line width = 1.5pt, dotted] (3.0, 0.75) -- (3.5, 0.0);

\draw[blue, line width = 1.5pt, dotted] (2.0, -0.75) -- (1.5, 0.0);

\draw[blue, line width = 1.5pt, dotted] (3.0, -0.75) -- (3.5, 0.0);

\end{scope}
\node at (-4.7,-2.75) {(a)};
\node at (1.65,-2.75) {(b)};
\end{tikzpicture}

%% file: figures_source/plus_sign_figure.tex
\definecolor{green1}{RGB}{0, 102, 0}
\begin{tikzpicture}[scale = 2]
  \node at (0.8,0.7) {$\gamma_{f_1}$};
  \node at (0.2,.3) {$\gamma_{f_2}$};
  \draw[black, very thick] (0,0) -- (1,0) -- (1,1) -- (0,1) -- (0,0);
  \node[label=below:$B$] (node name) at (0.15, 1) {};
  \node[label=below:$B_E$] (node name) at (1.2, 1) {};
  \node[label=below:$B_N$] (node name) at (0.2, 2.1) {};
  \node[label=below:$B_W$] (node name) at (-0.75, 1) {};
  \node[label=below:$B_S$] (node name) at (0.2, -0.55) {};
  \draw[black, very thick] (0,0) -- (1,0) -- (1,-1) -- (0,-1) -- (0,0);
  \draw[black, very thick] (0,0) -- (2,0) -- (2,1) -- (0,1) -- (0,0);
  \draw[black, very thick] (0,0) -- (1,0) -- (1,2) -- (0,2) -- (0,0);
  \draw[black, very thick] (0,0) -- (-1,0) -- (-1,1) -- (0,1) -- (0,0);
  
  \begin{scope}
    \clip (0, 0) rectangle (1, 1);
    \draw[color=blue, line width = 1pt] plot[smooth, tension = .65] coordinates{(0.4, 2) (.45, 1) (.7, 1.) (0.5, 0)  (0.47, -0.61) (0.45, -1)  };
  \end{scope}

  \begin{scope}
    \clip (0.5, 1) rectangle (0.8, 1.1);
    \draw[color=blue, line width = 1pt] plot[smooth, tension = .65] coordinates{(0.4, 2) (.45, 1) (.7, 1.) (0.5, 0)  (0.47, -0.61) (0.45, -1)  };
  \end{scope}

  \begin{scope}
    \clip (0, 1) rectangle (0.57, 2);
    \draw[color=blue, line width = 1pt, dashed] plot[smooth, tension = .65] coordinates{(0.4, 2) (.45, 1) (.7, 1.) (0.5, 0)  (0.47, -0.61) (0.45, -1)  };
  \end{scope}

  \begin{scope}
    \clip (0, -1) rectangle (1, 0);
    \draw[color=blue, line width = 1pt, dashed] plot[smooth, tension = .65] coordinates{(0.4, 2) (.45, 1) (.7, 1.) (0.5, 0)  (0.47, -0.61) (0.45, -1)  };
  \end{scope}

\begin{scope}
  \clip (0, -1) rectangle (1, 1);
  \draw[color=red, line width = 1pt] plot[smooth, tension = .6] coordinates{(-1, .4) (0, .2) (.3, 0) (0.5, -0.1) (.65, 0) (1, .2) ( 2, .34) }; 
\end{scope}

\begin{scope}
  \clip (1, 0) rectangle (2, 1);
  \draw[color=red, line width = 1pt, dashed] plot[smooth, tension = .6] coordinates{(-1, .4) (0, .2) (.3, 0) (0.5, -0.1) (.65, 0) (1, .2) ( 2, .34) }; 
\end{scope}

\begin{scope}
  \clip (-1, 0) rectangle (0, 1);
  \draw[color=red, line width = 1pt, dashed] plot[smooth, tension = .6] coordinates{(-1, .4) (0, .2) (.3, 0) (0.5, -0.1) (.65, 0) (1, .2) ( 2, .34) }; 
\end{scope}

  \begin{scope}
    \clip (-1,-1) rectangle (1,0);
    \draw[dashed] (0.5, 0) circle(0.5);
  \end{scope}

  \begin{scope}
    \clip (1,0) rectangle (2,2);
    \draw[dashed] (1, 0.5) circle(0.5);
  \end{scope}

  \begin{scope}
    \clip (0,1) rectangle (1,2);
    \draw[dashed] (0.5, 1) circle(0.5);
  \end{scope}

  \begin{scope}
    \clip (-1,0) rectangle (0,1);
    \draw[dashed] (0, 0.5) circle(0.5);
  \end{scope}
\end{tikzpicture}

%
%
%
%

%% file: figures_source/missing_intersection_schematic.tex
\begin{tikzpicture}[scale=1.75]
	\draw[black, line width = 0.7, dotted] (2,1) -- (2,2) -- (1, 2) ;
	\draw[black, line width = 0.7, dotted] (1,-1) -- (2,-1) -- (2,0);
	\draw[black, line width = 0.7, dotted]  (0,2) -- (-1, 2) -- (-1,1);
	\draw[black, line width = 0.7, dotted] (-1,-1) -- (0,-1)  (-1, 0) -- (-1,-1);
	
	\draw[black, line width = 0.7, dotted] (0,-2) -- (1,-2) -- (1, -1)  (0, -1) -- (0,-2);
	\draw[black, line width = 0.7, dotted] (1,2) -- (1,3) -- (0, 3) -- (0,2);
	\draw[black, line width = 0.7] (0,-1) -- (1,-1) (1,0) -- (0, 0)  (0,-1);
	\draw[black, line width = 0.7] (0,0) -- (1,0) -- (1,1) -- (0, 1);
	\draw[black, line width = 0.7] (0,1)  (1,1) -- (1,2)  (0, 2)  (0,1);

	\draw[black, line width = 0.7, dotted]  (-1, 0) -- (-2, 0) -- (-2, 1) -- (-1, 1);
	\draw[black, line width = 0.7, dotted] (2,0) -- (3,0) -- (3,1) -- (2, 1);
	\draw[black, line width = 0.7]  (0,0) -- (0,1);
	\draw[black, line width = 0.7] (1, 1) -- (1,0);
	\draw[color = black, line width = 2, dashed] (0, 0) -- (0, -1)  (1, -1) -- (1, 0) -- (2, 0) -- (2,1);
	\draw[color = black, line width = 2, dashed] (2, 1) -- (1, 1)  (1, 2) -- (0, 2) -- (0, 1) -- (-1, 1) -- (-1, 0) -- (0,0);

    \draw[color = blue, line width = 1.5] (-0.5, -0.25) -- (-0.5, 0.25);
    \draw[color = blue, line width = 1.5] (1.75, 0.5) -- (2.25, 0.5);
	\draw[color=orange, line width = 1] plot[smooth, tension = .6] coordinates{(1, 1.7) (0.4, 1.2)  (.4, 0.3) (0, -0.2) (-0.08, -0.35) (0, -0.5) (0.3, -1)}; 
	
	\draw[color=orange, line width = 1, dashed] plot[smooth, tension = .6] coordinates{(1, 1.7) (1.4, 1.85) (2, 1.9)}; 
	
	\draw[color = red, line width = 1.5] (0.5, -1) -- (0.5, 0) -- (0.5, 1) -- (1, 1.5) -- (2, 1.5) ;
	
	\draw[color=cyan, line width = 1] plot[smooth, tension = .6] coordinates{(-0.8, 0) (0, 0.6) (2, 0.8) (1.9, 0.35) (2, 0.2)}; 
	\begin{scope}
		\clip (1,1) rectangle (2,2);
		\draw[dashed] (1, 1.5) circle(0.5);
	\end{scope}
	
	\begin{scope}
		\clip (0,-1) rectangle (1,0);
		\draw[dashed] (0.5, -1) circle(0.5);
	\end{scope}

	\node[circle, draw = red, fill = red, inner sep = 1.25pt] at (1, 1.5) {};
	\node[circle, draw = red, fill = red, inner sep = 1.25pt, label=below:$\alpha_{f_1}$] at (0.5, -1) {};
    \node[label=below:$\alpha_{f_2}$] at (-0.75, 0) {};
	\node[circle, draw = blue, fill = blue, inner sep = 1.25pt, label=below:] at (-0.5, 0) {};
	\node[circle, draw = blue, fill = blue, inner sep = 1.25pt, label=right:] at (2, 0.5) {};
	\node[label=left:$\gamma_{f_1}$] at (.42,0.2) {};
	\node[label=above:$\widetilde{\gamma}_{f_1}$] at (0.4,1.25) {};
	\node[label=below:$\widetilde{\gamma}_{f_2}$] at (1.25,1.9) {};
	
	\draw[color=cyan, line width = 1] plot[smooth, tension = .9] coordinates{(0.15, -1) (0.3, -0.88) (0.45, -1)}; 
	\draw[color=cyan, line width = 1] plot[smooth, tension = .9] coordinates{(1, 1.4) (1.1, 1.7) (1, 1.9) };
	
	\node[label=above:$\mathcal{N}_2(B)$] at (-1.3,0.95) {};
	\node[label=above:$\mathcal{N}(B)$] at (-0.3,0.95) {};
\end{tikzpicture}

%% file: figures_source/two_ellipse_tangential_intersect_before.tex
\begin{tikzpicture}[scale = 2.1,rotate=-90,transform shape] 
	\begin{scope}
	\clip (0,0) rectangle (2,4);
	\draw[thin, black] (1.5, 3.5) -- (2.0, 3.5) -- (2.0, 4) -- (1.5, 4) -- (1.5, 3.5); 
	\draw [] (1.5, 3.5) rectangle (2.0, 4);

	\draw[thin, black] (1.0, 3.5) -- (1.5, 3.5) -- (1.5, 4) -- (1.0, 4) -- (1.0, 3.5); 
	\draw [] (1.0, 3.5) rectangle (1.5, 4);

	\draw[thin, black] (0.75, 3.75) -- (1.0, 3.75) -- (1.0, 4) -- (0.75, 4) -- (0.75, 3.75); 
	\draw [] (0.75, 3.75) rectangle (1.0, 4);

	\draw[thin, black] (0.5, 3.75) -- (0.75, 3.75) -- (0.75, 4) -- (0.5, 4) -- (0.5, 3.75); 
	\draw [] (0.5, 3.75) rectangle (0.75, 4);

	\draw[thin, black] (0.5, 3.25) -- (0.75, 3.25) -- (0.75, 3.5) -- (0.5, 3.5) -- (0.5, 3.25); 
	\draw [] (0.5, 3.25) rectangle (0.75, 3.5);

	\draw[thin, black] (0.5, 3.0) -- (0.75, 3.0) -- (0.75, 3.25) -- (0.5, 3.25) -- (0.5, 3.0); 
	\draw [] (0.5, 3.0) rectangle (0.75, 3.25);

	\draw[thin, black] (0.5, 2.75) -- (0.75, 2.75) -- (0.75, 3.0) -- (0.5, 3.0) -- (0.5, 2.75); 
	\draw [] (0.5, 2.75) rectangle (0.75, 3.0);

	\draw[thin, black] (0.5, 2.5) -- (0.75, 2.5) -- (0.75, 2.75) -- (0.5, 2.75) -- (0.5, 2.5); 
	\draw [] (0.5, 2.5) rectangle (0.75, 2.75);

	\draw[thin, black] (0.5, 2.25) -- (0.75, 2.25) -- (0.75, 2.5) -- (0.5, 2.5) -- (0.5, 2.25); 
	\draw [] (0.5, 2.25) rectangle (0.75, 2.5);

	\draw[thin, black] (0.5, 2.0) -- (0.75, 2.0) -- (0.75, 2.25) -- (0.5, 2.25) -- (0.5, 2.0); 
	\draw [] (0.5, 2.0) rectangle (0.75, 2.25);

	\draw[thin, black] (1.0, 2.75) -- (1.25, 2.75) -- (1.25, 3.0) -- (1.0, 3.0) -- (1.0, 2.75); 
	\draw [] (1.0, 2.75) rectangle (1.25, 3.0);

	\draw[thin, black] (1.0, 2.5) -- (1.25, 2.5) -- (1.25, 2.75) -- (1.0, 2.75) -- (1.0, 2.5); 
	\draw [] (1.0, 2.5) rectangle (1.25, 2.75);

	\draw[thin, black] (1.0, 2.25) -- (1.25, 2.25) -- (1.25, 2.5) -- (1.0, 2.5) -- (1.0, 2.25); 
	\draw [] (1.0, 2.25) rectangle (1.25, 2.5);

	\draw[thin, black] (1.0, 2.0) -- (1.25, 2.0) -- (1.25, 2.25) -- (1.0, 2.25) -- (1.0, 2.0); 
	\draw [] (1.0, 2.0) rectangle (1.25, 2.25);

	\draw[thin, black] (1.25, 2.0) -- (1.5, 2.0) -- (1.5, 2.25) -- (1.25, 2.25) -- (1.25, 2.0); 
	\draw [] (1.25, 2.0) rectangle (1.5, 2.25);

	\draw[thin, black] (1.25, 1.75) -- (1.5, 1.75) -- (1.5, 2.0) -- (1.25, 2.0) -- (1.25, 1.75); 
	\draw [] (1.25, 1.75) rectangle (1.5, 2.0);

	\draw[thin, black] (1.0, 1.75) -- (1.25, 1.75) -- (1.25, 2.0) -- (1.0, 2.0) -- (1.0, 1.75); 
	\draw [] (1.0, 1.75) rectangle (1.25, 2.0);

	\draw[thin, black] (1.0, 1.5) -- (1.25, 1.5) -- (1.25, 1.75) -- (1.0, 1.75) -- (1.0, 1.5); 
	\draw [] (1.0, 1.5) rectangle (1.25, 1.75);

	\draw[thin, black] (1.25, 1.5) -- (1.5, 1.5) -- (1.5, 1.75) -- (1.25, 1.75) -- (1.25, 1.5); 
	\draw [] (1.25, 1.5) rectangle (1.5, 1.75);

	\draw[thin, black] (1.25, 1.25) -- (1.5, 1.25) -- (1.5, 1.5) -- (1.25, 1.5) -- (1.25, 1.25); 
	\draw [] (1.25, 1.25) rectangle (1.5, 1.5);

	\draw[thin, black] (1.0, 1.25) -- (1.25, 1.25) -- (1.25, 1.5) -- (1.0, 1.5) -- (1.0, 1.25); 
	\draw [] (1.0, 1.25) rectangle (1.25, 1.5);

	\draw[thin, black] (1.0, 1.0) -- (1.25, 1.0) -- (1.25, 1.25) -- (1.0, 1.25) -- (1.0, 1.0); 
	\draw [] (1.0, 1.0) rectangle (1.25, 1.25);

	\draw[thin, black] (0.5, 1.75) -- (0.75, 1.75) -- (0.75, 2.0) -- (0.5, 2.0) -- (0.5, 1.75); 
	\draw [] (0.5, 1.75) rectangle (0.75, 2.0);

	\draw[thin, black] (0.5, 1.5) -- (0.75, 1.5) -- (0.75, 1.75) -- (0.5, 1.75) -- (0.5, 1.5); 
	\draw [] (0.5, 1.5) rectangle (0.75, 1.75);

	\draw[thin, black] (0.5, 1.25) -- (0.75, 1.25) -- (0.75, 1.5) -- (0.5, 1.5) -- (0.5, 1.25); 
	\draw [] (0.5, 1.25) rectangle (0.75, 1.5);

	\draw[thin, black] (0.5, 1.0) -- (0.75, 1.0) -- (0.75, 1.25) -- (0.5, 1.25) -- (0.5, 1.0); 
	\draw [] (0.5, 1.0) rectangle (0.75, 1.25);

	\draw[thin, black] (0.5, 0.75) -- (0.75, 0.75) -- (0.75, 1.0) -- (0.5, 1.0) -- (0.5, 0.75); 
	\draw [] (0.5, 0.75) rectangle (0.75, 1.0);

	\draw[thin, black] (0.5, 0) -- (0.75, 0) -- (0.75, 0.25) -- (0.5, 0.25) -- (0.5, 0); 
	\draw [] (0.5, 0) rectangle (0.75, 0.25);

	\draw[thin, black] (1.0, 0.75) -- (1.25, 0.75) -- (1.25, 1.0) -- (1.0, 1.0) -- (1.0, 0.75); 
	\draw [] (1.0, 0.75) rectangle (1.25, 1.0);

	\draw[thin, black] (1.0, 0.5) -- (1.25, 0.5) -- (1.25, 0.75) -- (1.0, 0.75) -- (1.0, 0.5); 
	\draw [] (1.0, 0.5) rectangle (1.25, 0.75);

	\draw[thin, black] (1.125, 3.375) -- (1.25, 3.375) -- (1.25, 3.5) -- (1.125, 3.5) -- (1.125, 3.375); 
	\draw [] (1.125, 3.375) rectangle (1.25, 3.5);

	\draw[thin, black] (1.0, 3.375) -- (1.125, 3.375) -- (1.125, 3.5) -- (1.0, 3.5) -- (1.0, 3.375); 
	\draw [] (1.0, 3.375) rectangle (1.125, 3.5);

	\draw[thin, black] (1.125, 3.0) -- (1.25, 3.0) -- (1.25, 3.125) -- (1.125, 3.125) -- (1.125, 3.0); 
	\draw [] (1.125, 3.0) rectangle (1.25, 3.125);

	\draw[thin, black] (0.5, 0.625) -- (0.625, 0.625) -- (0.625, 0.75) -- (0.5, 0.75) -- (0.5, 0.625); 
	\draw [] (0.5, 0.625) rectangle (0.625, 0.75);

	\draw[thin, black] (0.5, 0.25) -- (0.625, 0.25) -- (0.625, 0.375) -- (0.5, 0.375) -- (0.5, 0.25); 
	\draw [] (0.5, 0.25) rectangle (0.625, 0.375);

	\draw[thin, black] (0.625, 0.25) -- (0.75, 0.25) -- (0.75, 0.375) -- (0.625, 0.375) -- (0.625, 0.25); 
	\draw [] (0.625, 0.25) rectangle (0.75, 0.375);

	\draw[thin, black] (1.125, 0.375) -- (1.25, 0.375) -- (1.25, 0.5) -- (1.125, 0.5) -- (1.125, 0.375); 
	\draw [] (1.125, 0.375) rectangle (1.25, 0.5);

	\draw[thin, black] (1.0, 0.375) -- (1.125, 0.375) -- (1.125, 0.5) -- (1.0, 0.5) -- (1.0, 0.375); 
	\draw [] (1.0, 0.375) rectangle (1.125, 0.5);

	\draw[thin, black] (1.125, 0.25) -- (1.25, 0.25) -- (1.25, 0.375) -- (1.125, 0.375) -- (1.125, 0.25); 
	\draw [] (1.125, 0.25) rectangle (1.25, 0.375);

	\draw[thin, black] (0.25, 3.75) -- (0.5, 3.75) -- (0.5, 4) -- (0.25, 4) -- (0.25, 3.75); 
	\draw [] (0.25, 3.75) rectangle (0.5, 4);

	\draw[thin, black] (0, 3.75) -- (0.25, 3.75) -- (0.25, 4) -- (0, 4) -- (0, 3.75); 
	\draw [] (0, 3.75) rectangle (0.25, 4);

	\draw[thin, black] (0, 3.5) -- (0.25, 3.5) -- (0.25, 3.75) -- (0, 3.75) -- (0, 3.5); 
	\draw [] (0, 3.5) rectangle (0.25, 3.75);

	\draw[thin, black] (0, 0.75) -- (0.25, 0.75) -- (0.25, 1.0) -- (0, 1.0) -- (0, 0.75); 
	\draw [] (0, 0.75) rectangle (0.25, 1.0);

	\draw[thin, black] (0, 0.5) -- (0.25, 0.5) -- (0.25, 0.75) -- (0, 0.75) -- (0, 0.5); 
	\draw [] (0, 0.5) rectangle (0.25, 0.75);

	\draw[thin, black] (0, 0.25) -- (0.25, 0.25) -- (0.25, 0.5) -- (0, 0.5) -- (0, 0.25); 
	\draw [] (0, 0.25) rectangle (0.25, 0.5);

	\draw[thin, black] (0, 0) -- (0.25, 0) -- (0.25, 0.25) -- (0, 0.25) -- (0, 0); 
	\draw [] (0, 0) rectangle (0.25, 0.25);

	\draw[thin, black] (0.25, 0) -- (0.5, 0) -- (0.5, 0.25) -- (0.25, 0.25) -- (0.25, 0); 
	\draw [] (0.25, 0) rectangle (0.5, 0.25);

	\draw[thin, black] (3.0, 1.0) -- (4, 1.0) -- (4, 2.0) -- (3.0, 2.0) -- (3.0, 1.0); 
	\draw [] (3.0, 1.0) rectangle (4, 2.0);

	\draw[thin, black] (2.0, 0) -- (3.0, 0) -- (3.0, 1.0) -- (2.0, 1.0) -- (2.0, 0); 
	\draw [] (2.0, 0) rectangle (3.0, 1.0);

	\draw[thin, black] (3.0, 0) -- (4, 0) -- (4, 1.0) -- (3.0, 1.0) -- (3.0, 0); 
	\draw [] (3.0, 0) rectangle (4, 1.0);

	\draw[thin, black] (3.0, 3.0) -- (4, 3.0) -- (4, 4) -- (3.0, 4) -- (3.0, 3.0); 
	\draw [] (3.0, 3.0) rectangle (4, 4);

	\draw[thin, black] (2.0, 3.0) -- (3.0, 3.0) -- (3.0, 4) -- (2.0, 4) -- (2.0, 3.0); 
	\draw [] (2.0, 3.0) rectangle (3.0, 4);

	\draw[thin, black] (3.0, 2.0) -- (4, 2.0) -- (4, 3.0) -- (3.0, 3.0) -- (3.0, 2.0); 
	\draw [] (3.0, 2.0) rectangle (4, 3.0);

	\draw[thin, black] (1.0, 3.25) -- (1.125, 3.25) -- (1.125, 3.375) -- (1.0, 3.375) -- (1.0, 3.25); 
	\draw [] (1.0, 3.25) rectangle (1.125, 3.375);

	\draw[thin, black] (1.125, 3.125) -- (1.25, 3.125) -- (1.25, 3.25) -- (1.125, 3.25) -- (1.125, 3.125); 
	\draw [] (1.125, 3.125) rectangle (1.25, 3.25);

	\draw[thin, black] (1.0, 3.0) -- (1.125, 3.0) -- (1.125, 3.125) -- (1.0, 3.125) -- (1.0, 3.0); 
	\draw [] (1.0, 3.0) rectangle (1.125, 3.125);

	\draw[thin, black] (0.625, 3.625) -- (0.75, 3.625) -- (0.75, 3.75) -- (0.625, 3.75) -- (0.625, 3.625); 
	\draw [] (0.625, 3.625) rectangle (0.75, 3.75);

	\draw[thin, black] (0.5, 3.5) -- (0.625, 3.5) -- (0.625, 3.625) -- (0.5, 3.625) -- (0.5, 3.5); 
	\draw [] (0.5, 3.5) rectangle (0.625, 3.625);

	\draw[thin, black] (0.625, 0.625) -- (0.75, 0.625) -- (0.75, 0.75) -- (0.625, 0.75) -- (0.625, 0.625); 
	\draw [] (0.625, 0.625) rectangle (0.75, 0.75);

	\draw[thin, black] (0.5, 0.5) -- (0.625, 0.5) -- (0.625, 0.625) -- (0.5, 0.625) -- (0.5, 0.5); 
	\draw [] (0.5, 0.5) rectangle (0.625, 0.625);

	\draw[thin, black] (0.625, 0.375) -- (0.75, 0.375) -- (0.75, 0.5) -- (0.625, 0.5) -- (0.625, 0.375); 
	\draw [] (0.625, 0.375) rectangle (0.75, 0.5);

	\draw[thin, black] (1.0, 0.25) -- (1.125, 0.25) -- (1.125, 0.375) -- (1.0, 0.375) -- (1.0, 0.25); 
	\draw [] (1.0, 0.25) rectangle (1.125, 0.375);

	\draw[thin, black] (1.125, 0.125) -- (1.25, 0.125) -- (1.25, 0.25) -- (1.125, 0.25) -- (1.125, 0.125); 
	\draw [] (1.125, 0.125) rectangle (1.25, 0.25);

	\draw[thin, black] (1.0, 0) -- (1.125, 0) -- (1.125, 0.125) -- (1.0, 0.125) -- (1.0, 0); 
	\draw [] (1.0, 0) rectangle (1.125, 0.125);

	\draw[thin, black] (1.125, 3.25) -- (1.25, 3.25) -- (1.25, 3.375) -- (1.125, 3.375) -- (1.125, 3.25); 
	\draw [] (1.125, 3.25) rectangle (1.25, 3.375);

	\draw[thin, black] (1.0, 3.125) -- (1.125, 3.125) -- (1.125, 3.25) -- (1.0, 3.25) -- (1.0, 3.125); 
	\draw [] (1.0, 3.125) rectangle (1.125, 3.25);

	\draw[thin, black] (0.5, 3.625) -- (0.625, 3.625) -- (0.625, 3.75) -- (0.5, 3.75) -- (0.5, 3.625); 
	\draw [] (0.5, 3.625) rectangle (0.625, 3.75);

	\draw[thin, black] (0.625, 3.5) -- (0.75, 3.5) -- (0.75, 3.625) -- (0.625, 3.625) -- (0.625, 3.5); 
	\draw [] (0.625, 3.5) rectangle (0.75, 3.625);

	\draw[thin, black] (0.5, 0.375) -- (0.625, 0.375) -- (0.625, 0.5) -- (0.5, 0.5) -- (0.5, 0.375); 
	\draw [] (0.5, 0.375) rectangle (0.625, 0.5);

	\draw[thin, black] (0.625, 0.5) -- (0.75, 0.5) -- (0.75, 0.625) -- (0.625, 0.625) -- (0.625, 0.5); 
	\draw [] (0.625, 0.5) rectangle (0.75, 0.625);

	\draw[thin, black] (1.125, 0) -- (1.25, 0) -- (1.25, 0.125) -- (1.125, 0.125) -- (1.125, 0); 
	\draw [] (1.125, 0) rectangle (1.25, 0.125);

	\draw[thin, black] (1.0, 0.125) -- (1.125, 0.125) -- (1.125, 0.25) -- (1.0, 0.25) -- (1.0, 0.125); 
	\draw [] (1.0, 0.125) rectangle (1.125, 0.25);

	\draw[thin, black] (0.75, 2.75) -- (1.0, 2.75) -- (1.0, 3.0) -- (0.75, 3.0) -- (0.75, 2.75); 
	\draw [] (0.75, 2.75) rectangle (1.0, 3.0);

	\draw[thin, black] (0.75, 2.25) -- (1.0, 2.25) -- (1.0, 2.5) -- (0.75, 2.5) -- (0.75, 2.25); 
	\draw [] (0.75, 2.25) rectangle (1.0, 2.5);

	\draw[thin, black] (1.25, 2.75) -- (1.5, 2.75) -- (1.5, 3.0) -- (1.25, 3.0) -- (1.25, 2.75); 
	\draw [] (1.25, 2.75) rectangle (1.5, 3.0);

	\draw[thin, black] (1.25, 2.25) -- (1.5, 2.25) -- (1.5, 2.5) -- (1.25, 2.5) -- (1.25, 2.25); 
	\draw [] (1.25, 2.25) rectangle (1.5, 2.5);

	\draw[thin, black] (1.25, 1.0) -- (1.5, 1.0) -- (1.5, 1.25) -- (1.25, 1.25) -- (1.25, 1.0); 
	\draw [] (1.25, 1.0) rectangle (1.5, 1.25);

	\draw[thin, black] (0.75, 1.75) -- (1.0, 1.75) -- (1.0, 2.0) -- (0.75, 2.0) -- (0.75, 1.75); 
	\draw [] (0.75, 1.75) rectangle (1.0, 2.0);

	\draw[thin, black] (0.75, 1.25) -- (1.0, 1.25) -- (1.0, 1.5) -- (0.75, 1.5) -- (0.75, 1.25); 
	\draw [] (0.75, 1.25) rectangle (1.0, 1.5);

	\draw[thin, black] (0.75, 0.75) -- (1.0, 0.75) -- (1.0, 1.0) -- (0.75, 1.0) -- (0.75, 0.75); 
	\draw [] (0.75, 0.75) rectangle (1.0, 1.0);

	\draw[thin, black] (1.25, 0.5) -- (1.5, 0.5) -- (1.5, 0.75) -- (1.25, 0.75) -- (1.25, 0.5); 
	\draw [] (1.25, 0.5) rectangle (1.5, 0.75);

	\draw[thin, black] (0.75, 3.25) -- (1.0, 3.25) -- (1.0, 3.5) -- (0.75, 3.5) -- (0.75, 3.25); 
	\draw [] (0.75, 3.25) rectangle (1.0, 3.5);

	\draw[thin, black] (0.25, 0.75) -- (0.5, 0.75) -- (0.5, 1.0) -- (0.25, 1.0) -- (0.25, 0.75); 
	\draw [] (0.25, 0.75) rectangle (0.5, 1.0);

	\draw[thin, black] (0.75, 0.25) -- (1.0, 0.25) -- (1.0, 0.5) -- (0.75, 0.5) -- (0.75, 0.25); 
	\draw [] (0.75, 0.25) rectangle (1.0, 0.5);

	\draw[thin, black] (1.25, 3.25) -- (1.5, 3.25) -- (1.5, 3.5) -- (1.25, 3.5) -- (1.25, 3.25); 
	\draw [] (1.25, 3.25) rectangle (1.5, 3.5);

	\draw[thin, black] (0.25, 3.5) -- (0.5, 3.5) -- (0.5, 3.75) -- (0.25, 3.75) -- (0.25, 3.5); 
	\draw [] (0.25, 3.5) rectangle (0.5, 3.75);

	\draw[thin, black] (0.25, 0.25) -- (0.5, 0.25) -- (0.5, 0.5) -- (0.25, 0.5) -- (0.25, 0.25); 
	\draw [] (0.25, 0.25) rectangle (0.5, 0.5);

	\draw[thin, black] (1.25, 0) -- (1.5, 0) -- (1.5, 0.25) -- (1.25, 0.25) -- (1.25, 0); 
	\draw [] (1.25, 0) rectangle (1.5, 0.25);

	\draw[thin, black] (0.75, 2.5) -- (1.0, 2.5) -- (1.0, 2.75) -- (0.75, 2.75) -- (0.75, 2.5); 
	\draw [] (0.75, 2.5) rectangle (1.0, 2.75);

	\draw[thin, black] (1.25, 2.5) -- (1.5, 2.5) -- (1.5, 2.75) -- (1.25, 2.75) -- (1.25, 2.5); 
	\draw [] (1.25, 2.5) rectangle (1.5, 2.75);

	\draw[thin, black] (0.75, 2.0) -- (1.0, 2.0) -- (1.0, 2.25) -- (0.75, 2.25) -- (0.75, 2.0); 
	\draw [] (0.75, 2.0) rectangle (1.0, 2.25);

	\draw[thin, black] (0.75, 1.5) -- (1.0, 1.5) -- (1.0, 1.75) -- (0.75, 1.75) -- (0.75, 1.5); 
	\draw [] (0.75, 1.5) rectangle (1.0, 1.75);

	\draw[thin, black] (0.75, 1.0) -- (1.0, 1.0) -- (1.0, 1.25) -- (0.75, 1.25) -- (0.75, 1.0); 
	\draw [] (0.75, 1.0) rectangle (1.0, 1.25);

	\draw[thin, black] (1.25, 0.75) -- (1.5, 0.75) -- (1.5, 1.0) -- (1.25, 1.0) -- (1.25, 0.75); 
	\draw [] (1.25, 0.75) rectangle (1.5, 1.0);

	\draw[thin, black] (0.75, 3.5) -- (1.0, 3.5) -- (1.0, 3.75) -- (0.75, 3.75) -- (0.75, 3.5); 
	\draw [] (0.75, 3.5) rectangle (1.0, 3.75);

	\draw[thin, black] (0.75, 3.0) -- (1.0, 3.0) -- (1.0, 3.25) -- (0.75, 3.25) -- (0.75, 3.0); 
	\draw [] (0.75, 3.0) rectangle (1.0, 3.25);

	\draw[thin, black] (0.75, 0) -- (1.0, 0) -- (1.0, 0.25) -- (0.75, 0.25) -- (0.75, 0); 
	\draw [] (0.75, 0) rectangle (1.0, 0.25);

	\draw[thin, black] (0.75, 0.5) -- (1.0, 0.5) -- (1.0, 0.75) -- (0.75, 0.75) -- (0.75, 0.5); 
	\draw [] (0.75, 0.5) rectangle (1.0, 0.75);

	\draw[thin, black] (1.25, 3.0) -- (1.5, 3.0) -- (1.5, 3.25) -- (1.25, 3.25) -- (1.25, 3.0); 
	\draw [] (1.25, 3.0) rectangle (1.5, 3.25);

	\draw[thin, black] (0.25, 0.5) -- (0.5, 0.5) -- (0.5, 0.75) -- (0.25, 0.75) -- (0.25, 0.5); 
	\draw [] (0.25, 0.5) rectangle (0.5, 0.75);

	\draw[thin, black] (1.25, 0.25) -- (1.5, 0.25) -- (1.5, 0.5) -- (1.25, 0.5) -- (1.25, 0.25); 
	\draw [] (1.25, 0.25) rectangle (1.5, 0.5);

	\draw[thin, black] (0, 2.5) -- (0.5, 2.5) -- (0.5, 3.0) -- (0, 3.0) -- (0, 2.5); 
	\draw [] (0, 2.5) rectangle (0.5, 3.0);

	\draw[thin, black] (1.5, 1.5) -- (2.0, 1.5) -- (2.0, 2.0) -- (1.5, 2.0) -- (1.5, 1.5); 
	\draw [] (1.5, 1.5) rectangle (2.0, 2.0);

	\draw[thin, black] (0, 1.5) -- (0.5, 1.5) -- (0.5, 2.0) -- (0, 2.0) -- (0, 1.5); 
	\draw [] (0, 1.5) rectangle (0.5, 2.0);

	\draw[thin, black] (1.5, 2.5) -- (2.0, 2.5) -- (2.0, 3.0) -- (1.5, 3.0) -- (1.5, 2.5); 
	\draw [] (1.5, 2.5) rectangle (2.0, 3.0);

	\draw[thin, black] (1.5, 0.5) -- (2.0, 0.5) -- (2.0, 1.0) -- (1.5, 1.0) -- (1.5, 0.5); 
	\draw [] (1.5, 0.5) rectangle (2.0, 1.0);

	\draw[thin, black] (0, 3.0) -- (0.5, 3.0) -- (0.5, 3.5) -- (0, 3.5) -- (0, 3.0); 
	\draw [] (0, 3.0) rectangle (0.5, 3.5);

	\draw[thin, black] (0, 1.0) -- (0.5, 1.0) -- (0.5, 1.5) -- (0, 1.5) -- (0, 1.0); 
	\draw [] (0, 1.0) rectangle (0.5, 1.5);

	\draw[thin, black] (0, 2.0) -- (0.5, 2.0) -- (0.5, 2.5) -- (0, 2.5) -- (0, 2.0); 
	\draw [] (0, 2.0) rectangle (0.5, 2.5);

	\draw[thin, black] (1.5, 3.0) -- (2.0, 3.0) -- (2.0, 3.5) -- (1.5, 3.5) -- (1.5, 3.0); 
	\draw [] (1.5, 3.0) rectangle (2.0, 3.5);

	\draw[thin, black] (1.5, 2.0) -- (2.0, 2.0) -- (2.0, 2.5) -- (1.5, 2.5) -- (1.5, 2.0); 
	\draw [] (1.5, 2.0) rectangle (2.0, 2.5);

	\draw[thin, black] (1.5, 0) -- (2.0, 0) -- (2.0, 0.5) -- (1.5, 0.5) -- (1.5, 0); 
	\draw [] (1.5, 0) rectangle (2.0, 0.5);

	\draw[thin, black] (1.5, 1.0) -- (2.0, 1.0) -- (2.0, 1.5) -- (1.5, 1.5) -- (1.5, 1.0); 
	\draw [] (1.5, 1.0) rectangle (2.0, 1.5);

	\draw[thin, black] (2.0, 2.0) -- (3.0, 2.0) -- (3.0, 3.0) -- (2.0, 3.0) -- (2.0, 2.0); 
	\draw [] (2.0, 2.0) rectangle (3.0, 3.0);

	\draw[thin, black] (2.0, 1.0) -- (3.0, 1.0) -- (3.0, 2.0) -- (2.0, 2.0) -- (2.0, 1.0); 
	\draw [] (2.0, 1.0) rectangle (3.0, 2.0); 

 \draw[color=orange, line width = 1, dashed] plot[id=posalpha, raw gnuplot, smooth] function{
		f(x,y) = 8.0*x**2 - 9.6*x + 0.4*y**2 - 1.68*y + 3.644;
		set xrange [0:4];
		set yrange [0:4];
		set view 0,0;
		set isosample 3000,3000;
		set size square;
		set cont base;
		set cntrparam levels incre 0,0.1,0;
		unset surface;
		splot f(x,y)};
	
	\draw[color=cyan, line width = 1, dotted] plot[id=posalpha, raw gnuplot, smooth] function{
		g(x,y) = 8.0*x**2 - 18.88*x + 0.4*y**2 - 1.36*y + 11.2952;
		set xrange [0:4];
		set yrange [0:4];
		set view 0,0;
		set isosample 5000,5000;
		set size square;
		set cont base;
		set cntrparam levels incre 0,0.1,0;
		unset surface;
		splot g(x,y)};
	
	\draw[red, line width = 1.5pt, dashed] (0.625, 3.6875) -- (0.6875, 3.625); 
	
	\draw[red, line width = 1.5pt, dashed] (0.5625, 3.625) -- (0.5, 3.5625); 
	
	\draw[red, line width = 1.5pt, dashed] (0.6875, 0.625) -- (0.75, 0.6875); 
	
	\draw[red, line width = 1.5pt, dashed] (0.625, 0.5625) -- (0.5, 0.5625); 
	
	\draw[red, line width = 1.5pt, dashed] (0.625, 3.6875) -- (0.5625, 3.625); 
	
	\draw[red, line width = 1.5pt, dashed] (0.6875, 3.625) -- (0.75, 3.5625); 
	
	\draw[red, line width = 1.5pt, dashed] (0.6875, 0.625) -- (0.625, 0.5625); 
	
	\draw[red, line width = 1.5pt, dashed] (0.875, 2.75) -- (0.875, 3.0); 
	
	\draw[red, line width = 1.5pt, dashed] (0.875, 2.5) -- (0.875, 2.25); 
	
	\draw[red, line width = 1.5pt, dashed] (0.875, 2.0) -- (0.875, 1.75); 
	
	\draw[red, line width = 1.5pt, dashed] (0.875, 1.5) -- (0.875, 1.25); 
	
	\draw[red, line width = 1.5pt, dashed] (0.875, 1.0) -- (0.875, 0.75); 
	
	\draw[red, line width = 1.5pt, dashed] (0.875, 3.5) -- (0.875, 3.25); 
	
	\draw[red, line width = 1.5pt, dashed] (0.375, 0.75) -- (0.375, 1.0); 
	
	\draw[red, line width = 1.5pt, dashed] (0.5, 3.5625) -- (0.375, 3.5); 
	
	\draw[red, line width = 1.5pt, dashed] (0.875, 2.75) -- (0.875, 2.5); 
	
	\draw[red, line width = 1.5pt, dashed] (0.875, 2.25) -- (0.875, 2.0); 
	
	\draw[red, line width = 1.5pt, dashed] (0.875, 1.75) -- (0.875, 1.5); 
	
	\draw[red, line width = 1.5pt, dashed] (0.875, 1.25) -- (0.875, 1.0); 
	
	\draw[red, line width = 1.5pt, dashed] (0.75, 3.5625) -- (0.875, 3.5); 
	
	\draw[red, line width = 1.5pt, dashed] (0.875, 3.0) -- (0.875, 3.25); 
	
	\draw[red, line width = 1.5pt, dashed] (0.75, 0.6875) -- (0.875, 0.75); 
	
	\draw[red, line width = 1.5pt, dashed] (0.5, 0.5625) -- (0.375, 0.75); 
	
	\draw[red, line width = 1.5pt, dashed] (0.25, 3.0) -- (0.25, 2.5); 
	
	\draw[red, line width = 1.5pt, dashed] (0.25, 1.5) -- (0.25, 2.0); 
	
	\draw[red, line width = 1.5pt, dashed] (0.375, 3.5) -- (0.25, 3.0); 
	
	\draw[red, line width = 1.5pt, dashed] (0.375, 1.0) -- (0.25, 1.5); 
	
	\draw[red, line width = 1.5pt, dashed] (0.25, 2.5) -- (0.25, 2.0); 
	
	\draw[blue, line width = 1.5pt, dotted] (1.125, 3.3125) -- (1.0625, 3.25); 
	
	\draw[blue, line width = 1.5pt, dotted] (1.1875, 3.25) -- (1.25, 3.1875); 
	
	\draw[blue, line width = 1.5pt, dotted] (1.0625, 3.125) -- (1.0, 3.0625); 
	
	\draw[blue, line width = 1.5pt, dotted] (1.0625, 0.25) -- (1.0, 0.3125); 
	
	\draw[blue, line width = 1.5pt, dotted] (1.125, 0.1875) -- (1.25, 0.1875); 
	
	\draw[blue, line width = 1.5pt, dotted] (1.125, 3.3125) -- (1.1875, 3.25); 
	
	\draw[blue, line width = 1.5pt, dotted] (1.0625, 3.25) -- (1.0625, 3.125); 
	
	\draw[blue, line width = 1.5pt, dotted] (1.0625, 0.25) -- (1.125, 0.1875); 
	
	\draw[blue, line width = 1.5pt, dotted] (0.875, 2.75) -- (0.875, 3.0); 
	
	\draw[blue, line width = 1.5pt, dotted] (0.875, 2.5) -- (0.875, 2.25); 
	
	\draw[blue, line width = 1.5pt, dotted] (1.375, 2.75) -- (1.375, 3.0); 
	
	\draw[blue, line width = 1.5pt, dotted] (1.375, 2.5) -- (1.5, 2.375); 
	
	\draw[blue, line width = 1.5pt, dotted] (1.375, 1.0) -- (1.5, 1.125); 
	
	\draw[blue, line width = 1.5pt, dotted] (0.875, 2.0) -- (0.875, 1.75); 
	
	\draw[blue, line width = 1.5pt, dotted] (0.875, 1.5) -- (0.875, 1.25); 
	
	\draw[blue, line width = 1.5pt, dotted] (0.875, 1.0) -- (0.875, 0.75); 
	
	\draw[blue, line width = 1.5pt, dotted] (1.375, 0.75) -- (1.375, 0.5); 
	
	\draw[blue, line width = 1.5pt, dotted] (1.0, 0.3125) -- (0.875, 0.5); 
	
	\draw[blue, line width = 1.5pt, dotted] (1.25, 0.1875) -- (1.375, 0.25); 
	
	\draw[blue, line width = 1.5pt, dotted] (0.875, 2.75) -- (0.875, 2.5); 
	
	\draw[blue, line width = 1.5pt, dotted] (1.375, 2.75) -- (1.375, 2.5); 
	
	\draw[blue, line width = 1.5pt, dotted] (0.875, 2.25) -- (0.875, 2.0); 
	
	\draw[blue, line width = 1.5pt, dotted] (0.875, 1.75) -- (0.875, 1.5); 
	
	\draw[blue, line width = 1.5pt, dotted] (0.875, 1.25) -- (0.875, 1.0); 
	
	\draw[blue, line width = 1.5pt, dotted] (1.375, 1.0) -- (1.375, 0.75); 
	
	\draw[blue, line width = 1.5pt, dotted] (1.0, 3.0625) -- (0.875, 3.0); 
	
	\draw[blue, line width = 1.5pt, dotted] (0.875, 0.75) -- (0.875, 0.5); 
	
	\draw[blue, line width = 1.5pt, dotted] (1.25, 3.1875) -- (1.375, 3.0); 
	
	\draw[blue, line width = 1.5pt, dotted] (1.375, 0.5) -- (1.375, 0.25); 
	
	\draw[blue, line width = 1.5pt, dotted] (1.75, 2.0) -- (1.75, 1.5); 
	
	\draw[blue, line width = 1.5pt, dotted] (1.5, 2.375) -- (1.75, 2.0); 
	
	\draw[blue, line width = 1.5pt, dotted] (1.5, 1.125) -- (1.75, 1.5); 
	\end{scope}
\end{tikzpicture}

%% file: figures_source/snake_schematic.tex
\definecolor{purple1}{RGB}{147, 112, 219}
\begin{tikzpicture}[scale=1.5]

	\draw[black] (-1,0) -- (0,0) -- (0,1) -- (-1, 1) -- (-1,0);
	\draw[black] (0,0) -- (1,0) -- (1,1) -- (0, 1) -- (0,0);
	\draw[black] (0,1) -- (1,1) -- (1,2) -- (0, 2) -- (0,1);
	\draw[black] (1,1) -- (2,1) -- (2,2) -- (1, 2) -- (1,1);
	\draw[black] (2,1) -- (3,1) -- (3,2) -- (2, 2) -- (2,1);
	\draw[black] (3,1) -- (4,1) -- (4,2) -- (3, 2) -- (3,1);
	\draw[black] (3,0) -- (4,0) -- (4,1) -- (3, 1) -- (3,0);
	\draw[black] (4,0) -- (5,0) -- (5,1) -- (4, 1) -- (4,0);
	\begin{scope}
    \clip (-1, 0) rectangle (5, 2);
	\draw[color=orange, line width = .7] plot[smooth, tension = .6] coordinates{(-1, -0.4) (1, 1.4) (3.1, 1.2) (5, -0.4) }; 
	
	\draw[color=cyan, line width = .7] plot[smooth, tension = .5] coordinates{(-2, 0.2) (0.2, 0.5)  (1, 1.2) (2, 1.3)  (3.7, 0.85)(6, 0.45) }; 
    \end{scope}
	\draw[red, line width = 1.5pt, dashed] (-0.5, 0) -- (0, 0.5) (4, 0.5) -- (4.5, 0);
	\draw[blue, line width = 1.5pt, dotted] (-1, 0.5) -- (0, 0.5) (4, 0.5) -- (5, 0.5);
	
	\draw[purple1, line width = 2pt] (0, 0.5) -- (0.5, 1) -- (1, 1.5) -- (2, 1.5) -- (3, 1.5) -- (3.5, 1) -- (4, 0.5);

    \node[circle, draw = blue, fill = blue, inner sep = 1.25pt, label=left:$p_1$] at (-1, 0.5) {};
    \node[circle, draw = blue, fill = blue, inner sep = 1.25pt, label=right:$p_2$] at (5, 0.5) {};
    \node[circle, draw = red, fill = red, inner sep = 1.25pt, label=below:$q_1$] at (-0.5, 0) {};
    \node[circle, draw = red, fill = red, inner sep = 1.25pt, label=below:$q_2$] at (4.5, 0) {};
\end{tikzpicture}

%% file: figures_source/snake_proof_schematic.tex
\definecolor{purple1}{RGB}{147, 112, 219}
	  \definecolor{green1}{RGB}{0, 102, 0}
	  \begin{tikzpicture}[scale = 1.75]
	  	\draw[black] (-2,0) -- (-1,0) -- (-1,1) -- (-2, 1) -- (-2,0);
	  	\draw[black] (-1,0) -- (0,0) -- (0,1) -- (-1, 1) -- (-1,0);
	  	\draw[black] (0,0) -- (1,0) -- (1,1) -- (0, 1) -- (0,0);
	  	\draw[black] (0,1) -- (1,1) -- (1,2) -- (0, 2) -- (0,1);
	  	\draw[black] (1,1) -- (2,1) -- (2,2) -- (1, 2) -- (1,1);
	  	\draw[black] (2,1) -- (3,1) -- (3,2) -- (2, 2) -- (2,1);
	  	\draw[black] (2,0) -- (3,0) -- (3,1) -- (2, 1) -- (2,0);
	  	\draw[black] (3,0) -- (4,0) -- (4,1) -- (3, 1) -- (3,0);

	  	\node[label=above:$\gamma_{f_2}$] at (2.52,0.1) {};
	  	\node[label=below:$\gamma_{f_1}$] at (1.57,1.45) {};
	  	
	  	\node[label=below:$\gamma'_{f_2}$] at (-1.25,1.1) {};
	  	\node[label=below:$\gamma'_{f_1}$] at (-0.75,0.45) {};
	  	\node[label=below:$\gamma'_{f_1}$] at (0.6,0.8) {};
	  	
	  	\node[label=below:$\del {f_2}$] at (0.3,1.85) {};
	  	\node[label=above:$\gamma'_{f_2}$] at (1.4,1.4) {};
	  	\node[label=above:$\overline{\gamma}_{f_1}$] at (3.52,0.25) {};
	  	\begin{scope}
	  		\clip (-2, 0) rectangle (4, 2);
	  		\draw[color=orange, line width = .75] plot[smooth, tension = .6] coordinates{(-1.8, -1) (-0.8, 1) (-0.2, 1) (0.3, 0.85) (1, 1.4) (2.1, 1.2) (3.8, -1) }; 
	  		
	  		\draw[color=cyan, line width = .75] plot[smooth, tension = .6] coordinates{(-3, 0.2) (0.2, 0.5)  (1, 1.2) (1.5, 1.35)  (2.4, 1.25) (2.9, 0.9) (3.8, 0.9) (3.4, 0.7) (5, 0.65) (5, .35) (4.7, 0) }; 
	  	\end{scope}
  	
  		\begin{scope}
  			\clip (-2.43, -0.5) rectangle (4.47, 2);
  			\draw[color=orange, line width = .75, dashed] plot[smooth, tension = .6] coordinates{(-1.8, -1) (-0.8, 1) (-0.2, 1) (0.3, 0.85) (1, 1.4) (2.1, 1.2) (3.8, -1) }; 
  			
  			\draw[color=cyan, line width = .75, dashed] plot[smooth, tension = .6] coordinates{(-3, 0.2) (0.2, 0.5)  (1, 1.2) (1.5, 1.35)  (2.4, 1.25) (2.9, 0.9) (3.8, 0.9) (3.4, 0.7) (5, 0.65) (5, .35) (4.7, 0) }; 
  		\end{scope}
	  	
	  	\begin{scope}[shift = {(0.01,0)}]
	  		\draw[line width=1pt,black,-stealth](-1.2, 0.325)--(-0.6, 0.37) node[anchor=south west]{};
	  		\draw[line width=1pt,black,-stealth](-1.2, 0.325)--(-0.9, 0.9) node[anchor=south west]{};
	  		\draw[red, line width = 0.8pt] ($(-1.2, 0.325)+({0.4*cos(5)},{0.4*sin(5)})$) arc (5:28:0.4);
	  		\draw[red, line width = 0.8pt, <-] ($(-1.2, 0.325)+({0.4*cos(25)},{0.4*sin(25)})$) arc (25:63:0.4);
	  		
	  		\draw[line width=1pt,black,-stealth](1.91, 1.337)--(2.5, 1.28) node[anchor=south west]{};
	  		\draw[line width=1pt,black,-stealth](1.91, 1.337)--(2.42, 0.94) node[anchor=south west]{};
	  		\draw[red, line width = 0.8pt, ->] ($(1.91, 1.337)+({0.35*cos(-40)},{0.35*sin(-40)})$) arc (-40:-5:0.35);
	  		
	  	\end{scope}
	  	
	  	\draw[line width=1pt,black,-stealth](0.75, 1.27)--(1.32, 1.63) node[anchor=south west]{};
	  	\draw[line width=1pt,black,-stealth](0.75, 1.27)--(0.39, 1.84) node[anchor=south west]{};
	  	
	  	\draw[line width=1pt,black,-stealth](0.35, 0.535)--(0.96, 0.85) node[anchor=south west]{};

	  	\node[circle, draw = blue, fill = blue, inner sep = 1.25pt, label=right:$p_1$] at (-2, 0.5) {};
	  	\node[circle, draw = blue, fill = blue, inner sep = 1.25pt, label=right:$p_2$] at (4, 0.5) {};
	  	\node[circle, draw = red, fill = red, inner sep = 1.25pt, label=below:$q_1\quad$] at (-1.5, 0) {};
	  	\node[circle, draw = red, fill = red, inner sep = 1.25pt, label=below:$q_2$] at (3.5, 0) {};
	  	
	  	\begin{scope}
	  		\clip (-3,0) rectangle (-2,1);
	  		\draw[dashed] (-2, 0.5) circle(0.5);
	  	\end{scope}
	  	\begin{scope}
	  		\clip (-2,1) rectangle (-1,2);
	  		\draw[dashed] (-1.5, 1) circle(0.5);
	  	\end{scope}
	  	\begin{scope}
	  		\clip (-1,1) rectangle (0,2);
	  		\draw[dashed] (-0.5, 1) circle(0.5);
	  	\end{scope}
	  	\begin{scope}
	  		\clip (0,2) rectangle (1,3);
	  		\draw[dashed] (0.5, 2) circle(0.5);
	  	\end{scope}
	  	\begin{scope}
	  		\clip (1,2) rectangle (2,3);
	  		\draw[dashed] (1.5, 2) circle(0.5);
	  	\end{scope}
	  	\begin{scope}
	  		\clip (2,2) rectangle (3,3);
	  		\draw[dashed] (2.5, 2) circle(0.5);
	  	\end{scope}
	  	\begin{scope}
	  		\clip (-2,-1) rectangle (-1,0);
	  		\draw[dashed] (-1.5, 0) circle(0.5);
	  	\end{scope}
	  	\begin{scope}
	  		\clip (-1,-1) rectangle (0,0);
	  		\draw[dashed] (-0.5, 0) circle(0.5);
	  	\end{scope}
	  	\begin{scope}
	  		\clip (0,-1) rectangle (1,0);
	  		\draw[dashed] (0.5, 0) circle(0.5);
	  	\end{scope}
	  	\begin{scope}
	  		\clip (1,0) rectangle (2,1);
	  		\draw[dashed] (1.5, 1) circle(0.5);
	  	\end{scope}
	  	\begin{scope}
	  		\clip (2,-1) rectangle (3,0);
	  		\draw[dashed] (2.5, 0) circle(0.5);
	  	\end{scope}
	  	\begin{scope}
	  		\clip (3,-1) rectangle (4,0);
	  		\draw[dashed] (3.5, 0) circle(0.5);
	  	\end{scope}
	  	\begin{scope}
	  		\clip (4,0) rectangle (5,1);
	  		\draw[dashed] (4, 0.5) circle(0.5);
	  	\end{scope}
	  	\begin{scope}
	  		\clip (3,1) rectangle (4,2);
	  		\draw[dashed] (3.5, 1) circle(0.5);
	  	\end{scope}
	  	\begin{scope}
	  		\clip (3,1) rectangle (4,2);
	  		\draw[dashed] (3, 1.5) circle(0.5);
	  	\end{scope}
	  	\begin{scope}
	  		\clip (1,0) rectangle (2,1);
	  		\draw[dashed] (1, 0.5) circle(0.5);
	  	\end{scope}
	  	\begin{scope}
	  		\clip (1,0) rectangle (2,1);
	  		\draw[dashed] (2, 0.5) circle(0.5);
	  	\end{scope}
	  	\begin{scope}
	  		\clip (-1,1) rectangle (0,2);
	  		\draw[dashed] (0, 1.5) circle(0.5);
	  	\end{scope}
	  \end{tikzpicture}